\documentclass[pdflatex,sn-mathphys-num]{sn-jnl}% Math and Physical Sciences Numbered Reference Style
\usepackage{graphicx}%
\usepackage{multirow}%
\usepackage{amsmath,amssymb,amsfonts}%
\usepackage{amsthm}%
\usepackage{mathrsfs}%
\usepackage[title]{appendix}%
\usepackage{xcolor}%
\usepackage{textcomp}%
\usepackage{manyfoot}%
\usepackage{booktabs}%
\usepackage{algorithm}%
\usepackage{algorithmicx}%
\usepackage{algpseudocode}%
\usepackage{listings}%
\usepackage{setspace}%
\usepackage{isotope}%
\theoremstyle{thmstyleone}%
\newtheorem{theorem}{Theorem}%  meant for continuous numbers
\newtheorem{proposition}[theorem]{Proposition}% 

\theoremstyle{thmstyletwo}%
\newtheorem{remark}{Remark}%

\theoremstyle{thmstylethree}%
\newtheorem{definition}{Definition}%

\begin{document}

%added for ease of reviewing
%\doublespacing
\onehalfspacing

%Andrew's definitions

\newcommand{\bW}{{\bf W}}
\newcommand{\btheta}{\boldsymbol{\theta}}
\newcommand{\bullt}{\boldsymbol{\Theta}}
\newcommand{\blambda}{\boldsymbol{\lambda}}
\newcommand{\POE}{{\rm POE}}
\newcommand{\LCOEeff}{{{\rm LCOE}_{\rm eff}}}
\newcommand{\Qecon}{\rm{Q_{\rm econ}}}
\newcommand{\Ctarget}{\rm{C_{\rm target}}}
\newcommand{\Cgain}{\rm{C_{\rm gain}}}
\newcommand{\Cnet}{\rm{C_{\rm net}}}

\title[Article Title]{{\small ACCEPTED for publication in the \textbf{\textit{Journal of Fusion Energy}}} \\
\phantom{000} Criteria for the economic viability of \\[-.125in]
fusion power plants}


\author*[1,2]{\fnm{D.G.} \sur{Whyte}}\email{dwhyte@rutherfordev.com}

\author[1,3,4]{\fnm{A.} \sur{Lo}}\email{alo@rutherfordev.com}

\author[1,2]{\fnm{R.} \sur{Bielajew}}\email{rbielajew@rutherfordev.com}

\author[1]{\fnm{M.} \sur{Hancock}}\email{mhancock@rutherfordev.com}

\author[1]{\fnm{R.} \sur{Moeykens}}\email{rmoeykens@rutherfordev.com}

\author[1,2]{\fnm{G.} \sur{Shaw}}\email{gshaw@rutherfordev.com}

\affil*[1]{\orgname{Rutherford Energy Ventures}, \orgaddress{\city{Cambridge} \postcode{02139}, \state{MA}, \country{USA}}}

\affil[2]{\orgdiv{Plasma Science and Fusion Center}, \orgname{Massachusetts Institute of Technology}, \orgaddress{ \city{Cambridge}, \postcode{02139}, \state{MA}, \country{USA}}}

\affil[3]{\orgdiv{Sloan School of Management, CSAIL, EECS, ORC, and Laboratory for Financial Engineering}, \orgname{Massachusetts Institute of Technology}, \orgaddress{ \city{Cambridge}, \postcode{02139}, \state{MA}, \country{USA}}}

\affil*[4]{\orgname{Santa Fe Institute}, \orgaddress{\city{Santa Fe} \postcode{87501}, \state{NM}, \country{USA}}}


\abstract{Commercial fusion energy requires frameworks to assess both the scientific and economic viability of a wide variety of fusion concepts. Inspired by the Lawson criterion's ability to universally describe fusion energy gain, a generalized framework is developed to determine the economic gain of fusion power plants. The model exploits temporal equilibrium, and engineering and cost parameters normalized to the energy capture surface. The derived criteria for economic gain are therefore independent of the power plant's absolute power, impartial to the particulars of its fusion technology, and can be applied to any fusion confinement concept. The derivation of the economic gain factor, $\mathrm{Q_{econ}}$, results in nonlinear equations with ten controlling normalized design parameters ranging from fusion power density and surface component lifetime to energy fluence, price of energy, and component efficiency and cost. These ten controlling parameters are varied over a wide range to provide high-level insights in design, finance and operational tradeoffs that improve the prospects for economically viable fusion energy.}




\maketitle

%-------
%\newpage
%\tableofcontents
%\newpage
%-------

\section{Introduction}\label{sec:introduction}

Following decades of fundamental research, fusion development is pivoting from a purely research endeavor to practical energy applications. The need for a dispatchable, sustainable, carbon-free energy source with high power density is growing due to evolving requirements for energy security and environmental health, as well as emergent energy-intensive sectors like artificial intelligence.  In response to this demand, a host of fusion development companies have formed over the last decade, aiming to develop fusion energy systems using a wide variety of approaches \cite{FIA2024}.  Recent advances in technology and computing are improving the prospects for commercializing fusion. Two recent examples are the indirect laser drive fusion implosion with net energy gain and significant self-heating from the fusion products \cite{abu2024achievement} and the demonstration of high-temperature superconductor fusion magnets at very high magnetic field \cite{hartwig2023sparc}. This naturally poses the question: how might we evaluate these accomplishments in terms of economic and commercial success, rather than merely scientific success?  
The principal gauge of scientific success in fusion as a practical energy source is the fusion plasma energy gain, $\mathrm{Q_p}$ the ratio of fusion energy produced to external energy applied to the plasma fuel. This is determined by the ``Lawson criterion,'' first derived in 1957 by John Lawson. \cite{JDLawson_1957}
The framework of the Lawson criterion only requires knowledge of the binary reaction rates for nuclear fusion and electron-ion continuum radiation processes.
Its framework makes several key assumptions to reach this simplified form:

\begin{itemize}

    \item The power balance of the fusion plasma is solved volumetrically, so that all rates of power gain and loss are normalized per unit volume.

    \item Thermonuclear fusion rates are set by the plasma density and temperature, which solely sets the reactivity,  producing self-heating of the plasma volume from the charged reaction products.  Continuum radiation losses are set by electron-ion binary collisions in a thermal plasma. Only the fusion fuel species must be prescribed, which will determine the reactivity, reaction energy gain, effective ion charge, and fusion products.
    
    \item  Power balance is determined from fusion fuel parameters, only requiring the density ($\text{n}$) and temperature (T), and an energy confinement time ($\mathrm{\tau_E \equiv 3knT/p_{heat} }$) defined from those parameters and the volumetric heating power density $\mathrm{p_{heat}}$. It does not require specifying the physical mechanisms governing confinement.

    \item The plasma fuel is in temporal equilibrium ($\mathrm{\delta / \delta t \rightarrow 0}$), or more precisely, temporal variations of plasma parameters that occur on sufficiently short timescales can be ignored if they don't affect the power balance.
      
\end{itemize}

Its simplicity and transparency make the Lawson criterion universal and profound.   Lawson provides the crucial insight that the temperature acts as an independent variable. The $\mathrm{n \cdot \tau_E}$ product determines the energy gain at a given T, not separately from each other, which leads to the staggering ranges of density and confinement times ($\mathrm{\approx 10^{10}}$) used across confinement methods.
Lawson does \emph{not} require any detailed knowledge of the plasma, such as its stability or turbulence, nor of the confinement method or its parameters (e.g., magnetic field, laser energy, etc.), making it impartial with regard to technology.
A recent evaluation  \cite{wurzel2022progress}  of fusion concepts provides details on the Lawson criterion for continuous designs (e.g., magnetic fusion) and pulsed designs (e.g., inertial and magneto-inertial fusion).

Taking inspiration from the Lawson criterion, we seek to develop a general framework that allows the economic evaluation of a mature fusion power plant (FPP).  Our simplifying assumptions are as follows:

\begin{itemize}

    \item We solve for economic gains and costs at a control surface surrounding the fusion fuel, exploiting the fact that fusion must occur in an isolated volume, and that all fusion energy must be extracted through this control surface.

    \item Gain and loss rates are normalized to the control surface, set by the power density, fusion energy fluence, plant and component financing cost, and other engineering design parameters. The framework is thus independent of absolute power production, and is impartial to fusion fuel cycle and confinement method.
    
    \item The economic gain and loss rates are in temporal equilibrium over the lifetime of the FPP. The power plant is assumed to have two periods: one for power production and one for component replacement at the control surface. Thus, the FPP lifetime is assumed to be significantly longer than these periods.

    \item The temporal equilibrium provides an economic gain factor, $\mathrm{Q_{econ}}$, which is a constant ratio of economic gain to costs over the lifetime of the FPP.  A fundamental requirement of a commercially viable FPP design is $\mathrm{Q_{econ}} \ge 1$. Meeting this threshold is ``necessary, but insufficient'' to meet real-world commercial viability because, by definition, most but not all costs can be included in the model. However, without $\mathrm{Q_{econ}} \ge 1$ there is no prospect of net returns.  
      
\end{itemize}

\section{Derivation of the framework and criteria for fusion economics}\label{sec:derivation}

\subsection{Normalization definitions for the framework}

The starting point of the framework for the economic evaluation of fusion is the realization that all fusion concepts have a control surface, S, measured in $\mathrm{m^2}$, through which all fusion energy must be extracted. This requirement stems from the Lawson criterion, which states that fusion reactions and/or energy gain occur at thermonuclear temperatures; therefore, the volume for fusion reactions must be completely isolated from engineered, terrestrial objects. We further note that we are using the broader definition of the Lawson criterion, which determines the magnitude of
$\mathrm{Q_p}$, and the model does not require that the plasma is ignited, i.e., $\mathrm{Q_p \rightarrow \infty}$.
 The size and specific shape of surface S is not required, only that the surface is directly participating in the areal removal of the volumetric fusion power $\mathrm{P_f}$ in megawatts (MW). 

\begin{figure}[h]
    \centering
    \includegraphics[width=0.95\textwidth]{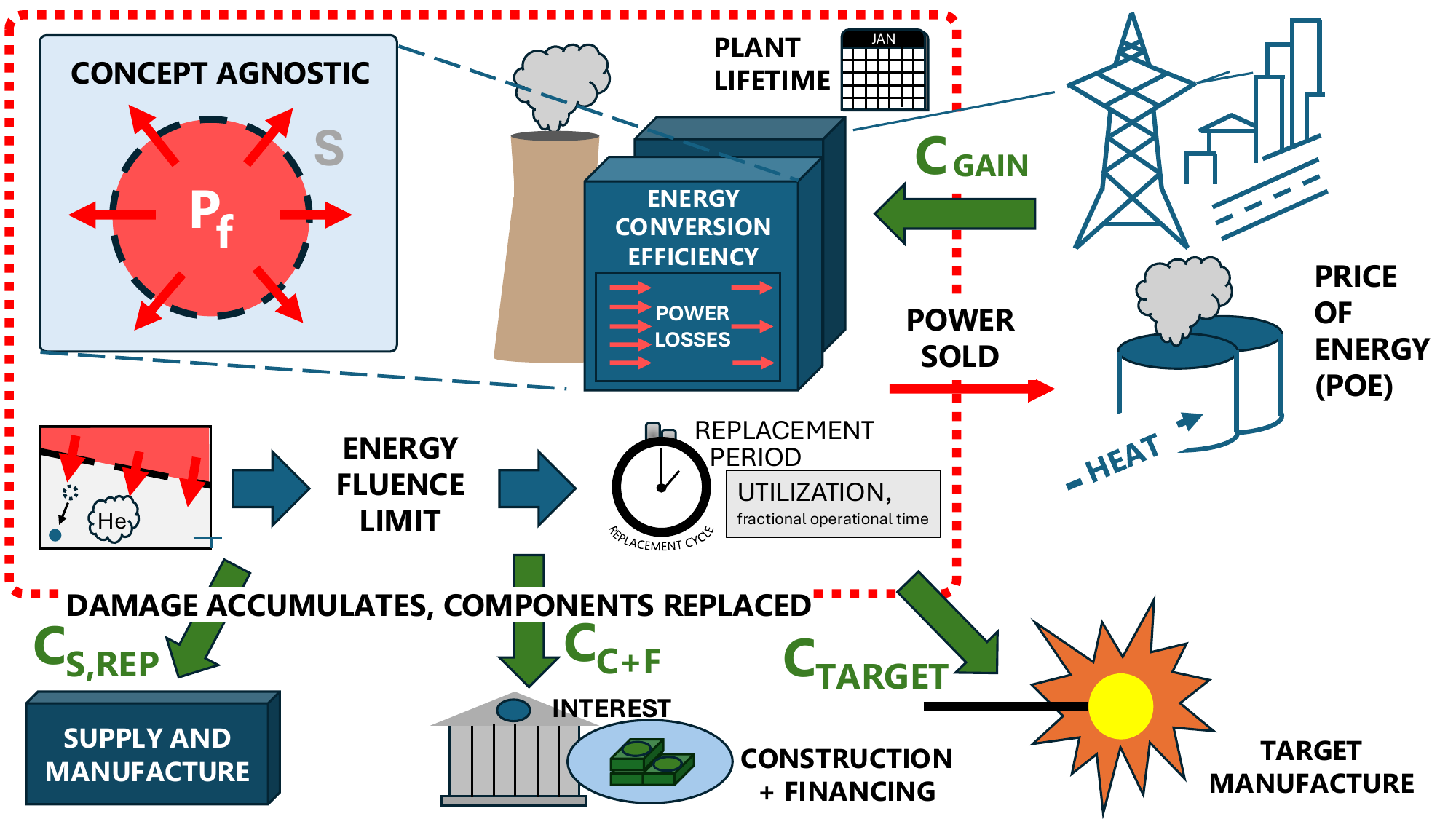}
    \caption{Graphical representation of the FPP economic framework. Red arrows represent energy flows normalized to S, the energy extracting surface surrounding the fusion volume producing fusion power $\mathrm{P_f}$.   Green arrows represent monetary rates, both gains and costs, normalized to S.}
    \label{fig:graphical_framework}
\end{figure}

Figure \ref{fig:graphical_framework} provides a graphical representation of the economic framework, with this section providing definitions and derivations of the controlling rate equations.
A fusion power plant (FPP) produces a time-averaged areal fusion power $\mathrm{P_f/S}$ in $\mathrm{MW /m^2}$ during the period of power operations.  Our model framework links this normalized fusion power/energy output to a normalized economic gain and cost rate.  The units of monetary gain/cost used are in M\$ (US million dollars) per calendar year (y), so that the normalized economic gain and cost rates through S are derived in units of $\mathrm{M\$/m^2 \text - y}$. This calendar year is not the same as an operational year, since the model will intrinsically incorporate the economic impact of maintenance periods when fusion power is not being produced. This unit for economic rate is convenient, providing results of order unity from typical input parameters. As much as possible, costing and performance parameters are represented in the same units typically used by fusion researchers and the energy industry, with appropriate unit conversions applied so that all economic rates are in $\mathrm{M\$/m^2 \text - y}$. Stated economic values are in 2025 US dollars; however, the primary output of the model, the FPP economic gain $\mathrm{Q_{econ}}$, is inherently inflation-independent because it is a ratio of economic rates.

\subsection{Temporal equilibrium}\label{sec:temporal equilibrium}

As with the Lawson criterion, temporal equilibrium is imposed for the economic framework so that gain and cost rates are constant in time over the FPP lifetime. To achieve constant annual cost estimates over the lifetime of the plant, we will apply an amortization formula with an averaged real interest rate. While this is a simplification of real life financing scenarios, its purpose is to transparently capture the impact of the cost of capital on economic viability.

We impose that the FPP has two periods that are repetitive and sequential throughout the lifetime of the FPP, denoted as $\mathrm{\tau_{life}}$ in years [y].  In one period $\mathrm{\tau_{op}}$ [y] the FPP is in fusion-producing operation with a fixed $\mathrm{P_f/S >0}$. The extraction of energy through S leads to sufficient degradation that it is no longer operational, forcing a second period $\mathrm{\tau_{rep}}$ [y], in which the FPP is producing zero fusion power, while the control surface S is being replaced or refurbished. By definition $\mathrm{\tau_{rep}}$ must include any associated decommissioning or recommissioning time, since it must capture the duration of when the FPP is not producing a commercial product.   The FPP therefore has a cycle period of 
\begin{equation}\label{eq:tau_cycle_definition}
    \tau_{cycle} = \tau_{op} + \tau_{rep} \hspace{3mm}. \vspace{2mm}
\end{equation}
Thus the temporal equilibrium assumption is simply averaging economic gains and losses over multiple cycles, and the mathematic requirement is $\mathrm{\tau_{life} \gg \tau_{cycle}}$, which is reasonable since $\mathrm{\tau_{life}}$ is typically several decades. The model is insensitive to variations of  $\mathrm{P_f/S >0}$ at timescales significantly less than $\mathrm{\tau_{op}}$. Therefore, the model is applicable to both steady-state fusion power concepts and pulsed fusion systems, where the longest pulse timescales are $\mathrm{\sim 10^3 s \ll \tau_{cycle} }$ in inductive magnetic confinement.  

The fractional calendar time that the FPP is in operation (and thus producing energy to sell) is called the utilization factor U, which is given by
\begin{equation}\label{eq:U_definition}
    U = \frac{\tau_{op}}{\tau_{cycle}} =\frac{\tau_{op}}{\tau_{op} + \tau_{rep}} \hspace{3mm}. \vspace{2mm}
\end{equation}
The utilization is connected to the capacity factor, which is typically used in energy systems to denote the fractional operational period. The distinction here is that the capacity factor has an upper limit of the utilization factor, since in practical energy systems, there are other limitations on capacity, including climate intermittency (e.g., in renewables), market pricing and demand, and mechanical maintenance. These latter considerations are ignored in the present model, with the assumption that FPPs will operate until the components in S fail due to fundamental physical limits arising from the energy throughput from the fusion reaction.  The framework mathematically allows utilization to be at unity to accommodate FPP designs where S is continually replaced so that $\mathrm{\tau_{rep} \rightarrow0}$.

\subsection{Economic gain rate from selling energy}\label{sec:economic gain}

The first term we will derive within our framework is the economic gain rate from selling the energy product generated by the FPP. 
Schematically, the fusion power passes through the surface area S into an engineered blanket volume, which converts the kinetic energy of the fusion products into useful energy to be sold. We assign a conversion efficiency $\eta_E$ to the FPP
\begin{equation}\label{eq:eta_definition}
    \eta_E \equiv \frac{P_{net}}{P_f} \hspace{3mm}. \vspace{2mm}
\end{equation}
Here, $\mathrm{P_{net}}$ is the surface-averaged, time-averaged net power output. It is important to note that $\mathrm{\eta_E}$ is not a single parameter (e.g., the thermal conversion efficiency) but rather, it is a systems-wide accounting of net power (or net energy) production, and therefore it must include all aspects of internal power conversion, plasma energy gain, recirculating power, and so on.  Neither is it limited to electricity as the sole energy product of the FPP, which may also include industrial heat or fuel production.  For illustrative purposes, we include a sample derivation of $\eta_E$ for an electricity-producing FPP in Appendix \ref{sec:etae_appendix}.

The received net income for the FPP energy product is measured in the energy industry standard unit, the  ``price of energy'' in $\mathrm{\$ /MW \text - h}$. For clarity, we define
\begin{equation}\label{eq:POE_definition}
    POE_{net} \space [\$ /MW \text - h] = POE_{price} - POE_{O\&M}, \vspace{2mm}
\end{equation}
where $\mathrm{POE_{net}}$ is primarily set by $\mathrm{POE_{price}}$, the time-averaged price that the FPP energy product is sold at during FPP operations. While $\mathrm{POE_{price}}$ is a constant in the model, it can reflect the impact of local market conditions and variable energy pricing (e.g. peak power prices) as long as these are incorporated into a model that time-averages over $\mathrm{\tau_{op}}$. $\mathrm{POE_{price}}$ should use the output weighted average of all the products sold by the FPP, including energy market products (e.g. electricity, heat, fuel) and non-energy market products whose price can be linked to net energy output (e.g. transmutation products \cite{rutkowski2025scalablechrysopoeian2n},  desalinated water).
$\mathrm{P_{O\&M}}$ reflects the FPP variable O\&M ``Operations \& Maintenance" costs, in $\mathrm{\$ /MW \text - h}$, which are tied to net energy output \cite{larsen2023nuclear}, and so are typically driven by fuel costs, coolant consumption, etc. In this economic framework, the cost of the ``raw'' fusion fuel is effectively zero, since fusion target costs are treated separately in Sec. \ref{sec:target_cost}. For most FPP concepts one would expect $\mathrm{P_{O\&M} \ll P_{price} }$ but this will be design dependent. $\mathrm{POE_{net}}$ is the controlling parameter in the model. Fixed O\&M costs are considered in Sec. \ref{sec:fixed_costs}.

Taking into account the utilization factor U, then the economic gain rate, $\mathrm{C_{gain}}$ [$\mathrm{M\$/m^2 \text -  y}$], is given by
\begin{equation}\label{eq:Cgain_definition}
\begin{split}
    C_{gain} [M\$/m^2 \text -  y] = \frac{M\$}{10^{6}\$} \frac{8760 h}{y}\enspace POE_{net} \enspace(P_f/S) \cdot U \cdot \eta_E  \\ = 8.76\times10^{-3}POE_{net} \enspace(P_f/S) \cdot U \cdot \eta_E
    \hspace{3mm}, \vspace{2mm}
    \end{split}
\end{equation}
where the final term provides the consolidated constant. 

\subsection{Utilization factor and S energy fluence limit}\label{sec:utilization_factor}

As defined, the utilization factor U arises from the physical limits of the surface S due to its degradation, which forces the replacement of S (and its associated components) during $\mathrm{\tau_{rep}}$. The degradation is quantified as an energy fluence limit of S by $\mathrm{X_S}$ in units of $\mathrm{MW \text -  y/m^2}$. By definition, once this limit is reached, S is no longer operational, fusion power operations must cease, and S must be replaced during $\mathrm{\tau_{rep}}$.  Conceptually, this takes advantage of the fact that, regardless of the nature of energy removal, every joule of fusion energy must first be generated in the plasma volume and then pass through S. Thus, we link the energy fluence limit to the total amount of fusion power, such that the operational period is  defined by
\begin{equation}\label{eq:tau_op_solved}
\tau_{op} [y] = \frac{X_S}{P_f/S} \hspace{3mm}. \vspace{2mm}
\end{equation}
The choice of linking energy fluence to the durability of S arises from the  nature of fusion energy and the physical realities of S, which must be an engineered component in the solid or liquid phase.

In general, the fusion energy is transmitted to S in discrete high-energy forms that will force S to degrade. First, the primary particles produced by fusion reactions will have kinetic energies $\mathrm{>}$ MeV. Therefore, the direct removal of these particles, be they charged particles or uncharged neutrons, must necessarily lead to significant damage or perturbation at the atomic level in S, given that interatomic potentials in S will be $\mathrm{\sim}$10-100 eV.  The outgoing fusion products must undergo collisions in S (whether coulombic or neutronic) which will disorder the atoms in S, leading to cumulative degradation. That damage level is linked to the areal energy transfer for a given distribution of fusion products and composition of S. Furthermore, products with energies $\mathrm{>}$ MeV have the possibility to energetically engage in nuclear reactions with the isotopes composing S, which would result in a cumulative effect on S, also linked to the energy fluence for a given neutron spectrum. 

In addition to high kinetic energy fusion products, one must also consider the nature of plasma energy removal. Sufficient fusion reactivity and energy gain is accessed at plasma temperatures/energies $\mathrm{\sim>}$ 1-100 keV.  Plasmas will transmit their energy into S via charged particles, neutral particles or photons, all of which are likely to contribute to the accumulated damage and degradation of S.  Charged particles and ions will accelerate towards S due to ambipolar potentials with energies $\mathrm{\sim 5 \times}$ the local plasma temperature, typically leading to atomic surface damage in the form of sputtering, since surface binding energies are $\mathrm{\sim}$ 1-5 eV. Similarly, neutral particles may arrive at S above this sputter threshold, due to charge-exchange between neutral species from incident fuel  and hot ions in the plasma, or from the expulsion of  non-ionized fuel in pulsed fusion concepts. The bulk of photons arriving at S will reflect the characteristic plasma temperature/energy, which will be $\mathrm{>}$ 1-100 keV, since  they arise from coulombic/inelastic collisions of free/bound electrons in the plasma. These photons will cause damage through ionization and displacement processes, since their energies surpass the typical thresholds for these effects.

Another consideration arises because the primary product of practical fusion reactions for energy production is helium, due to its high nuclear binding energy.  Because it is an inert element, helium accumulation poses particular challenges in solid components, due to its insolubility.  Since, for a given fusion reaction, the total helium production is directly linked to fusion energy, one may expect further cumulative damage caused by helium linked to energy fluence.

These facts justify the generic use of energy fluence through S as a determinant of its operational lifetime given by Equation \ref{eq:tau_op_solved}, since it appears inevitable that one of the energy removal mechanisms described above will occur at a sufficient rate to limit S.   However, the actual limit of $\mathrm{X_S}$ must be determined from the specific details of the FPP and the design of S.  The use of $\mathrm{X_S}$ fits with the stated goal of an economic model impartial to fusion approaches; the FPP designer, however, will be obligated to determine $\mathrm{X_S}$ based on the design specifics. This will depend on a wide variety of physical and engineering parameters, including the fusion fuel cycle products, the spectrum and flux of primary fusion products into S, the plasma temperature/energy, and the plasma energy loss mechanisms.  

An estimate of $\mathrm{X_S}$ across the entirety of proposed fusion designs is outside the scope of this work. Nevertheless, it is useful to provide an approximation of $\mathrm{X_S}$ for D-T fusion, the most common choice for proposed commercial FPPs due its high reactivity and power density at its given plasma conditions. In D-T fusion, 80 \% of the fusion power exhaust arrives as free-streaming 14.1 MeV neutrons at S, which will likely determine $\mathrm{X_S}$ due to the cumulative material displacement and nuclear transmutation levels in S.  

Appendix \ref{sec:Xs_appendix} provides a derivation of $\mathrm{X_S}$ that illustrates a simplified but robust link between $\mathrm{X_S}$ and the displacement per atom (dpa), a commonly used figure of merit for neutron energy fluence. The simplicity of this example arises from the general nature of collision kinematics for the highly penetrating neutrons, which pass through S and thermalize. Besides its relevance to commercial FPPs, the D-T system provides a relatively straightforward way to link degradation to energy fluence, since high energy neutrons and weak neutron-matter interactions (the mean free path is $\mathrm{\sim 0.1 m}$ in solids and liquids) lead to unavoidable volumetric damage and/or heating in S at the atomic level.  

It must be noted that $\mathrm{X_S}$ represents the most optimistic limit for the lifetime of S. There are clearly other failure modes for the components that make up S, for example, excursions in peak power density that pass the local limits of actively cooled components.  Another potential failure mechanism is thermal fatigue in components of S that undergo thermal cycling.  Regardless of the details of the dominant degradation mechanism, however, $\mathrm{X_S}$ represents the ultimate S limit, because it is linked to the energy fluence which is necessary for producing energy that the FPP can sell.

Having discussed the individual parameters that make up the operational period, we insert its formula into Equation \ref{eq:U_definition} and rearrange the utilization factor U:
\begin{equation}\label{eq:U_rearranged_definition}
    U =  [1 + \frac{\tau_{rep}\cdot (P_f/S)} {X_S}]^{-1} \hspace{3mm}. \vspace{2mm}
\end{equation}
To see how this affects the economic framework we can insert the solution for U into  Equation \ref{eq:Cgain_definition}, providing the full economic gain rate:
\begin{equation}\label{eq:Cgain_full}
    C_{gain}    = 
    8.76\times10^{-3}POE_{net}  \cdot \eta_E \cdot
    \enspace(P_f/S) \cdot [1 + \frac{\tau_{rep}\cdot (P_f/S)} {X_S}]^{-1} \hspace{3mm}. \vspace{2mm}
\end{equation}

The replacement time $\mathrm{\tau_{rep}}$ for S remains an independent parameter, since it has no link to fusion power density, but is set by the integrated FPP and S design.  Equations \ref{eq:tau_op_solved} - \ref{eq:Cgain_full} indicate the utility of using the power (or energy) normalized per surface area; it provides a simple relationship to the normalized cumulative damage in S, which in turn links energy transmission through S to FPP utilization.  It is noted that this framework accommodates FPP designs where S is continually replaced, such as the flowing molten salt blanket in the Hylife design  \cite{moir1994hylife} or the liquid lead-lithium blanket proposed by General Fusion \cite{laberge2019magnetized}, in which case $X_S$ can be set arbitrarily large to force $\mathrm{U \rightarrow 1}$.  The economic gain rate is insensitive to absolute power output, as desired for a general framework.

\subsection{Fusion target cost rate}\label{sec:target_cost}

Multiple fusion concepts require the fabrication of discrete fuel targets which are consumed by the process of achieving fusion energy and gain.  This is distinguished from the consumable fuel elements used in D-T fusion, deuterium and lithium, which are assumed in this present model to incur no cost due to the extraordinarily high energy gain achieved per unit mass---indeed, this is the allure of fusion energy! The target costs are treated separately from the replacement costs for S (next section), because while obviously linked to fusion power, their design is not directly linked to the lifetime of the components in S. The cost rate for the targets is defined as:
\begin{equation}\label{eq:C_target_def}
   C_{target}[M\$/m^2 \text - y]=(\frac{M\$}{target})\cdot (\frac{N_{target}}{S \cdot year}) \hspace{3mm}. \vspace{2mm}
\end{equation}
The framework does not specify the technical requirements of the target, simply that it must reflect the integrated costs associated with the discrete engineered objects (e.g., fabrication, delivery, and removal) which are required to achieve the fusion power $\mathrm{P_f}$ and are consumed at a known rate.  For example, laser-driven inertial fusion includes a spherical target and its delivery assembly, while indirect drive fusion would include the target's hohlraum, which are fully consumed with each fusion event.  In pulsed magnetic fusion, this includes the fuel target and any electrodes or wires which are intentionally consumed in a pulse or a finite number of pulses.  In magnetized target fusion, this would include the cost of developing any consumable engineered material involved in the formation of the plasma for compression and/or the cost of mechanical/electrical components that are limited to a finite number of compression events. In current pinches, this could include the cost of the electrodes, which could have a lifetime significantly shorter than S due to mechanical degradation from repeated pulses. (In this case, the ``target'' is used over many fusion events.)  In magnetic fusion, this could include complex fuel pellets required to achieve fusion performance.  

The framework uses the same target costing convention of $\mathrm{c_{target}}$ in  $\mathrm{[\$/target]}$ from laser inertial fusion energy (IFE). Each target use is assumed to produce a fixed amount of fusion energy yield, $\mathrm{Y_{target}}$ in MJ, before it is consumed. The fusion power density is given by the use rate of the targets $\mathrm{f_{target}}$ [Hz]:
\begin{equation}\label{eq:Pf_and_yield_def}
   P_f/S=\frac{Y_{target}\cdot f_{target}}{S} \vspace{2mm}
\end{equation}
or 
\begin{equation}\label{eq:ftarget_def}
   \frac{f_{target}}{S}=\frac{P_f/S}{Y_{target}} \hspace{3mm}. \vspace{2mm}
\end{equation}
In a calendar year, the number of targets consumed is:
\begin{equation}\label{eq:Ntarget_year_def}
   \frac{N_{target}}{year} = f_{target} \frac{3.15\times10^7 s}{year}\cdot U \hspace{3mm}, \vspace{2mm}
\end{equation}
and 
\begin{equation}\label{eq:Ntarget_year_S_def}
   \frac{N_{target}}{S \cdot y} = \frac{f_{target}}{S} \frac{3.15\times10^7 s}{y}\cdot [1 + \frac{\tau_{rep}}{\tau_{op}}]^{-1} \enspace \hspace{3mm}. \vspace{2mm}
\end{equation}

Substituting Equation \ref{eq:ftarget_def} into the first term on the right-hand side of Equation \ref{eq:Ntarget_year_S_def}, and using Equation \ref{eq:C_target_def} with unit conversion to $\mathrm{c_{target}}$, we obtain a complete expression for the target cost rate:
\begin{equation}\label{eq:C_target_full}
\begin{split}
   C_{target}[M\$/m^2 \text -  y]=  10^{-6} \enspace c_{target} \cdot \frac{(P_f/S)} {Y_{target}} \cdot 3.15\times10^7 
   \cdot [1 + \frac{\tau_{rep}}{\tau_{op}}]^{-1} \\
   =31.5 \enspace (\frac{c_{target}}{Y_{target}}) \cdot (P_f/S) \cdot [1 + \frac{\tau_{rep}\cdot (P_f/S)} {X_S}]^{-1} \hspace{3mm}. \vspace{2mm}
   \end{split}
\end{equation}
Since $\mathrm{c_{target}}$ and $\mathrm{Y_{target}}$ only appear in the fusion target cost rate we make one further definition to minimize the number of controlling parameters,
\begin{equation}\label{eq:cY_definition}
   (c/Y)_{target} \enspace [\$ / MJ] \equiv \frac{c_{target}}{Y_{target}} \hspace{3mm},
\end{equation}\vspace{1mm}
defined as the normalized target cost per energy yield, providing 
\begin{equation}\label{eq:C_target_full_collapsed}
   C_{target}[M\$/m^2 \text -  y]= 31.5 \enspace (c/Y)_{target} \cdot (P_f/S) \cdot [1 + \frac{\tau_{rep}\cdot (P_f/S)} {X_S}]^{-1} \hspace{3mm}, \vspace{2mm}
\end{equation}
where the full definition of utilization has been used.

\subsection{Economic cost rate for replacing the control surface S}\label{sec:S_replacement_cost}

The replacement of S, and its associated components, incurs a rate of economic cost for the FPP, $\mathrm{C_{S,rep}}$. Consistent with the framework, this cost is normalized to the surface area of S, given by the parameter $\mathrm{(M\$/S)_S}$  in $\mathrm{M\$/m^2}$.  This cost must include the entirety of the costs required to remove and replace S each time, including but not limited to its fabrication, off-site qualification, installation and disposal.

The use of a surface area S may seem contradictory with engineering reality since components, which are characterized by a volume or mass, are what will be actually replaced and paid for, not a surface. 
As detailed in Sec. \ref{sec:utilization_factor} the replacement will be required because S, and associated components ``behind'' S (heat removal systems, mechanical attachments and interfaces, etc.), will reach their energy fluence limit. This limit depends on the details of the FPP design and the nature of the energy fluence (particles, S geometry, energy spectrum) through S, and therefore the definition of what volume specifically must be replaced is design specific. Therefore while S cannot have a specific definition in the model framework, it is a universally applicable normalization for any FPP fuel cycle, confinement concept or S design. Physical reasoning requires that all the energy pass through S, and it this energy, changed in form, that will constitute the economic product of the FPP.

Given the framework assumption that the FPP has a lifetime much longer than the cycle lifetime, $\mathrm{\tau_{cycle}}$ we will neglect the financing costs for the replacement of S.  The full average replacement period for S is $\mathrm{\tau_{cycle}}$, and therefore the S replacement cost rate is:
\begin{equation}\label{eq:C_S_Rep_definition}
    C_{S,rep} [M\$/m^2 \text -  y]  = \frac{(M\$/S)_S }{\tau_{cycle}}= 
    \frac{(M\$/S)_S }{\tau_{op}+\tau_{rep}}  \hspace{3mm}, \vspace{2mm}
\end{equation}
with this derivation following the definition of Equation \ref{eq:tau_cycle_definition}.  Substituting Equation \ref{eq:tau_op_solved} results in:
\begin{equation}\label{eq:C_S_Rep_full}
    C_{S,rep}  = (M\$/S)_S \cdot
    [\frac{X_S}{P_f/S}+\tau_{rep}]^{-1} \hspace{5mm} \hspace{3mm}. \vspace{2mm}
\end{equation}
This can be rearranged to follow the forms used for the gain and target cost rate:
\begin{equation}\label{eq:C_S_Rep_U_full}
\begin{split}
    C_{S,rep}  = (M\$/S)_S \cdot \frac{P_f/S}{X_S} \cdot
    [1 + \frac{\tau_{rep}\cdot (P_f/S)} {X_S}]^{-1} \hspace{5mm} \hspace{3mm}. \vspace{2mm}
\end{split}
\end{equation}

\subsection{FPP construction, financing and operation fixed cost rate}\label{sec:fixed_costs}

The construction, financing and operations of a FPP will incur cost rates which are fixed over its lifetime.  In a commercial FPP, as with any commercial power plant, there will be a fixed cost associated with paying the principal and interest of the funds borrowed for construction and delivery.  In addition, there will be fixed O\&M (``operations and maintenance'') costs per year associated with FPP staffing, compliance and routine maintenance outside of the replacement of S. There are many assumptions that have to be made to assess this cost realistically, including depreciation, fixed versus variable interest rates, etc., which are beyond the scope of this framework due to its temporal equilibrium requirement. For transparency and simplicity we use the  amortization formula for the fixed rate of payments based on the FPP cost constant $\mathrm{(M\$/S)_{FPP}}$ in units of $\mathrm{M\$/m^2}$ normalized to the surface area S, 
\begin{equation}\label{eq:C_CF_definition}
    C_{fixed} [M\$/m^2 \text -  y]  = (M\$/S)_{FPP} 
    \frac{0.01 \cdot i(1+0.01 \cdot i)^{\tau_{life}}}{(1+0.01 \cdot i)^{\tau_{life}}-1} \hspace{3mm}, \vspace{2mm}
\end{equation}
which uses the standard amortization formula, with i in $\mathrm{\%}$ being the real interest rate, with the principal and interest assumed to be paid over the lifetime of the FPP $\mathrm{\tau_{life}}$ in years.  By using the real interest rate, i.e. inflationary rates are subtracted, the framework intrinsically incorporates the time-varying cost of capital, even through the framework is in temporal equilibrium.   

The normalized FPP cost constant is composed of two components: ``construction and delivery''(C+D) and O\&M, namely 
\begin{equation}\label{eq:normalized_FPP_cost_definition}
    (M\$/S)_{FPP} \hspace{2mm} [M\$/m^2] =(M\$/S)_{C+D} \hspace{2mm}+ \hspace{2mm} (M\$/S)_{O\&M}
    \hspace{3mm}. \vspace{2mm}
\end{equation}
$(M\$/S)_{C+D}$ is the full cost required to construct and deliver an operational FPP, normalized to $S$ of the FPP, including financing costs during construction. Thus $(M\$/S)_{C+D}$ is related to, but is larger than, normalized FPP overnight costs.  The second term $(M\$/S)_{O\&M}$ is associated with the fixed O\&M costs. Fixed O\&M costs are typically determined as a fractional cost linked to an engineering feature of a power plant, such as peak power output (see for example from fission \cite{larsen2023nuclear}) or total overnight cost (see \cite{najmabadi2006aries, Overview_ARIES-RS} for fusion design examples.) Because the C+D and O\&M costs are additive in determining $\mathrm{(M\$/S)_{FPP}}$, mathematically the O\&M would incur the financing costs associated with amortization, which is not usually used in economic forecasting. However, if the annual fixed O\&M cost is accurately known, e.g. for a Nth-of-a-kind (NOAK) FPP, this can easily be remedied by using Eq. \ref{eq:C_CF_definition} to determine the required $(M\$/S)_{O\&M}$ with knowledge of i and $\mathrm{\tau_{life}}$. This will provide a fully accurate fixed annual cost rate since the amortization formula is linear in $(M\$/S)_{FPP}$ and $(M\$/S)_{O\&M}$. Alternatively, if the fixed O\&M cost is not well known the model user can simply assign a fixed O\&M cost as a fractional increment to the C+D cost, and accept that the modest mathematical inaccuracy introduced by the amortization will reflect the uncertainty in forecasting the O\&M cost; which will likely be appropriate for a First-of-a-kind (FOAK) FPP.

A discussion regarding the S normalization is warranted. The units used in the areal cost of $\mathrm{(M\$/S)_{FPP}}$ are related to the fact that the FPP will require construction and finance costs which must be paid back by eventually selling energy. It is thus reasonable to normalize these costs to a physical feature of the FPP.  In this framework, the cost is normalized to the surface area surrounding the fusion plasma. The actual cost may have  other dependencies, for example, in magnetic fusion energy the volume of the magnets, or in inertial fusion energy the costs of the lasers, but it is not unreasonable to expect some correlation (even if not exactly  linear) between S and the total cost, since S and the surrounding engineered objects are a prime driver in the cost of a FPP. This costing relationship is examined in Section \ref{sec:Areal_cost_appendix}, which finds there is reasonable agreement with the premise that cost varies as S. Therefore, while $\mathrm{(M\$/S)_{FPP}}$ may not be a fixed number for a particular FPP concept, the total cost should increase monotonically with S.

Therefore, $\mathrm{C_{fixed}}$ will depend on details such as overnight costs, FPP size, construction time, financing rates, and operator costs.  There is no single formula that can capture the $\mathrm{(M\$/S)_{FPP}}$ value for an FPP design.  Rather the onus is placed on the FPP developer to provide a determination of $\mathrm{(M\$/S)_{FPP}}$. In cases of a highly mature NOAK FPP, it is desirable that $\mathrm{(M\$/S)_{FPP}}$ asymptote close to the overnight costs, which will require additional considerations such as  minimum construction time, time effective FPP siting and licensing, and small fractional O\&M costs relative to overnight costs. From the model's viewpoint the imperative is on the model user that Eq. \ref{eq:C_CF_definition} reflects the most accurate determination of fixed costs during the FPP lifetime.  While an accurate determination of $\mathrm{(M\$/S)_{FPP}}$ involves a complex technical and financial analysis of each FPP, mathematically the three controlling parameters in Eq. \ref{eq:C_CF_definition} are sufficient to capture the linear ($\mathrm{(M\$/S)_{FPP}}$), and non-linear ( i and $\mathrm{\tau_{life}}$ ) dependencies on fixed costs in the model framework.  

It is useful to extract another standard metric for the NOAK power plant costing, namely the overnight cost per net power, $\mathrm{c_{O/N}}$ in $\mathrm{\$/MW}$.  Following the discussion above this can be estimated by
\begin{equation}\label{eq:covernight_def}
    c_{O/N} [\$/W] = \frac{(M\$/S)_{FPP}}{(P_f/S)\cdot\eta_E} \hspace{3mm}. \vspace{2mm}
\end{equation}
 One might consider using this definition to define the $\mathrm{(M\$/S)_{FPP}}$ term in Equation \ref{eq:C_CF_definition}. However, this would be misleading, since the fabrication cost would then scale linearly as $\mathrm{P_f/S}$, leading to a nonsensical solution where the FPP fabrication costs would approach zero as the fusion power output decreased to zero, while $\mathrm{c_{O/N} \rightarrow \infty }$. The reality is that all fusion FPPs require the manufacturing and assembly of technologically complex devices before any fusion power is generated, and this must incur a finite fabrication and financing cost. The proper interpretation of Equation \ref{eq:covernight_def} is that it allows one to translate the assigned value of $\mathrm{(M\$/S)_{FPP}}$ to the widely used figure of merit $\mathrm{c_{O/N}}$ in the energy industry, with a finite value of $\mathrm{P_f/S}$ either assigned or optimized for the FPP.  Thus, the $\mathrm{P_f/S}$ ratio can be viewed as both a design feature and an operational choice of an FPP.

It should be noted that the true governing cost equation is in fact Equation \ref{eq:C_CF_definition}, while Equation \ref{eq:covernight_def} is a convenient form that touches base with standard pricing practices for energy systems. The form of Equation \ref{eq:covernight_def} recognizes that some costs, such as for generation equipment, scale as the total fusion power at fixed S. However, it may be counterintuitive that costs should increase with the efficiency $\mathrm{\eta_E}$. The reasoning here is that $\mathrm{c_{O/N}}$ contains within itself critical factors such as economies of scale, generating costs, and so on. Therefore, $\mathrm{c_{O/N}}$ is not a parameter that is a fixed metric for a particular FPP approach or technology, but rather an indicator of the costing and market maturity of the FPP. Hence, $\mathrm{c_{O/N}}$ is generally used to indicate the stage of energy market penetration, with pilot projects typically having higher values of $\mathrm{c_{O/N} \sim 8-20 \enspace \$/W}$, and later deployment having values of $\mathrm{c_{O/N}\sim3-8 \enspace \$/W}$ at competitive costing. This is discussed in detail in  \cite{national2021bringingfusion} in the context of an entry level FPP.  Thus, $\mathrm{c_{O/N}}$ will be used as an output of our economic model, rather than as a constraint, to provide insight into the market stage of a FPP design.

\subsection{Examples of adopting the framework across varying FPP concepts}\label{sec:application_concepts}

The framework seeks to provide economic performance metrics impartial to the fusion concept.  The purpose of this section is to illustrate the model's adoption across a variety of FPP concepts and discuss how the controlling parameters might be determined. As evident from the wide array of controlling parameters, summarized in Sec. \ref{sec:summary_equation_controlling_parameters}, their numerical evaluation requires information on costing, fusion performance, finances, etc. particular to each design and marketplace, which is beyond our scope.  Instead, we qualitatively discuss adoption of the framework across a large distribution of FPP concepts covering magnetic, inertial and magnetized-target fusion.  The following list is not meant to be exhaustive of FPP concepts, nor is it an exercise to evaluate the viability of any particular FPP design and/or fusion company.  Generic names are given for FPP concepts, with references provided to commercial designs from specific companies or groups. D-T fusion is the default unless stated otherwise.
The discussion focuses on locating S and how related controlling parameters in the model ($\mathrm{X_S}$, $\mathrm{\tau_{rep}}$, $\mathrm{(M\$/S)_S}$, $\mathrm{(c/Y)_{target}}$ etc.) might be determined.  

\begin{itemize}
  
    \item \textbf{Flow-stabilized Z-pinch} \cite{levitt2023zap,shumlak2024fusion}
    S is comprised of the pinch-viewing internal solid surfaces.   $\mathrm{X_S}$ is set by neutron damage to either the weir wall (over which the liquid blanket flows to form the pinch cavity) or upstream electrode materials (upwards in Fig. 3 of \cite{levitt2023zap}. The LiPb liquid blanket is a lifetime component and should be included in  $\mathrm{(M\$/S)_{FPP}}$.  The electrodes may suffer degradation (e.g. erosion, corrosion). Since these are primarily used to provide the target plasma for compression (rather than energy extraction) electrode replacement costs should be captured in   $\mathrm{(c/Y)_{target}}$ with its value determined by energy yield per pulse, electrode costs and cycle limit. Any downtime to replace electrode targets would be accounted in time-averaged power, $\mathrm{P_f}$ along with pulse rate and energy yield. $\mathrm{\tau_{rep}}$ should be positively impacted by the linear geometry and absence of magnets.

    \item \textbf{High-B tokamak with liquid immersion blanket} \cite{sorbom2015arc,CreelySOFE_2025,rutherford2024manta}
    S is the inner surface of the close-fitting vacuum vessel, including the divertor.
     $\mathrm{(M\$/S)_S}$ should include the vacuum vessel,first wall solid components (including divertor), vacuum vessel support structures and launcher structures used for plasma heating. 
    The liquid blanket/breeder is a lifetime component so its costs should be added to $\mathrm{(M\$/S)_{FPP}}$. For a fixed blanket and VV design one notes the convenience of normalized cost to S, since the volume of both solid and liquid components is linearly proportional to S.  The  $\mathrm{X_S}$ will be likely set by 14.1 MeV neutron damage to the main chamber first wall or the erosion limit of the divertor.   The  $\mathrm{\tau_{rep}}$
    needs to include any required cooling/heating and energizing/de-energizing of superconductor magnets, liquid blanket emptying/filling, and mechanical disassembly/assembly of the vacuum vessel. 
   $\mathrm{\eta_E}$ includes the thermal efficiency of the liquid blanket, the wall plug efficiency of RF heating, recirculating power fraction and  any external electricity sources use for solenoid operations.

    \item \textbf{Laser-driven inertial fusion with flowing liquid wall} \cite{thomas2024hybrid,ogando2024preliminary,moir1994hylife}
    S is the interior surface of the flowing liquid wall / blanket. Since the D-T products (neutrons, charge particles) and target remnants only interact with the liquid, this comprises the energy extraction and therefore $\mathrm{X_S}$ may be very large since ``replacement'' of S would only entail the steel structures well behind the thick liquids.  This may not be the case if refractory final optics are used which would presumably have finite energy fluence and would set $\mathrm{X_S}$ and a $\mathrm{\tau_{rep}}$ determined by optics replacement. We note that with $\mathrm{X_S}$ very large, then $\mathrm{C_{s,rep}}$ could go to negligibly small values, which is allowed mathematically (Sec. \ref{sec:Evaluating}). The cost of the flowing liquid system (liquid blanket, jets, collection system) should be included in $\mathrm{(M\$/S)_{FPP}}$.  The  $\mathrm{(c/Y)_{target}}$ is determined from the D-T target and hohlraum costs normalized to the implosion fusion yield.  $\mathrm{\eta_E}$ is strongly affected  by the thermal conversion efficiency of the blanket liquid and the wall plug efficiency of the laser drivers.

    \item \textbf{Field-reversed configuration (FRC) using advanced fuels }\cite{rostoker1997colliding,momota1992conceptual,kirtley2023fundamental,kirtley2024fundamental,slough2025compact}
    S will comprise of the surface area of the central cell (where fusion occurs) and end plates where open field lines intercept solid surfaces and/or plasma formation occurs, since both of these surface participate in energy removal. With $\mathrm{D-\isotope[3]He}$ fusion $\mathrm{X_S}$ may be set by 2.45 MeV D-D neutron fluence (Appendix \ref{sec:Xs_appendix}) to central cell surfaces. With $\mathrm{p-\isotope[11]B}$, which should have minimal neutron production, the $\mathrm{X_S}$ is more likely set by high energy particle fluence used in direct energy conversion; see for example \cite{momota1992conceptual} where finite lifetime grids are used in charged particle energy extraction in end cells. Whatever the details of the energy extraction, the very nature of fusion energy will likely impose a finite $\mathrm{X_S}$ as discussed in Sec. \ref{sec:utilization_factor}. Yet $\mathrm{\tau_{rep}}$ should be positively impacted by the FRC linear geometry and the separation/reduction of activated components.  The $\mathrm{(c/Y)_{target}}$ must take into account frequency and cost of wearable plasma formation equipment. The $\mathrm{\eta_E}$ must account for the distribution of energy removal (e.g. central cell versus end cell) and the direct conversion efficiency of both particle and electromagnetic energy recovery methods.

    \item \textbf{Tandem high-field mirror} \cite{frank2025confinement,forest2024prospects}
    S includes the surface of central and tandem cells, that receive mostly neutron flux, as well as the end cell targets that receive charged particle power exhaust. Both of these surface participate significantly in energy exhaust.
    $\mathrm{X_S}$ will be set by either neutron fluence limit in the central cells or the charged particle fluence (e.g. erosion) limit at end cell targets, which can be calculated based on the known fixed fraction of energy fluence to both of the surfaces.  $\mathrm{(M\$/S)_S}$ must include normalized costs for both of these replacements appropriately weighted for replacement frequency, as well as the replacement costs for neutron damaged neutral beam components which receive a fixed fraction of neutron fluence compared to the first wall.   $\mathrm{\tau_{rep}}$ will be strongly impacted by the linear geometry and any requirement to energize/cool the magnets. The $\mathrm{\eta_E}$ should account for neutral-beam wall plug efficiency and the use of any direct energy conversion of charged particles in the end cell \cite{forest2024prospects}.

    \item \textbf{Magnetized target fusion using acoustic compression} \cite{laberge2019magnetized,laberge2013acoustically}. The inner surface of the spherical, 
    $\mathrm{\sim 4 \pi}$ steradian PbLi liquid blanket is a natural choice for S. $\mathrm{X_s}$ will likely \emph{not} be set by neutrons since the liquid should not have a damage limit, and neutron fluence to the blanket tank will be minimal through the thick (2 m) liquid. However the compressing pistons may have a lifetime $\ll \tau_{life}$ set by cyclic fatigue since they are undergoing stresses exceeding 1 GPa at $\sim$ 1 Hz \cite{suponitsky2017propagation}.  In this case the $\mathrm{X_s}$ is readily calculable from the piston cycle limit because there is a fixed ratio of compressive energy to fusion energy for each pulse \cite{laberge2013acoustically}. This illustrates the power of using S as the normalization since both the compressive energy, which is a large component of the externally supplied energy, and the produced fusion energy must pass ``through'' S. Then $\mathrm{\tau_{rep}}$ would be set by the time to empty the PbLi liquid from the tank, replace the pistons, and refill the tank. The PbLi liquid  cost should be included in $\mathrm{(M\$/S)_{FPP}}$ since it is a lifetime component. The concept compresses plasma targets formed by coaxial helicity injection, therefore target costs must include the cost of electrodes, which have finite erosion per pulse. The $\mathrm{(c/Y)_{target}}$ value is obtained from the electrode cycle limit, their unit cost and the fixed fusion energy yield per pulse.  The determination of $\mathrm{\eta_E}$ must include the compression energy recovered by piston recoil from the liquid expansion following the pulse.

    \item \textbf{Toroidal MFE with segmented  metal-structure blanket} \cite{lion2025stellaris,Overview_ARIES-RS}
    S is the first-wall components attached to the blanket structure, and $\mathrm{X_S}$ is likely set by neutron fluence to those solid components. 
    The $\mathrm{(M\$/S)_S}$ would include the replaceable neutron damaged blanket segments, while the vacuum vessel exterior to the blanket is a lifetime component included in $\mathrm{(M\$/S)_{FPP}}$. Due to the highly varied neutron fluence and spectrum as a function of distance into the blanket, it can be the case that the innermost radial segments are replaced at a faster rate than the outer segments (e.g. blankets and shields in \cite{Overview_ARIES-RS, Overview_ARIES-AT}). In this case the $\mathrm{(M\$/S)_S}$ should reflect the lifetime-averaged replacement frequency to fit with the temporal equilibrium of the framework. For instance, if the outer segments are replaced every other maintenance, then $\mathrm{(M\$/S)_S}$ would be half of the outer segment cost added to the full cost of the inner segments.  $\mathrm{\tau_{rep}}$ will be determined from the mechanical and installation requirements of the removal and installation of the sectors through access ports. $\mathrm{(M\$/S)_S}$ should account for disposal costs of the full blanket segments and associated hardware.
    When used \cite{lion2025stellaris} the production and delivery cost per cryogenic fuel pellet must be included in $\mathrm{(c/Y)_{target}}$ where its numerical value would be determined from the fusion power, the required pellet injection rate and the pellet cost.

    \item \textbf{Current-pulse driven magnetized liner inertial fusion} \cite{alexander2025affordable}
    S is the internal area of the chamber/blanket with a direct view of the injected target.  $\mathrm{X_S}$ will be set by neutron damage to the chamber/blanket solid structures holding the working fluid of the blanket (see Fig. 19 of \cite{alexander2025affordable}) and possibly the target-injection hardware which can receive finite neutron flux. The  $\mathrm{(M\$/S)_S}$ and  $\mathrm{\tau_{rep}}$ would be determined by the cost and replacement times of the chamber/blanket solid structures, while the blanket liquid cost should be in     $\mathrm{(M\$/S)_{FPP}}$.  In this FPP concept the $\mathrm{(c/Y)_{target}}$ involves the cost of the MagLIF target and the replaceable transmission lines which are consumed with each pulse. 

\end{itemize}

\section{Results: Evaluating FPP economics}\label{sec:Evaluating}

\subsection{Summary of controlling equations and parameters for $Q_{econ}$}
\label{sec:summary_equation_controlling_parameters}

Section \ref{sec:derivation} contains the detailed derivations of the four economic gain and loss rates, and the detailed definitions of their ten controlling parameters. For convenience the controlling parameters are summarized in Table \ref{tab:Control_parameters}.

\begin{table}[ht]
\caption{Controlling parameters for the FPP economic framework \label{tab:Control_parameters}}
\centering
\begin{tabular}{|c||c||c|}
\hline
Description             &   Symbol  &   Unit    \\
\hline \hline
Areal fusion power density during operations   & $P_f/S$     &  $MW/m^2$      \\
\hline
S replacement time           & $\tau_{rep}$  &  $y$      \\
\hline
FPP  lifetime           & $\tau_{life}$ &  $y$      \\
\hline
Net Price of Energy (\$ / $P_{net}$ )             & $POE_{net}$          &  $ \$ /MW \text - h$      \\
\hline
% Energy to dpa conversion*     & $F_{dpa}$     & $\frac{dpa}{(MW/m^2) \text -  y}$    \\
% \hline
% S dpa lifetime limit*         & $L_{dpa}$     &  dpa      \\
% \hline
S energy fluence limit         & $X_{S}$     &  $MW \text - y /m^2$     \\
\hline
FPP net energy conversion efficiency   & $\eta_E$   &  $\equiv P_{net}/{P_f}$      \\
\hline
Fusion target cost per energy yield        & $(c/Y)_{target}$   &  $ \$ / MJ$      \\
\hline
Integrated S replacement areal cost                & $(M\$/S)_S$    &  $M\$/m^2$      \\
\hline
Integrated FPP areal cost            & $(M\$/S)_{FPP}$ &  $M\$/m^2$  \\
\hline
Construction + Financing real interest rate               &  $i$                & $\%$ \\
\hline
\end{tabular}
\end{table}

For convenience we summarize the governing equations for the economic gain and loss rates.
The economic gain rate from selling energy is
\begin{equation}\label{eq:Cgain_summary}
    C_{gain}   = 
    8.76\times10^{-3}POE_{net}  \cdot \eta_E \cdot(P_f/S) \cdot U \hspace{5mm}, \vspace{2mm}
\end{equation}
the cost rate of consumable fusion targets is
\begin{equation}\label{eq:C_target_summary}
   C_{target}
   =31.5 \enspace (c/Y)_{target} \cdot (P_f/S) \cdot U \hspace{5mm}, \vspace{2mm}
\end{equation}
the replacement cost rate of S is
\begin{equation}\label{eq:C_S_Rep_summary}
    C_{S,rep}  = (M\$/S)_S \cdot \frac{(P_f/S)}{X_S} \cdot U \hspace{5mm}, \vspace{2mm}
\end{equation}
and the construction, financing and operational fixed cost rate is
\begin{equation}\label{eq:C_CF_cost_summary}
    C_{fixed} [M\$/m^2 \text -  y]  = (M\$/S)_{FPP} 
    \frac{0.01 \cdot i(1+0.01 \cdot i)^{\tau_{life}}}{(1+0.01 \cdot i)^{\tau_{life}}-1}
    \hspace{3mm} . \vspace{2mm}
\end{equation}

As shown in Section \ref{sec:utilization_factor} the utilization factor U, which appears in three of the rate equations, can be determined from the controlling parameters by solving for $\mathrm{\tau_{op}}$ based on the energy fluence limit of S, $\mathrm{X_S}$. For completeness we restate Equation \ref{eq:U_rearranged_definition} as the full solution to U,
\begin{equation}\label{eq:Utilization_definition_summary}
   U = [1 + \frac{(P_f/S) \cdot \tau_{rep} }{X_{S}}]^{-1} \hspace{5mm}. \vspace{2mm}
\end{equation}
The design parameters appear here in an inverse sum while also appearing elsewhere in the rate equations. This makes the economic model non-linear, particularly in fusion power density.
Finally, the surface normalized net rate of economic gain or loss, $\mathrm{C_{net}}$ in [$M\$ /m^2 \text -  y$] for the FPP is attained by the balance of economic gain and the three cost terms, specifically:
\begin{equation}\label{eq:C_net_definition}
    C_{net} \enspace [\frac{M\$}{m^2 \text -  y}]  = 
    C_{gain} - C_{target} - C_{S,rep} - C_{fixed} \hspace{3mm}. \vspace{2mm}
\end{equation}

The economic framework took its inspiration from the Lawson criterion, which allows one to examine the physical requirements for the plasma energy gain $\mathrm{Q_p}$, the ratio of fusion power output to external power input, in power balance.  Analogously, with this economic model at steady-state economic rates, we can now define and calculate
\begin{equation}\label{eq:Q_econ_definition}
  Q_{econ} \equiv \frac{C_{gain}}{C_{target}+C_{S,rep}+C_{fixed}} \hspace{3mm}, \vspace{2mm}
\end{equation}
which is the ratio of economic output to external economic ``input'' (i.e., the cost or expenditure), thus providing a conceptual figure of merit similar to $\mathrm{Q_p}$. As with the Lawson criterion, it is generic and can be applied to any fusion concept or design. Likewise, as with an engineered fusion system, it captures the ``amplifier'' effect of fusion; a successful FPP amplifies both energy and economic input.

The model is intrinsically inflation-adjusted due to the use of real interest rates in Equation \ref{eq:C_CF_cost_summary} in the financing costs.  The operational parameters involving explicit prices ($\mathrm{(c/Y)_{target}}$, $\mathrm{(M\$/S)_S}$,$\mathrm{POE_{net}}$) will vary in monetary value with inflation over the FPP lifetime, but these variations will not affect the $\mathrm{Q_{econ}}$ ratio since all its terms vary together with inflation.  We adopt a simple timing convention which is to assign prices at the time the FPP has been constructed, by expending $\mathrm{(M\$/S)_{C+F}}$, and energy production operations are beginning for the first time for the FPP.  This means that the numerical value of $\mathrm{C_{net}}$ will vary over the FPP lifetime but $\mathrm{Q_{econ}}$ and the purchasing power of $\mathrm{C_{net}}$ are constant in time.

The controlling parameters of $\mathrm{Q_{econ}}$ come from science, engineering, finance, and their interactions.  $\mathrm{Q_{econ}}$ is designed as a concept-agnostic design-window metric, analogous to the role of scientific $\mathrm{Q_p}$ in fusion plasma science. As with the distinction between scientific and engineering $\mathrm{Q_{eng}}$ \cite{menard2011prospects}, obtaining $\mathrm{Q_{econ}}>1$ is a necessary but insufficient condition for commercial viability. A complete FPP finance assessment which incorporates Net Present Value (NPV), Internal Rate of Return (IRR), cash flow timing, capital structure, and market dynamics, requires plant-specific inputs that are generally not available at this development stage for fusion. While these factors are  outside the scope of this present work, expanding the framework for their inclusion is discussed in Sec. \ref{sec:additional considerations}.
Nevertheless, NPV and IRR are strongly linked to the concepts developed in this framework and $\mathrm{Q_{econ}}$, such as an understanding of how capital outflow and the cost of capital affect economic viability.
The $\mathrm{LCOE_{eff}}$ (Eq \ref{eq:LCOE_definition} derived in Sec. \ref{sec:Evaluation_Cnet}) similarly carries an "effective" qualifier acknowledging its omission of depreciation, tax treatment, and other factors included in a standard NPV-based LCOE.

In addition, certain categories of upfront costs (e.g., regulatory burden, development costs) have been excluded by default from the framework due to the temporal equilibrium required for the model. Such upfront costs typically become relatively less important as an industry matures, stressing that this economic framework is most accurately applied to mature NOAK FPPs. 
To link this conceptually back to the analogous Lawson criterion, this is the equivalent of excluding the energy requirements to get the plasma to the operating point required for a given $\mathrm{Q_p}$. However, if the steady-state equilibrium gain is too low, there is no fundamental viability, regardless of those transient details. These features make both $\mathrm{Q_p}$ and $\mathrm{Q_{econ}}$  useful and fundamental targets for FPP design evaluations.

\subsection{Solving economic gain, $\mathrm{Q_{econ}}$, overnight costs and LCOE}\label{sec:Evaluation_Cnet}

The net gain equations are cast into in a more convenient form for solving.
Examination of the controlling parameters and Equations \ref{eq:Cgain_summary} - \ref{eq:Utilization_definition_summary} indicates that the areal fusion power density $\mathrm{P_f/S}$ and the utilization U appear in three of four equations (gain, target, and replacement), all of which have the same functional dependence on $\mathrm{(P_f/S)\cdot U}$, with U depending on $\mathrm{P_f/S}$, suggesting the following form to solve: 
\begin{equation}\label{eq:C_BE_equation_form}
   \frac{A \cdot(P_f/S)}{1+(P_f/S)\cdot B}=C \hspace{3mm}. \vspace{2mm}
\end{equation}
We define the variable A, which is a ratio of economic worth to energy:
\begin{equation}\label{eq:A_definition}
\begin{split}
    A [\frac{M\$}{MW \text -  y}]=A_1-A_2-A_3 \\
   A_1 \enspace  \equiv 8.76\times10^{-3}POE_{net}  \cdot \eta_E \enspace \\
   A_2 \enspace \equiv 31.5 \cdot (c/Y)_{target} \\
   A_3 \enspace \equiv \frac{(M\$/S)_S}{X_S} \hspace{3mm}. \vspace{2mm}
\end{split}
\end{equation}
The variable A thus has no dependence on $\mathrm{P_f/S}$. We also define
\begin{equation}\label{eq:B_definition}
   B \enspace [\frac{m^2}{MW}] \equiv \frac{\tau_{rep}}{X_S} \hspace{3mm}, \vspace{2mm}
\end{equation}
and
\begin{equation}\label{eq:C_definition}
   C \enspace [\frac{M\$}{m^2 \text -  y}] \equiv C_{fixed}+C_{net} \hspace{3mm}. \vspace{2mm}
\end{equation}
This allows us to solve for the required $\mathrm{P_f/S}$ at a targeted $\mathrm{C_{net}}$, which sets C since $\mathrm{C_{fixed}}$ is fixed by definition, namely:
\begin{equation}\label{eq:PfS_solution}
   (P_f/S)_C=\frac{C}{A-(B\cdot C)} \hspace{3mm}. \vspace{2mm}
\end{equation}

A special case occurs when the FPP achieves economic breakeven (denoted as ``BE"), at $\mathrm{C_{net} = 0}$ and $\mathrm{Q_{econ}=1}$, which occurs when the gain and loss rates offset each other:
\begin{equation}\label{eq:C_BE_definition}
   [C_{gain} - C_{target} - C_{S,rep} - C_{fixed}]_{BE} = 0 \hspace{3mm}. \vspace{2mm}
\end{equation}
The breakeven solution, $\mathrm{C_{net}=0}$ for areal power density, follows from Equation \ref{eq:PfS_solution},
\begin{equation}\label{eq:PfS_BE_solution}
   (P_f/S)_{BE}=\frac{C_{fixed}}{A-(B\cdot C_{fixed})} \hspace{3mm},  \vspace{2mm}
\end{equation}
which is the minimum power density required for economic viability when the other controlling parameters are fixed.  

This formulation may be adapted to solve for the value of $\mathrm{P_f/S}$ for a targeted $\mathrm{Q_{econ}}$,
\begin{equation}\label{eq:PfS_Q_econ_solution}
   (P_f/S)_{Q_{econ}}=\frac{Q_{econ} \cdot C_{fixed}}
   {A_1-Q_{econ} \cdot (A_2 + A_3+ B \cdot C_{fixed})} \hspace{3mm}. \vspace{2mm}
\end{equation}
While the relationship between $\mathrm{P_f/S}$ and $\mathrm{Q_{econ}}$ is nonlinear it is monotonic, thus assuring unique solutions. 

Equations \ref{eq:C_BE_equation_form} through \ref{eq:PfS_Q_econ_solution} constitute the set of solvable equations and definitions that will link the model's engineering parameters to its economic parameters of net gain rate and $\mathrm{Q_{econ}}$. Using solutions for $\mathrm{P_f/S}$ at a targeted $\mathrm{Q_{econ}}$ (or at breakeven), various operational and economic parameters can be evaluated, since $\mathrm{P_f/S}$ is unique.  The utilization parameter is of interest as an engineering parameter, as it informs the FPP operator about maximum availability:
\begin{equation}\label{eq:Utilization_PfS_solution}
   U_{Q_{econ}} = [1 + \frac{(P_f/S)_{Q_{econ}} \cdot \tau_{rep} }{X_S}]^{-1} \hspace{3mm}. \vspace{2mm}
\end{equation}
Likewise, as an economic and market parameter, the overnight cost can be solved from Equation \ref{eq:covernight_def}:
\begin{equation}\label{eq:covernight_def_again}
    c_{O/N,Q_{econ}} [\$/W] = \frac{(M\$/S)_{FPP}}{(P_f/S)_{Q_{econ}}\cdot\eta_E} \hspace{3mm}. \vspace{2mm}
\end{equation}
We note that this definition is consistent with the energy industry standard by using the maximum power generation possible, rather than the time-averaged power. 

We may also extract an effective levelized cost of energy (LCOE) from the model. The LCOE, which is typically given in units of $\mathrm{ \$/MW\text -  h }$, is defined as the total cost (including operating, construction, and financing costs) per energy output over the lifetime of the FPP. Due to the temporal equilibrium of the model, this will be a constant in time. Again exploiting the normalization to S, our derivation starts with:
\begin{equation}\label{eq:LCOE_definition}
   LCOE_{eff} \enspace [\$/MW\text -  h] =
   114.2 \frac{C_{target}+C_{S,rep}+C_{fixed}}{{(P_f/S)} \cdot \eta_E \cdot U} \hspace{3mm}. \vspace{2mm}
\end{equation}
The denominator provides the time-averaged net power output and the constant provides unit conversion to the commercial convention. The ``effective'' subscript acknowledges the lack of fixed costs, depreciation and other factors  which are normally included in a NPV-based LCOE calculation. As with $\mathrm{Q_{econ}}$, $\mathrm{LCOE_{eff}}$ is a useful figure of merit to gauge market access. By examination, we obtain another solution form:
\begin{equation}\label{eq:LCOE_definition_Qecon}
   LCOE_{eff} = \frac{POE_{net}}{Q_{econ}} \hspace{3mm}, \vspace{2mm}
\end{equation}
which allows for easy evaluation of the $\mathrm{LCOE_{eff}}$ from graphical representations of $\mathrm{Q_{econ}}$ with knowledge of $\mathrm{POE_{net}}$, and therefore will not be plotted explicitly going forward as a model output.   
We note that $\mathrm{LCOE_{eff}}$ does not actually depend on $\mathrm{POE_{net}}$, since it does not appear in any term in Equation \ref{eq:LCOE_definition} and  both the numerator and denominator of Equation \ref{eq:LCOE_definition_Qecon}  vary linearly with $\mathrm{POE_{net}}$.
The parameters $\mathrm{U_{Q_{econ}}}$, $\mathrm{c_{O/N,Q_{econ}}}$ and $\mathrm{LCOE_{eff}}$ are \textit{outputs} of the model framework (not assumptions) at the required power density to meet its economic gain constraints. 

\subsection{Model base case and ranged values }\label{sec:base_case_parameters}

We estimate a set of base case values for the input control parameters in order to exercise the economic framework.
The base case values in Table \ref{tab:Default_parameters} are chosen based on their reasonability with projected FPP designs and on historic data from earlier fusion and nuclear fission power projects.  \emph{It is important to note that this effort is in no way a cost or design prediction for an FPP; that must occur in its bottom-up technical design.} However, we desire to have reasonable base case values to explore trends and to understand their relative sensitivity to design parameters around that point. Regardless, the input parameters will be varied over the large range of values listed in Table \ref{tab:Default_parameters} in the following sections.  The ease in scoping the economic viability with a  large range of controlling parameters is one of the fundamental motivations for developing a high-level model framework.

\begin{table}[ht]
\caption{Base case and ranged values of input parameters (Table \ref{tab:Control_parameters}) for scoping sensitivity studies.  \label{tab:Default_parameters}}
\centering
\begin{tabular}{|c||c||c||c|}
\hline
Parameter       &    Base case &   Range      &    Unit     \\
\hline \hline
$\tau_{rep}$    &       0.1  &  0 - 0.5     &  $y$      \\
\hline
$\tau_{life}$   &      30.0  &  10 - 40     & $y$      \\
\hline
$POE_{net}$             &       100  &  20 - 200   &  $ \$ /MW\text -  h$ \\
\hline
$X_S$           &      3.125 &  0.1 - 6     &  $MW \text -  y /m^2$\\
\hline
$\eta_E$        &       0.4  &  0.1 - 0.6    &  $\equiv \frac{P_{net}}{P_f}$  \\
\hline
$(c/Y)_{target}$&       0  &  0 - 0.01     &  $ \$ / MJ$      \\
\hline
$(M\$/S)_S$     &       0.3  &  0 - 1     &  $M\$/m^2$      \\
\hline
$(M\$/S)_{FPP}$ &      10    &  3 - 25     &  $M\$/m^2$  \\
\hline
$i$              &        2.0  &  -2 - 5     & $\%$ \\
\hline
\end{tabular}
\end{table}

The replacement time of 0.1 y ($\sim $ 5 weeks) is taken from typical refueling time of a fission power plant. A plant lifetime $\mathrm{\tau_{life}}$ of several decades is expected and a financing timescale of 15 to 30 years is typical of large construction development. The $\mathrm{POE_{net}}$ is estimated from the average inflation-adjusted US retail price of electricity over the last decade, $\mathrm{\sim 160 \enspace \$ /MW \text -  h}$, under the assumption that generation is $\mathrm{\sim 2/3}$ of that cost and that variable O\&M costs (Eq. \ref{eq:POE_definition} ) are already captured in this value.   The energy conversion efficiency $\mathrm{\eta_E=0.4}$ is typical of predictions for fusion blanket designs (c.f. ARIES designs in Table \ref{tab:FPP_parameters}, Section \ref{sec:Areal_cost_appendix}). A real interest rate of 2 \% is estimated based on the US Federal Reserve inflation target of 2 \%. The real interest rate can vary below zero, for example in cases where the cost of capital is substantially less than inflation. The base case target cost is set to zero, but will be varied over a large range, up to 0.01 \$/MJ or \$10 per target for gigajoule yield. Thus, the default case has effectively zero cost for fuel, a distinguishing feature of many fusion concepts.

There is little to no direct data on the energy fluence limit for FPP components, and thus $\mathrm{X_S}$ is given wide variability.   From the examples of D-T fusion in Appendix \ref{sec:Xs_appendix}, taking $\mathrm{L_{dpa}=25}$ and $\mathrm{F_{dpa} \simeq 10 \enspace dpa/m^2/MW \text -  y}$ provides a value of $\mathrm{X_S=3.125 \enspace MW \text -  y /m^2}$. The reason D-T fusion is used as an example is that it is highly likely that fast neutron fluence determines $\mathrm{X_S}$. However, the results presented here are generic to any FPP fuel cycle, and the value of $\mathrm{X_S}$ will be widely varied regardless in sensitivity scans. To place this assumption in context,  the solution example in Figure \ref{fig:Cnet_example} has a required fusion power density $\mathrm{\sim 2-5 \enspace MW/m^2}$ for economic viability and so $\mathrm{X_S \simeq \enspace 3}$ implies an operational lifetime of S $\sim 0.6-1.5$ years. This seems like a reasonable starting point and again aligns with the fission experience, where refueling occurs approximately once a year. We  note that the model accommodates fusion concepts where S is continually being replenished (e.g. flowing liquid walls) in which case either $\mathrm{\tau_{rep}} \rightarrow 0$ or $\mathrm{X_S \rightarrow \infty}$ capture this design feature without diverging the values $\mathrm{C_{net}}$ or $\mathrm{Q_{econ}}$.

The base case FPP  areal cost is set at $\mathrm{(M\$/S)_{FPP}=10 \enspace M\$/m^2}$ based on historical data and FPP design studies. The default S areal cost is taken as $3\%$ of the FPP areal cost, $\mathrm{(M\$/S)_{S}=0.3 \enspace M\$/m^2}$, with the realization that there is a such a wide variety of technology choices for S and its associated blanket that the value is inherently difficult to estimate.  Both of these cost choices are covered in detail in the Discussion, Section \ref{sec:Discussion} . It is noted that $\mathrm{(M\$/S)_{S}}$ is only the replacement cost, and therefore would not include any materials and components in S or its blanket that are reused.  An example of this would be the cost of coolants and liquid breeders, which would presumably be recovered and recycled after each operational period.   As both these areal costs have large uncertainties, they will be varied by large fractions of their base case values in the sensitivity analysis.

\subsection{Example solution}

An illustrative solution for $\mathrm{C_{net}}$ versus $\mathrm{P_f/S}$ is provided in Figure \ref{fig:Cnet_example}. The choice of base case input parameters listed in Table \ref{tab:Default_parameters} is described in the previous section. From the input parameters, the derived solution variables are: $\mathrm{A_1}$=0.35, $\mathrm{A_2}$=0, $\mathrm{A_3}$=0.032 $\mathrm{M\$/MW \text -  y}$ (Equation \ref{eq:A_definition}),  B=0.096 $\mathrm{m^2/MW}$ (Equation \ref{eq:B_definition}) and $\mathrm{C_{fixed}}$=0.446 $\mathrm{M\$/m^2 \text -  y}$ (Equation \ref{eq:C_definition}). 

\begin{figure}[H]
    \centering
    \includegraphics[width=0.85\textwidth]{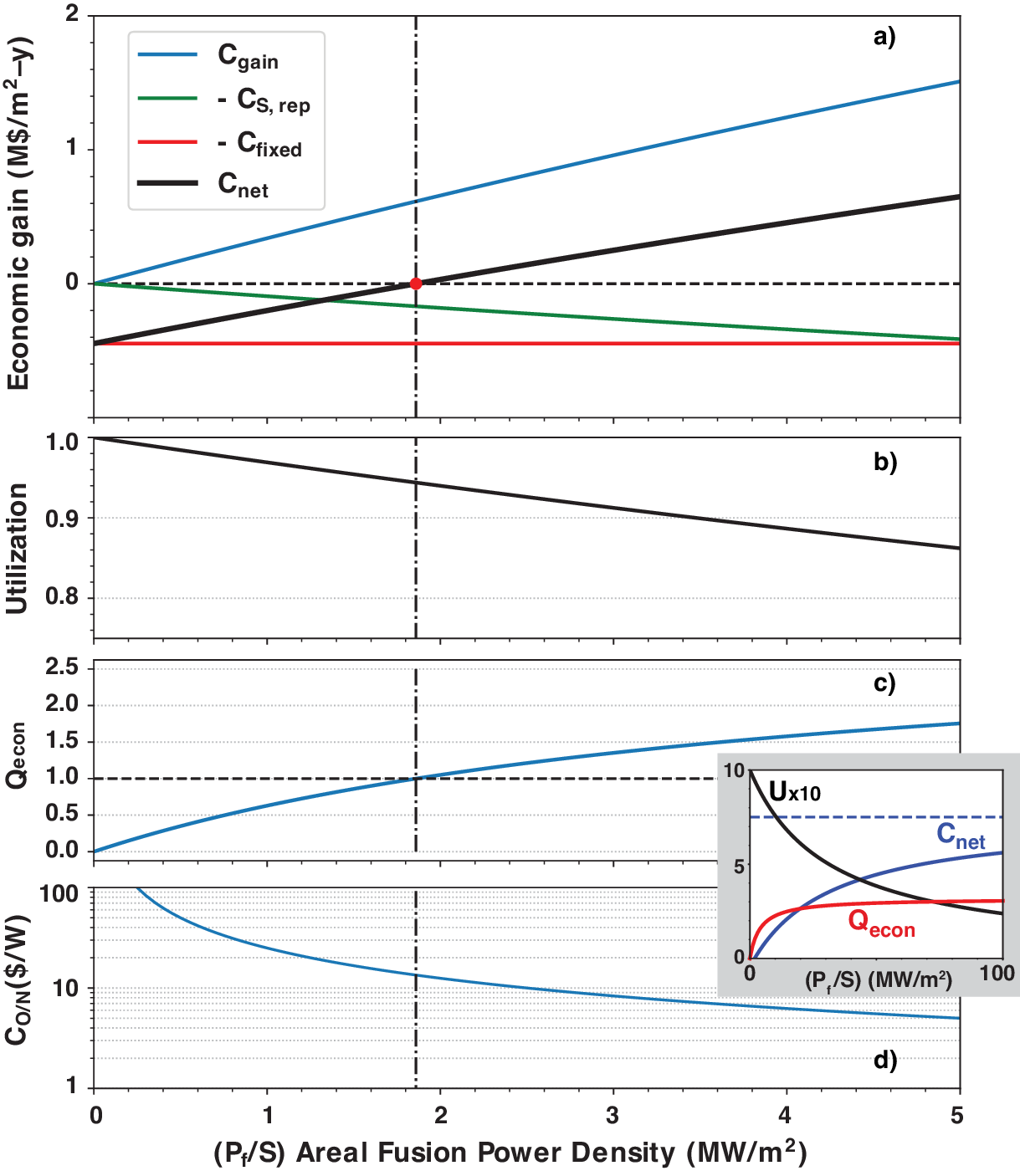}
    \caption{Example calculation of economic model parameters versus $\mathrm{P_f/S}$ with parameters listed in Table \ref{tab:Default_parameters}: a) Economic gain, cost and net rates with costs plotted as negative values. Breakeven occurs where $\mathrm{C_{net}}$ crosses zero at $\mathrm{(P_f/S)_{BE} \simeq 1.9 \enspace MW/m^2}$ indicated by the vertical line, with the analytic solution (Equation \ref{eq:PfS_BE_solution}) for $\mathrm{(P_f/S)_{BE}}$ indicated by red dot. b) Utilization factor U  c) Economic gain $\mathrm{Q_{econ}}$ and d) Overnight cost. (Insert) $\mathrm{C_{net}}$, utilization and $\mathrm{Q_{econ}}$ at higher $\mathrm{P_f/S}$, asymptote maximum values of $\mathrm{C_{net,max}}$ is the horizontal dashed line.}
    \label{fig:Cnet_example}
\end{figure}

We see that $\mathrm{C_{net}}$ increases monotonically, but nonlinearly, with $\mathrm{P_f/S}$ starting from negative values at $\mathrm{P_f/S=0}$, and crossing zero at $\mathrm{(P_f/S)_{BE}\simeq 1.9 \enspace MW/m^2}$. At BE the utilization is 0.94, and \textit{decreases} as the economic gain rate increases for $\mathrm{P_f/S}$ above breakeven. This already provides us with an important insight: the economic attractiveness of fusion prefers to be at lower utilization values due to the direct link between energy output and the energy ``usage'' of S. This trend will push against FPP market requirements for high availability.

%Joule a topic which will be explored later.
 
By definition, $\mathrm{Q_{econ}}$=1.0 at BE, and then varies nonlinearly with $\mathrm{P_f/S}$, with $\mathrm{Q_{econ}}$=1.5 achieved at $\mathrm{(P_f/S)\simeq  3.7 \enspace MW/m^2}$, more than 1.5 times the power density at BE. At BE the overnight cost is $\mathrm{c_{O/N}\simeq 12 \enspace \$/W}$, which is typical of entry-level price points for new energy sources. The overnight cost decreases with higher $\mathrm{P_f/S}$, and at $\mathrm{Q_{econ}}$=1.5, it has decreased to $\mathrm{c_{O/N}\simeq 6 \enspace \$/W}$, which is becoming competitive for dispatchable power. At $\mathrm{Q_{econ}}$=1.5, and with $\mathrm{POE_{net}}$ set at $\mathrm{100 \enspace \$/MW \text - h }$,  this gives $\mathrm{LCOE_{eff}}$ at 67 $\mathrm{\$/MW\text -  h}$. The fact that the overnight and levelized costs generated by the economic model are market typical indicates the practical nature of the model, and that our base case parameters (Table \ref{tab:Default_parameters}) are a reasonable starting point for exploring their relative sensitivities for an economically interesting FPP.

The insert in Figure \ref{fig:Cnet_example} shows the asymptotic behavior of the economic model at very high fusion power densities.  Above $\mathrm{P_f/S\sim 50 \enspace MW/m^2}$,  both $\mathrm{C_{net}}$ and $\mathrm{Q_{econ}}$ are highly nonlinear and start to saturate, although $\mathrm{Q_{econ}}$ flattens at a lower power density and becomes constant at $\mathrm{Q_{econ}\sim 3}$. In this region, the utilization falls below 0.5, indicating that the economics of the FPP are becoming dominated by its operating costs.  

\subsection{General observations on economic gain solutions}\label{sec:Evaluation_observations}

An examination of the constituent equations in the previous sections provides three high-level insights. 

\begin{enumerate}
   
    \item

All cases of interest should have at least a positive economic return $\mathrm{C_{net} \geq 0}$, in which case C is definitely positive, as are $\mathrm{P_f/S}$ and B. Therefore, a fundamental requirement for economic viability, $\mathrm{C_{net} \geq 0}$ or $\mathrm{Q_{econ} \geq 1}$,  is that $\mathrm{A \geq 0}$, or specifically:
\begin{equation}\label{eq:A_requirement}
    A_1 \geq A_2+A_3 \hspace{3mm}, \vspace{2mm}
\end{equation}
which will depend on a variety of market and technical parameters. 

    \item 
    
While $\mathrm{A \geq 0}$  is necessary, it is an insufficient condition for economic viability. Notably, at an arbitrarily low areal fusion power density, $\mathrm{U \rightarrow 1}$, and the left-hand side of Equation \ref{eq:C_BE_equation_form} collapses to $\mathrm{A \cdot P_f/S }$, which cannot meet the breakeven criterion  as $\mathrm{P_f/S \rightarrow 0 }$ (see Figure \ref{fig:Cnet_example}).  
While this may seem trivial mathematically, it is in fact fundamental to the framework, because it speaks to the purpose of an FPP, which is to make economic gains selling energy for its operator, \emph{not} assuring that the FPP has an extremely high utilization by slowly wearing through the components of S.  Further to this point, the left-hand side of Equation \ref{eq:C_BE_equation_form} increases monotonically with $\mathrm{P_f/S}$, and therefore so does C and $\mathrm{C_{net}}$. Stated differently, increasing $\mathrm{P_f/S}$ always improves economic gain if $\mathrm{A > 0}$. Correspondingly, if $\mathrm{A < 0}$, then increasing $\mathrm{P_f/S}$ always increases the economic loss rate, since $\mathrm{C_{net}}$ will become more negative, which means the FPP is losing money in its operation.

    \item 
    
In the limit of an arbitrarily high power density, the net gain reaches an asymptote and maximizes at:
\begin{equation}\label{eq:Cnet_PfS_high}
   C_{net,max} = \frac{A}{B}-C_{fixed} \enspace
   \textrm{ , for} \ P_f/S \rightarrow \infty \hspace{3mm}. \vspace{2mm}
\end{equation}
Similarly, the asymptotic behavior of economic gain is:
\begin{equation}\label{eq:Q_econ_max}
   Q_{econ,max} = \frac{A_1}{A_2+A_3+B \cdot C_{fixed}} \hspace{3mm}. \vspace{2mm}
\end{equation}
This limit corresponds to the case where the fusion power density is so high that the energy fluence limit through S is reached nearly instantaneously.  This represents the maximum economic gain in the model, where one is producing energy as fast as it is possible to sell, the only limitation being the required waiting time to replace S in  $\mathrm{\tau_{rep}}$.  This limit does not judge whether or not the power density is achievable, nor if the  market would accept a large instantaneous influx of energy. Regardless, this represents an interesting theoretical limit to economic gain in an FPP.  An example of this asymptotic behavior at high $\mathrm{P_f/S}$ is shown in the insert of Figure \ref{fig:Cnet_example}, with substantial flattening of $\mathrm{Q_{econ}}$ and $\mathrm{C_{net}}$ occurring past $\mathrm{P_f/S \sim 50 \enspace  MW/m^2}$ and utilization $\mathrm{U< 0.5}$.  In this example, $\mathrm{Q_{econ}}$ reaches $\sim$ 3, but at areal power densities which likely surpass surface cooling limits. Regardless of its practicality, it is important to note that an FPP will have a finite maximum $\mathrm{Q_{econ}}$.

\end{enumerate}

\section{Results: Economic gain evaluation varying two control parameters}\label{sec:results_two_params}

With the economic framework derived in Section \ref{sec:derivation} and tools developed for its evaluation in Section \ref{sec:Evaluating}, we are now interested in evaluating the criteria that provide economic viability  design ``spaces" by varying the control parameters. However, with ten controlling parameters it is not possible to provide graphical evaluation of these spaces for the full model. In this section we will vary two controlling parameters, and the key model outputs ($\mathrm{Q_{econ}}$, utilization, overnight cost) are presented in contour plots.

\subsection{Varying $\mathrm{P_f/S}$ and one other parameter }\label{sec:Sensitivity_PfS_scan}

Fusion power density is a natural choice of controlling parameter to vary in our framework. Of all the controlling parameters, this is the only one with a direct link to fusion plasma performance (and indirectly, to {$\mathrm{Q_p}$}), and thus it must be a focus for any FPP designer. Furthermore, $\mathrm{P_f/S}$ is the term that appears most frequently in the governing equations, and so we reasonably expect interesting behavior when it is varied against the other controlling parameters. In this section, we will examine the results of the economic model as $\mathrm{P_f/S}$ is varied from 0.5 to 10 $\mathrm{MW/m^2}$ while one other controlling parameter is varied over the range provided in Table \ref{tab:Default_parameters}, whilst the other parameters are fixed at their base case values. The lower limit power density is being set near the lowest $\mathrm{P_f/S}$ of FPP-class devices (ITER 0.7  $\mathrm{MW/m^2}$, Table \ref{tab:FPP_parameters}), up to a likely global heat exhaust technology limit at S at 10 $\mathrm{MW/m^2}$.

\begin{figure}[h]
    \centering
    \includegraphics[width=0.6\textwidth]{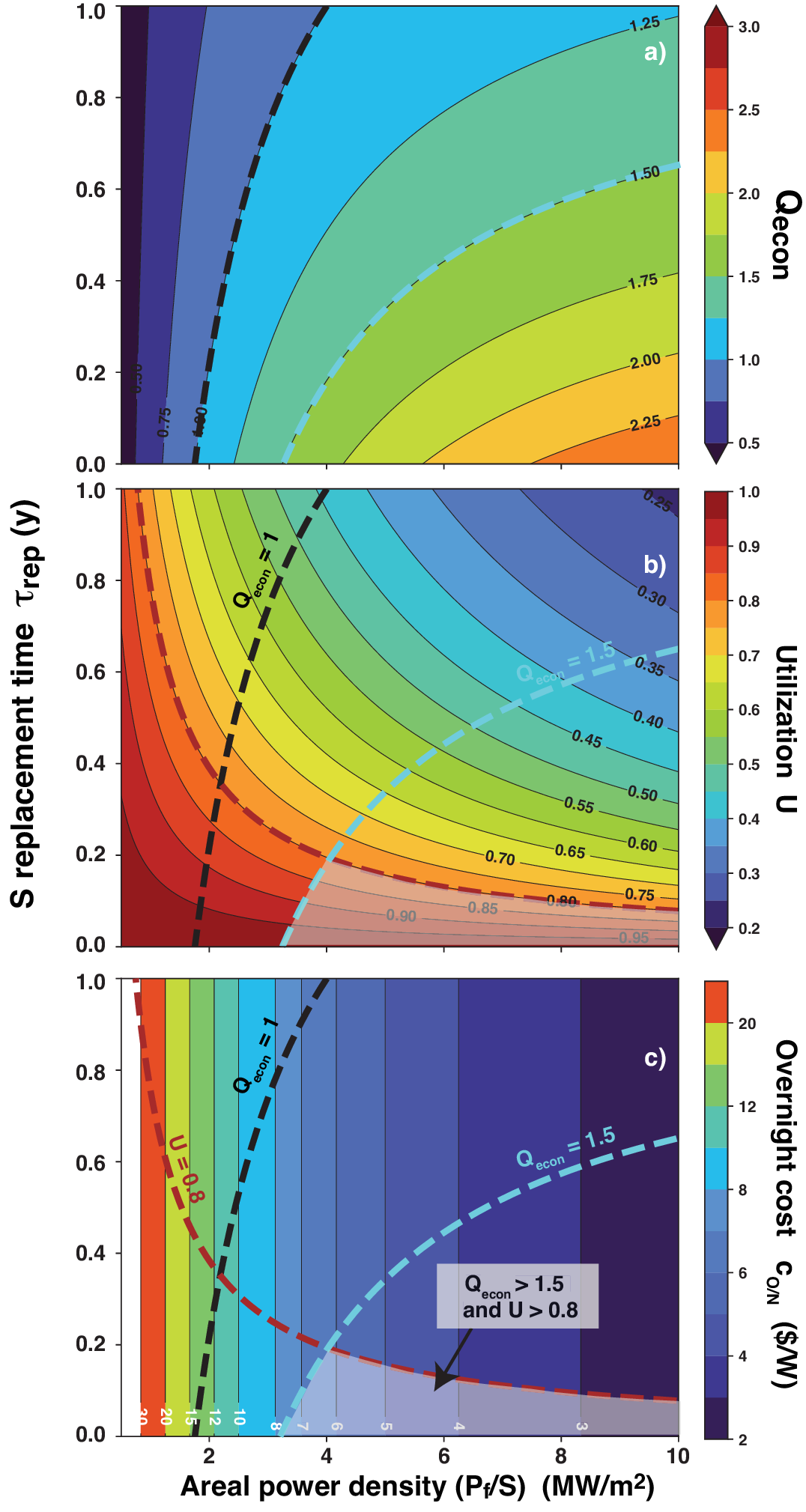}
    \caption{Contour plots of economic model outputs with varying $\mathrm{P_f/S}$ and S replacement time $\mathrm{\tau_{rep}}$ (a) Economic gain $\mathrm{Q_{econ}}$  (b) Utilization U (c) Overnight cost $\mathrm{c_{O/N}}$.
    Overlays of possible minimum commercial goals for $\mathrm{Q_{econ}}$ 
    and U define required design criteria shown by the shaded regions for the example case of $\mathrm{Q_{econ} > 1.5}$ and $\mathrm{U>0.8}$. 
    }
    \label{fig:Contour_PfS_taurep}
\end{figure}

Contour plots of the results of the economic model are depicted in Figure \ref{fig:Contour_PfS_taurep} as a function of $\mathrm{P_f/S}$ and S replacement time $\mathrm{\tau_{rep}}$. This provides insight as to the interaction between these two design features and the necessary combinations that reach a FPP design target for $\mathrm{Q_{econ}}$. The positive economic gain criterion $\mathrm{Q_{econ}=1}$ has a threshold value $\mathrm{\sim 2 \enspace MW/m^2}$ at $\mathrm{\tau_{rep}=0}$, which represents instantaneous S replacement.  This threshold is quite insensitive to the replacement time for $\mathrm{\tau_{rep} \lesssim 0.5}$ years. At higher levels of $\mathrm{P_f/S} \ge 6 $, $\mathrm{Q_{econ}}$ becomes increasingly sensitive to $\mathrm{\tau_{rep}}$  and less sensitive to $\mathrm{P_f/S}$. Access to high gain, e.g., $\mathrm{Q_{econ}>2}$, is disallowed at $\mathrm{\tau_{rep}} \gtrsim 0.2$ years, effectively becoming a threshold design requirement.  The two other important outputs of the model, the utilization and overnight normalized cost, are also shown in  Figure \ref{fig:Contour_PfS_taurep}. The utilization contours are convex in shape, since increasing either $\mathrm{P_f/S}$ or $\mathrm{\tau_{rep}}$ decreases utilization.

A minimum utilization value will likely be a design target for an FPP, for example, to meet customer goals for availability to supply power to an electric utility, or in specialized uses such as powering a data center using a cluster of several FPPs.  A targeted $\mathrm{Q_{econ}}$ for the FPP further constrains the design space.  In the example indicated by the shaded area in Figure \ref{fig:Contour_PfS_taurep}, we show that example simultaneous design targets of $\mathrm{Q_{econ} \geq1.5}$ and $\mathrm{U\geq0.80}$ require a $\mathrm{P_f/S \gtrsim 3 \enspace MW/m^2}$ and $\mathrm{\tau_{rep} \lesssim 0.2 \enspace y}$. While the overnight cost only depends on $\mathrm{P_f/S}$, the permitted values are constrained by both parameters at the given design targets for $\mathrm{Q_{econ}}$ and U.  These plots demonstrate the usefulness of the model framework to inform the design space about potential tradeoffs. For example, if there is a technical decision that low $\mathrm{\tau_{rep}}$ will be more difficult to obtain than high power density, then the design will choose the top corner of the shaded area to meet the other design targets near $\mathrm{\tau_{rep}}\sim0.2$. This in turn will set the overnight cost at $\mathrm{\sim6 \enspace \$/W}$.  We will continue to use this example target region of $\mathrm{Q_{econ} \geq1.5}$ and $\mathrm{U\geq0.80}$ to illuminate how the model framework provides design insights.

The interaction between power density and the S energy fluence limit $\mathrm{X_S}$ is shown in Figure \ref{fig:Contour_PfS_Xs}. The $\mathrm{Q_{econ}}$ contours have a convex shape, and the sharpness of the lower-left corner of isolinear contours increases at low  $\mathrm{Q_{econ}}$. As a consequence, there are effective threshold values for $\mathrm{Q_{econ}=1}$ where the contour becomes vertical at $\mathrm{P_f/S}\sim2$ and horizontal at $\mathrm{X_S}\sim 1$. It is apparent from this plot that there are regions in the design space where increases in $\mathrm{X_S}$ or $\mathrm{P_f/S}$ lead to minimal increases in $\mathrm{Q_{econ}}$. The implications of these contours to technology design and economic effectiveness are further quantified in Section \ref{sec:extract}. 
The isolinear utilization contours have a constant positive slope, with the slope varying with U.  The target design example of $\mathrm{Q_{econ} \geq1.5}$ and $\mathrm{U\geq0.80}$  gives a  design window in the shaded area, further defining threshold values of  $\mathrm{P_f/S \gtrsim 3  \enspace MW/m^2}$ and $\mathrm{X_S \gtrsim \enspace 2 \enspace MW \text - y/m^2}$. At the lower corner of the shaded area, the overnight cost is $\mathrm{\sim5.5 \enspace \$/W}$, similar to the result from Figure \ref{fig:Contour_PfS_taurep}. 

\begin{figure}[h]
    \centering
    \includegraphics[width=0.6\textwidth]{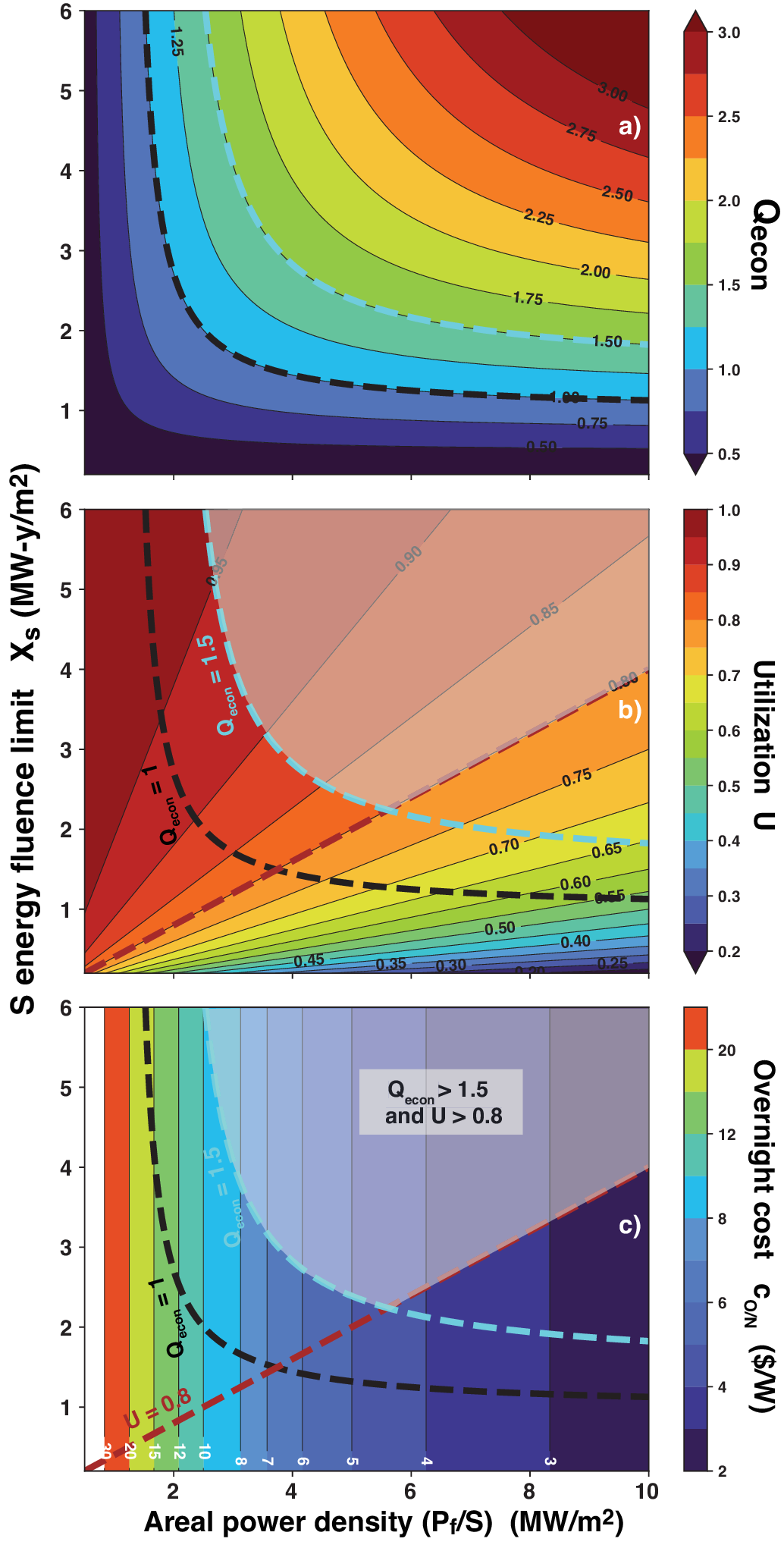}
    \caption{
    Contour plots of economic model outputs with varying $\mathrm{P_f/S}$ and $\mathrm{X_S}$ (a) Economic gain $\mathrm{Q_{econ}}$  (b) Utilization U (c) Overnight cost $\mathrm{c_{O/N}}$
    Overlays of possible minimum commercial goals for $\mathrm{Q_{econ}}$ and U indicate define required design space. Shaded regions shows example target of $\mathrm{Q_{econ} > 1.5}$ and $\mathrm{U>0.8}$.
    }
    \label{fig:Contour_PfS_Xs}
\end{figure}

Figure \ref{fig:Contour_simple_two} shows the interactions between power density and FPP areal cost and conversion efficiency. In these cases, the utilization depends only on $\mathrm{P_f/S}$, which is indicated in the top panel. The $\mathrm{Q_{econ}}$ have positive slopes in the $\mathrm{(M\$/S)_{FPP}}$ versus $\mathrm{P_f/S}$ plot, with $\mathrm{Q_{econ}}$ becoming more sensitive to $\mathrm{(M\$/S)_{FPP}}$ at higher $\mathrm{P_f/S}$. The threshold power density on the $\mathrm{Q_{econ}}=1$ contour decreases to $\mathrm{P_f/S} \sim 1$ at the lowest cost point.  The $\mathrm{c_{O/N}}$ contours have a  similar shape to the  $\mathrm{Q_{econ}}$ contours, but with a constant slope at a given $\mathrm{c_{O/N}}$. The $\mathrm{Q_{econ}}$ and $\mathrm{c_{O/N}}$ isolinear contours are convex in the $\mathrm{\eta_E}$ versus $\mathrm{P_f/S}$ space. The example $\mathrm{Q_{econ} \geq1.5}$ and $\mathrm{U\geq0.80}$ spaces provide both lower and upper limits on  $\mathrm{P_f/S}$,  an upper limit on $\mathrm{(M\$/S)_{FPP}\sim 18 \enspace M\$/m^2}$, a lower limit $\mathrm{\eta_E\sim0.3}$, and at its edges, a maximum $\mathrm{c_{O/N}\sim 5-6 \enspace \$/W}$, again similar to the cases discussed above. 

\begin{figure}[h]
    \centering
    \includegraphics[width=0.55\textwidth]{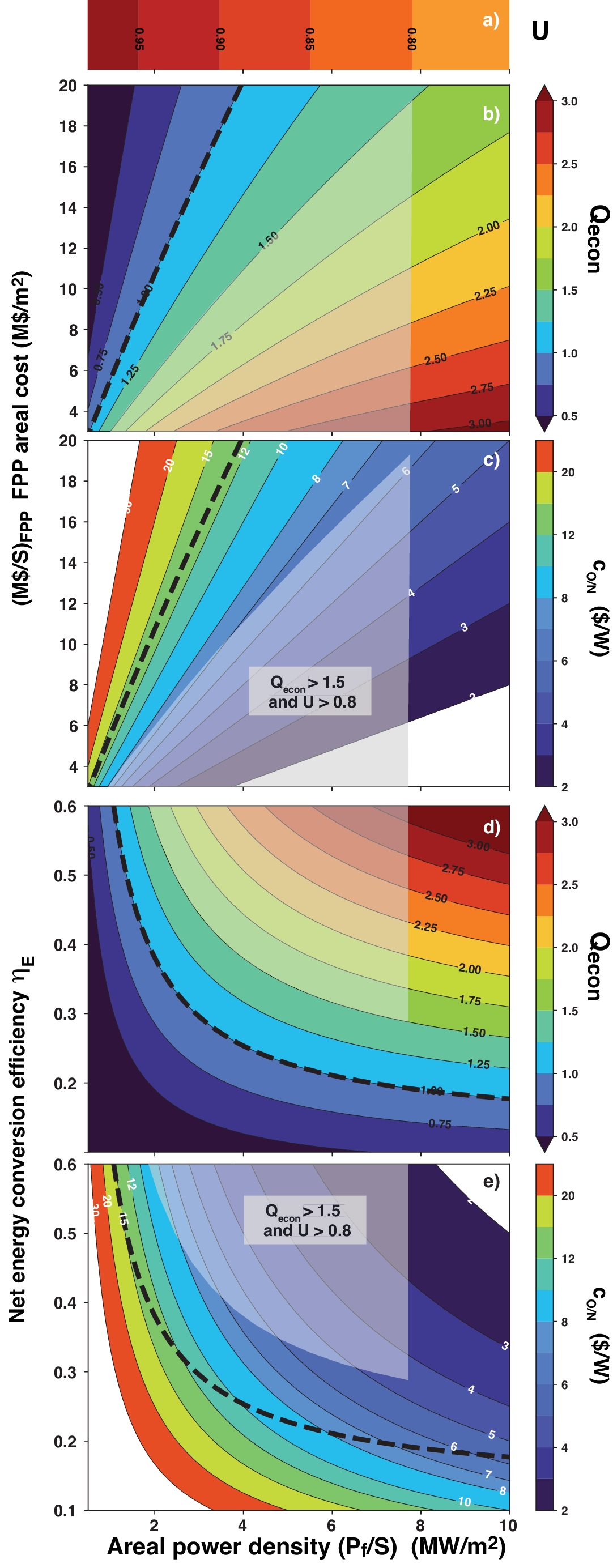}
    \caption{Economic model output varying $\mathrm{P_f/S}$ versus scanned controlling parameters (a) Utilization which only varies with $\mathrm{P_f/S}$ for these parameter scans.  Economic gain $\mathrm{Q_{econ}}$ contours and overnight cost $\mathrm{c_{O/N}}$ for (b,c) FPP areal cost and (d,e) net energy conversion efficiency $\mathrm{\eta_E}$. Shaded regions shows example  $\mathrm{Q_{econ} > 1.5}$ and $\mathrm{U>0.8}$.      
    }
    \label{fig:Contour_simple_two}
\end{figure}

Figure \ref{fig:Contour_simple_five} shows the contours of $\mathrm{Q_{econ}}$ while varying $\mathrm{P_f/S}$ and the five remaining control parameters not covered in Figures \ref{fig:Contour_PfS_taurep}-\ref{fig:Contour_simple_two}. For these controlling parameters, the  utilization and $\mathrm{c_{O/N}}$ only depend on $\mathrm{P_f/S}$, whose values are shown at the top of the figure.

The S areal cost and normalized target costs versus $\mathrm{P_f/S}$ yield identically shaped $\mathrm{Q_{econ}}$ contours. This is because mathematically they behave identically in their governing equations. This provides us with the insight that the replacement of S should be viewed like the target costs as de facto ``consumables'' in the FPP, i.e., they take the place of fuel costs in other energy systems. The implications of the S replacement cost $\mathrm{(M\$/S)_S}$ are considered in the Discussion. At high $\mathrm{P_f/S}$, the contours of $\mathrm{Q_{econ}}$ in Figure \ref{fig:Contour_simple_five} (c) and (e) become nearly horizontal, thus providing approximate threshold values. Taking the $\mathrm{Q_{econ} \geq 1.5}$ and $\mathrm{U\geq0.80}$ as a design goal, these thresholds are $\mathrm{(M\$/S)_S}\sim0.4$ and $\mathrm{(c/Y)_{target}}\sim 2\times10^{-3} \enspace \$/MJ$.  

Convex isolinear contours of $\mathrm{Q_{econ}}$ are found in both the $\mathrm{POE_{net}}$ and $\mathrm{\tau_{life}}$ versus power density plots. As previously discussed, this shape leads to effective threshold requirements for both power density, which is $\mathrm{P_f/S}\sim 1$ in the $\mathrm{Q_{econ} \geq 1.5}$ example.  The POE result reflects the critical nature of the market conditions needed to make a FPP viable; however, this result is unsurprising, since that is the case for any energy source. This criticality should motivate the designer to maximize the value of the energy product. A FPP might produce carbon-free fuel rather than electricity due to market conditions (e.g., the present POE of hydrogen from renewable energy is $\mathrm{\sim 5 \$/kg}$ or  $\mathrm{150 \$/MW \text -  h }$ in US \cite{DOEHydrogen2024}).  

The $\mathrm{Q_{econ}}$ contours are close to linear in the interest rate versus power density plots. This points to a somewhat contradictory condition for an FPP:  the power density must increase if interest rates are higher, which likely would give the FPP design and its operation greater technical risk. However, if the FPP is riskier, then that might increase the commercial interest rate imposed by the lender, thus spiraling up both parameters. This  underlies the strong desire that the first generation of FPPs have low financing costs, e.g., through public loan guarantees, which would allow less risky operations at lower power densities.

Figures \ref{fig:Contour_PfS_taurep} - \ref{fig:Contour_simple_five} illustrate the sensitivity of economic viability to power density over a very large range of design parameters.  A surprisingly consistent result is that there exists a threshold power density $\mathrm{P_f/S}\sim2 \enspace MW/m^2$ for basic economic viability $\mathrm{Q_{econ}=1}$. Of course, this threshold is dependent on the base case values chosen for the FPP; however, the base case parameters were chosen to reflect reasonable values around which to evaluate, and the results come from these scans of the solution space, not a point design.  

This threshold is significant in that the survey of commercial FPP designs shown in Table \ref{tab:FPP_parameters} in the Discussion comes to a similar conclusion regarding power density, with no FPP falling below $\mathrm{P_f/S}\sim2 \enspace MW/m^2$. However that result is arrived at independently by a different path, where FPP designs were derived ``bottom-up'' from their confinement concept, and technology-specific. This confirms the power of using a ``top-down'' economic model as developed here; while it necessarily must make simplifications, it seems to have correctly captured the quantitative design constraint of fusion power density with respect to making fusion economically viable. This also lays to rest the idea that an arbitrarily low power density is attractive because it will prolong the usable lifetime of the S and blanket components. That pathway does not have economic viability because one cannot make back enough revenue to justify the expense of building and financing the FPP. Furthermore, the top-down model provides insights to the sensitivity of the economic return across multiple design parameters in a transparent way that is not possible or practical with bottom-up design efforts. 

\begin{figure}[h]
    \centering
    \includegraphics[width=0.95\textwidth]{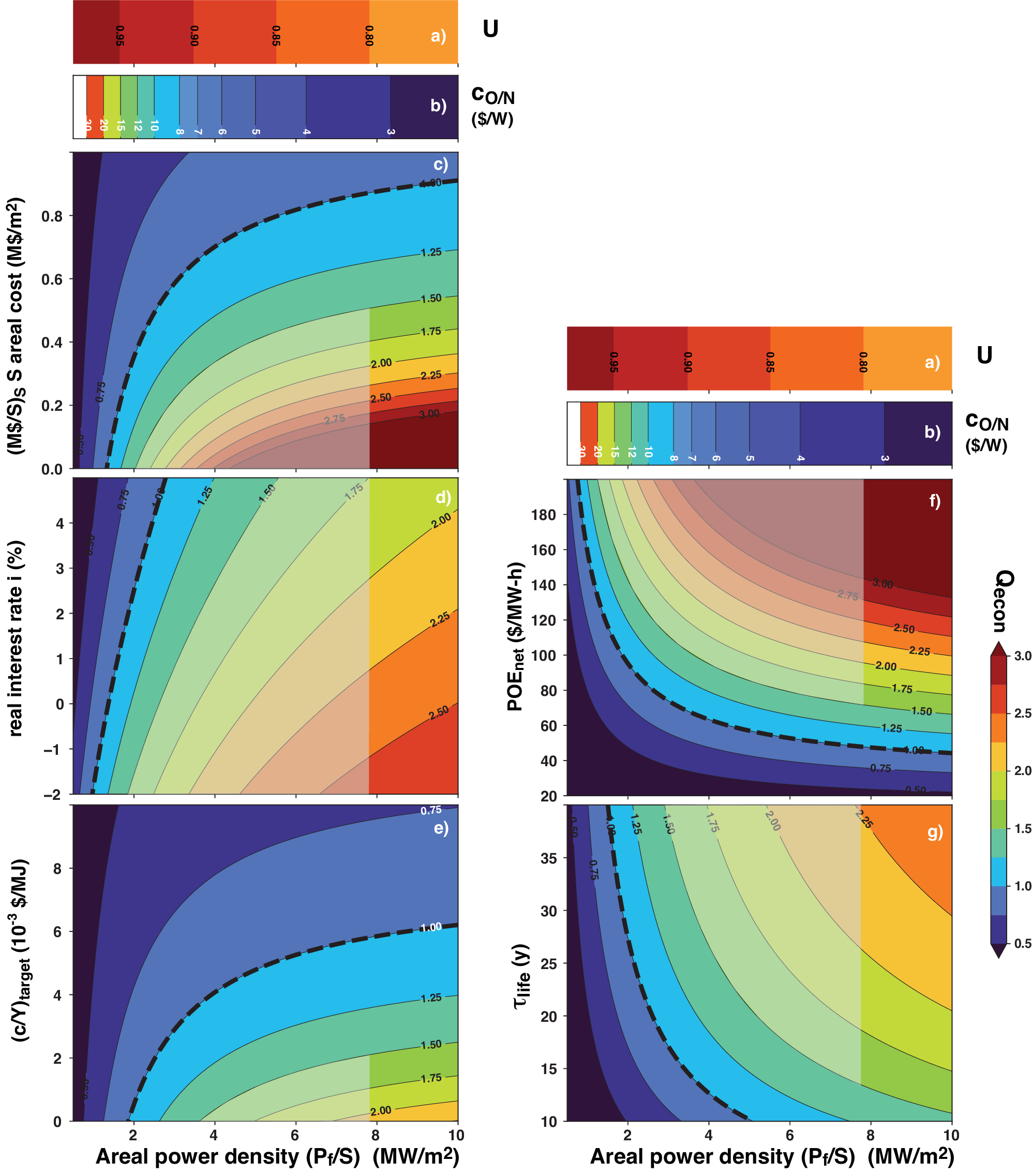}
    \caption{Economic model output varying $\mathrm{P_f/S}$ versus scanned controlling parameters (a) Utilization and (b) Overnight cost which only vary with $\mathrm{P_f/S}$ for these parameter scans.  Economic gain $\mathrm{Q_{econ}}$ contours for  (c) S areal cost (d) real interest rate (e) Target cost per yield (f) $\mathrm{POE_{net}}$ (g) FPP lifetime. See Table \ref{tab:Control_parameters} for definitions and Table \ref{tab:Default_parameters} for default parameters held fixed in scans. Shaded regions shows example  $\mathrm{Q_{econ} > 1.5}$ and $\mathrm{U>0.8}$.}
    \label{fig:Contour_simple_five}
\end{figure}

\subsection{Varying two control parameters with fixed power density }\label{sec:Sensitivity_PfS_fixed}

The interaction between two controlling parameters exclusive of power density $\mathrm{P_f/S}$ is examined by fixing $\mathrm{P_f/S}$.  For these examples, a value of $\mathrm{P_f/S=4\enspace MW/m^2}$ is chosen as representative of FPP designs (Table \ref{tab:FPP_parameters}), and as a value that generally meets economic viability from  $\mathrm{P_f/S}$ sensitivity scans of the previous section. Unvarying control parameters are fixed at their base case values listed in Table \ref{tab:Default_parameters}. 

The first set of comparisons focus on the design parameters of S and its associated blanket: the energy fluence limit $\mathrm{X_S}$, the net energy conversion efficiency $\mathrm{\eta_E}$ (linked to the blanket thermal efficiency), the S replacement cost $\mathrm{(M\$/S)_S}$ and the S replacement time $\mathrm{\tau_{rep}}$. 

\begin{figure}[h]
    \centering
    \includegraphics[width=1.0\textwidth]{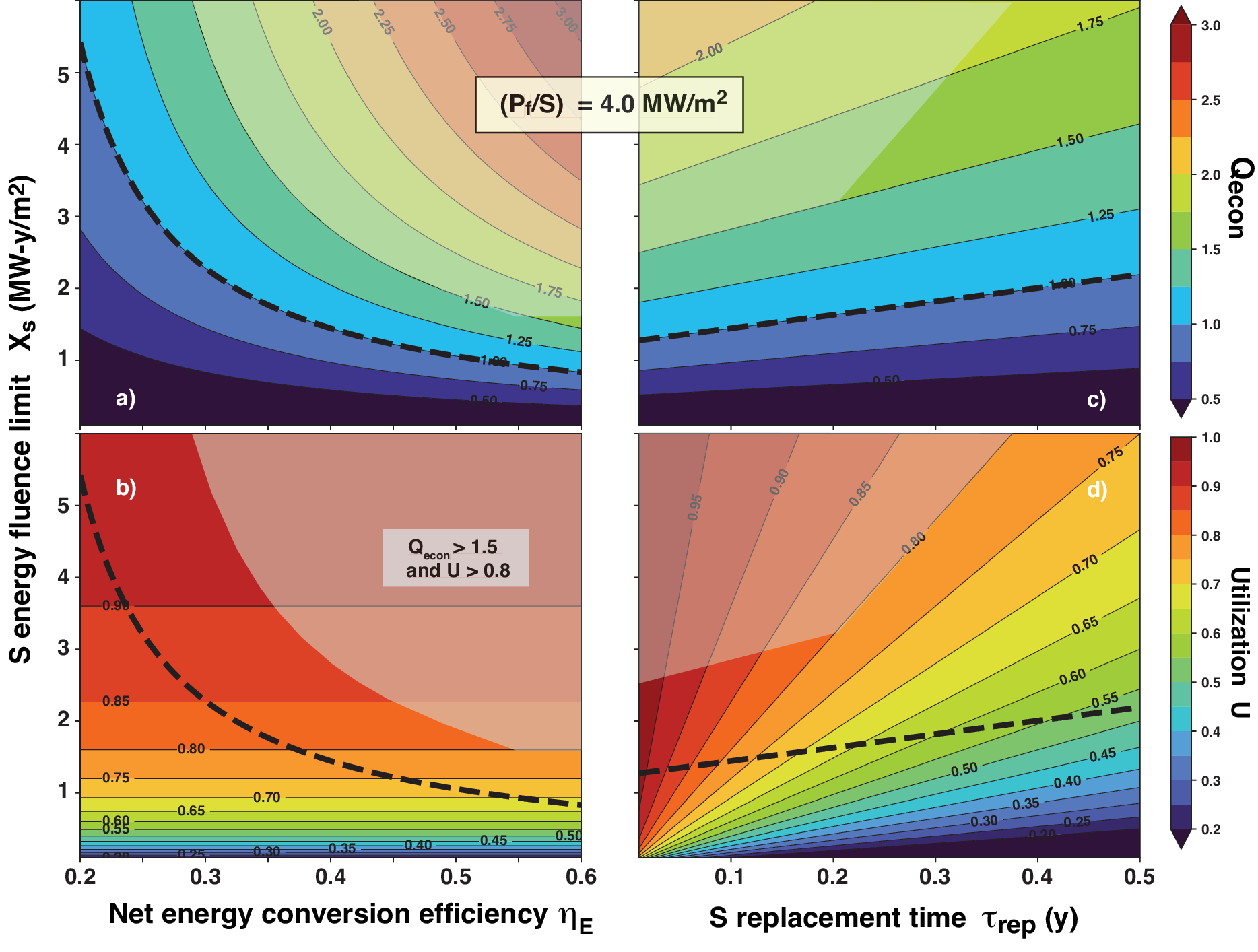}
    \caption{Economic model output with fixed $\mathrm{P_f/S=4 \enspace MW/m^2}$ and varying two control parameters. Energy fluence limit $\mathrm{X_S}$ on vertical axis versus (a-b) $\mathrm{\eta_E}$ and (c-d) $\mathrm{\tau_{rep}}$.  (a, c) Economic gain $\mathrm{Q_{econ}}$ (b,d) Utilization, with the $\mathrm{Q_{econ}=1}$ contour overlaid. Shaded regions shows example  $\mathrm{Q_{econ} > 1.5}$ and $\mathrm{U>0.8}$.}
    \label{fig:Contour_fixedPfS_Xs_yaxis}
\end{figure}

The $\mathrm{Q_{econ}}$ contours are convex in the $\mathrm{\eta_E}$ versus $\mathrm{X_S}$ space (Figure \ref{fig:Contour_fixedPfS_Xs_yaxis}).
There is a threshold value $\mathrm{X_S}\sim1$, but only a marginal increase in $\mathrm{Q_{econ}}$ for $\mathrm{X_S} \gtrsim 4-5 $, mostly because the utilization is surpassing 0.9. For the design goal example of $\mathrm{Q_{econ} \geq 1.5}$ and $\mathrm{U\geq0.80}$ used in Section \ref{sec:Sensitivity_PfS_scan}, the lower corner has thresholds of $\mathrm{X_S}\sim2$ and $\mathrm{\eta_E}\sim0.3$. The $\mathrm{Q_{econ}}$ and U contours are linear in the $\mathrm{\tau_{rep}}$ versus $\mathrm{X_S}$ plot, with a threshold value of $\mathrm{X_S}\sim1$ as $\mathrm{\tau_{rep}\rightarrow0}$.  The $\mathrm{Q_{econ}}$ contours show weak sensitivity at higher values of $\mathrm{X_S}$ and low $\mathrm{\tau_{rep}}$. In the example design goal, a threshold of $\mathrm{\tau_{rep} \lesssim 0.3 \enspace y}$ arises.  One notes similarities in the $\mathrm{X_S}$ threshold values across different controlling parameters.

\begin{figure}[h]
    \centering
    \includegraphics[width=1.0\textwidth]{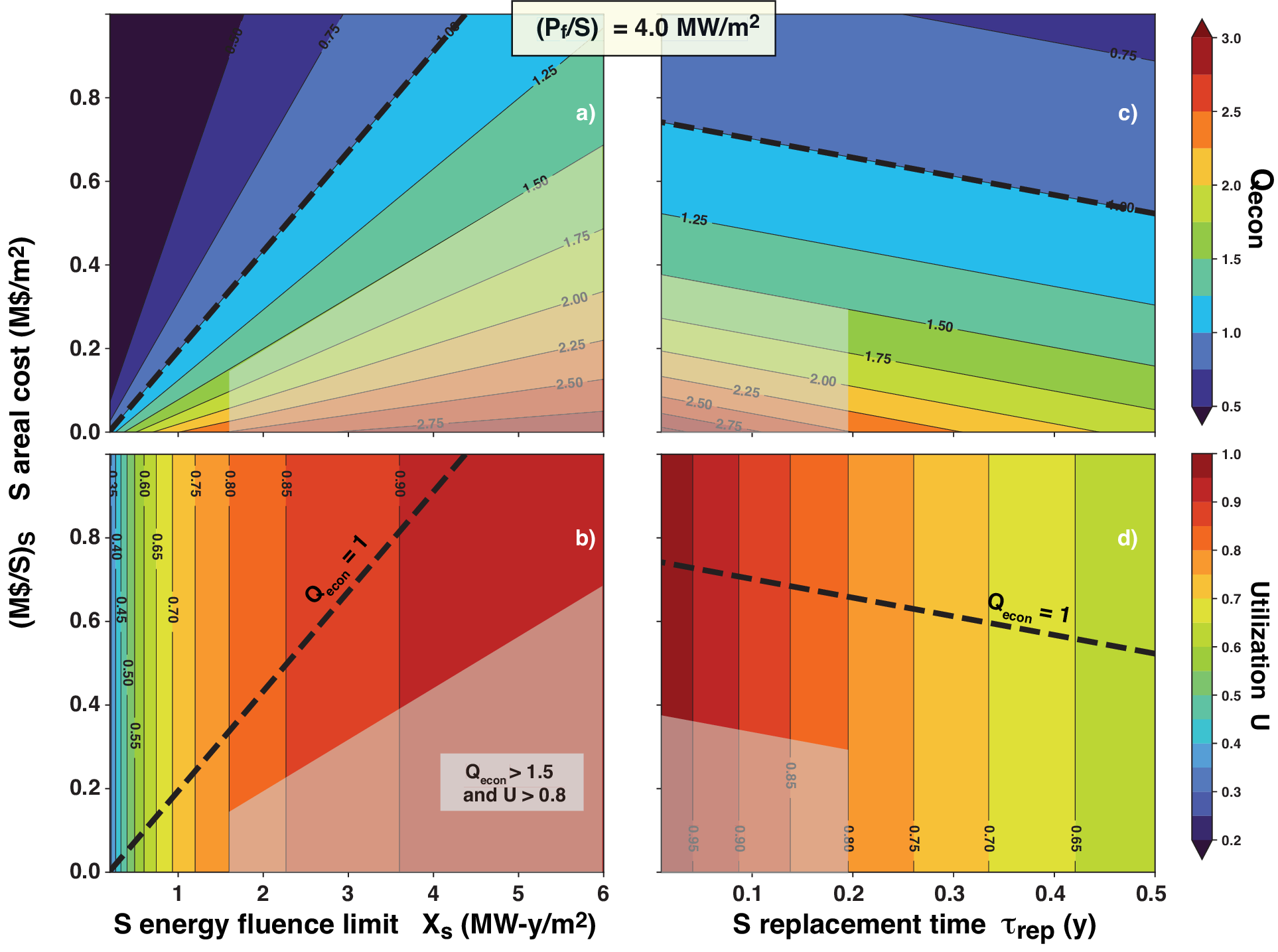}
    \caption{Economic model output with fixed $\mathrm{P_f/S=4\enspace MW/m^2}$ and varying two control parameters. Replacement S areal cost $\mathrm{(M\$/S)_S}$ on vertical axis versus (a-b) $\mathrm{X_S}$ and (c-d) $\mathrm{\tau_{rep}}$.  (a, c) Economic gain $\mathrm{Q_{econ}}$ (b,d) Utilization, with the $\mathrm{Q_{econ}=1}$ contour overlaid with dashed black line. Shaded regions shows example  $\mathrm{Q_{econ} > 1.5}$ and $\mathrm{U>0.8}$.}
    \label{fig:Contour_fixedPfS_ArealcostS}
\end{figure}

The $\mathrm{Q_{econ}}$ contours are linear in the $\mathrm{(M\$/S)_S}$ versus $\mathrm{X_S}$ space (Figure \ref{fig:Contour_fixedPfS_ArealcostS}). At increased $\mathrm{X_S}$ and higher $\mathrm{Q_{econ}}$, the contours become more horizontal, so that $\mathrm{(M\$/S)_S}$ is the more important factor on economic return. However, the utilization depends only on $\mathrm{X_S}$. Thus, in the design space example $\mathrm{Q_{econ} \geq 1.5}$ and $\mathrm{U\geq0.80}$, there are clear thresholds in both parameters: $\mathrm{X_S} \geq 1.5$ and $\mathrm{(M\$/S)_S} \lesssim 0.5 $.   The more horizontal shape of the design space (i.e., the shaded area) further indicates that cost effective S fabrication is more important than high energy fluence limits, as long as the threshold $\mathrm{X_S}$ value is met. The $\mathrm{Q_{econ}}$ contours are inversely linear in the $\mathrm{(M\$/S)_S}$ versus $\mathrm{\tau_{rep}}$ space. This leads to the most severe threshold restrictions, with $\mathrm{(M\$/S)_S} \leq 0.3 $ and $\mathrm{\tau_{rep} \leq 0.2}$. Thus, one conclusion from these scans of parameter space is that a less expensive and rapid replacement of S is more critical to the FPP design than very large S energy fluence limits.

This trend is further confirmed in Figure \ref{fig:Contour_fixedPfS_Xs_xaxis}, where $\mathrm{X_S}$ is varied against two other key costing parameters of FPP areal cost, $\mathrm{(M\$/S)_{FPP}}$ and $\mathrm{POE_{net}}$. Here, the utilization only varies with $\mathrm{X_S}$. The $\mathrm{Q_{econ}}$ contours increase with $\mathrm{X_S}$ but are nonlinear, exhibiting a minimum threshold requirement and then increasing with diminishing marginal economic gains at higher $\mathrm{X_S}$ due to the utilization surpassing 0.9 (note that $\mathrm{\tau_{rep}}$ is fixed at the base case value of 0.1 year). In the design space example of $\mathrm{Q_{econ} \geq 1.5}$ and $\mathrm{U\geq0.80}$, there are clear thresholds at both $\mathrm{X_S} \geq 1.5$ and $\mathrm{(M\$/S)_{FPP}} \lesssim 14$. While the overnight cost only depends on $\mathrm{(M\$/S)_{FPP}}$, the design space shows an accessible $\mathrm{c_{O/N} \lesssim9 \enspace \$/W }$.The $\mathrm{Q_{econ}}$ contours are convex in the $\mathrm{POE_{net}}$ versus $\mathrm{X_S}$ plot, and again $\mathrm{Q_{econ}}$ stops increasing strongly when $\mathrm{X_S}\geq 4$. The design space example also yields clear thresholds $\mathrm{X_S} \geq 1.5$ and $\mathrm{\mathrm{POE_{net}} \gtrsim 80 \enspace \$/MW \text -  h}$.

\begin{figure}[h]
    \centering
    \includegraphics[width=0.6\textwidth]{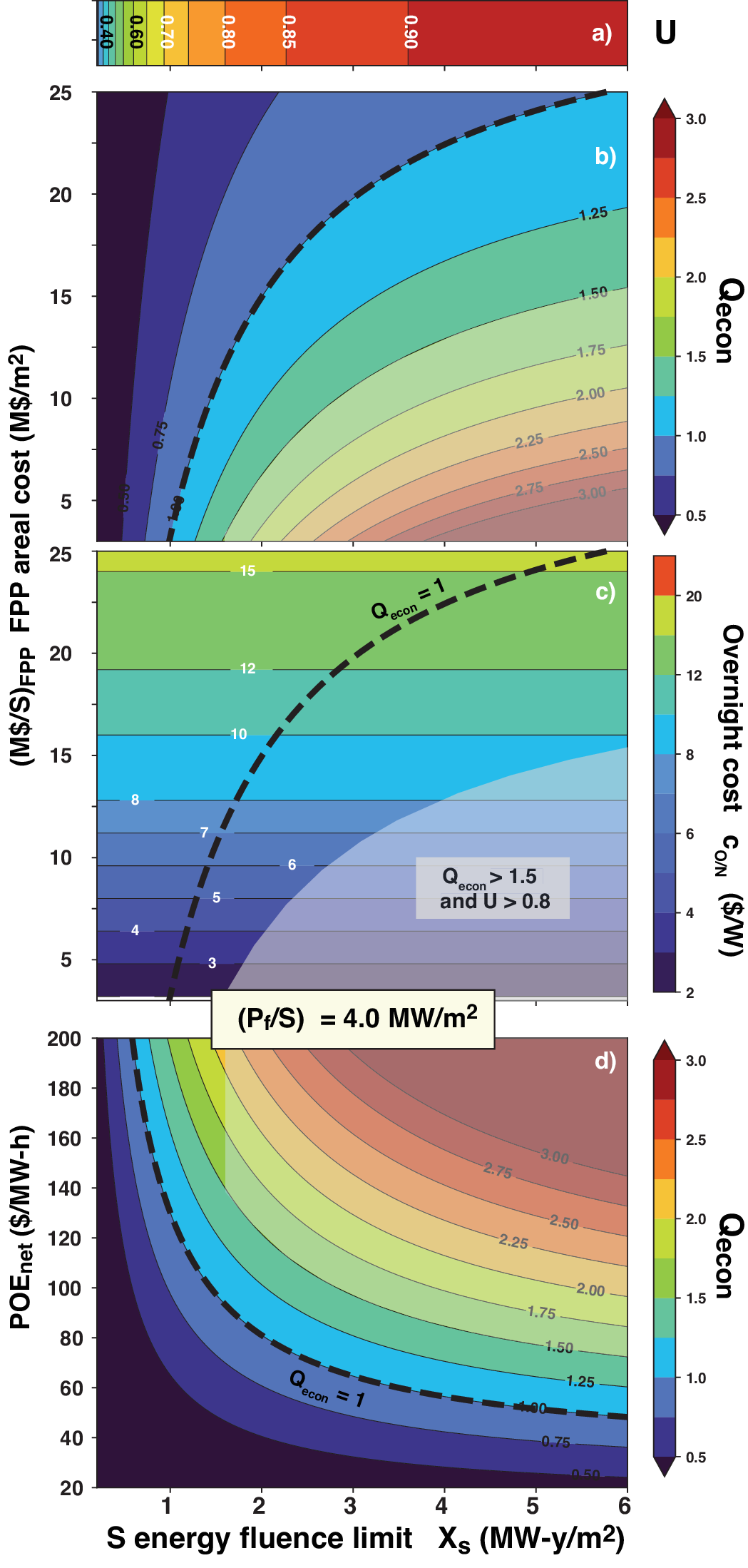}
    \caption{Economic model output with fixed $\mathrm{P_f/S=4\enspace MW/m^2}$ and varying two control parameters. Energy fluence limit $\mathrm{X_S}$ on horizontal axis. (a) Utilization which only varies as $\mathrm{X_S}$ (b) Economic gain $\mathrm{Q_{econ}}$ and (c) overnight cost $\mathrm{c_{O/N}}$ with varying FPP areal cost  (d) Economic gain $\mathrm{Q_{econ}}$ with varying $\mathrm{POE_{net}}$. Shaded regions shows example  $\mathrm{Q_{econ} > 1.5}$ and $\mathrm{U>0.8}$. }
    \label{fig:Contour_fixedPfS_Xs_xaxis}
\end{figure}

\clearpage

\section{Results: General assessments of economic viability}\label{sec:general_assessment}

A fusion power plant is economically self-sustaining only if $\Qecon \ge 1$ (break-even or better). In practice, one would want $\Qecon$ comfortably above 1 to justify investment (providing profit margin for the operator and its investors, and accounting for other fixed costs or risks). $\Qecon < 1$ implies economic loss in steady state which is not sustainable. The combination of parameters must be balanced to achieve $\Qecon \ge 1$, which is complicated by their nonlinear dependence on each other. For example, lowering $\mathrm{P_{f}/S}$ puts pressure on $\mathrm{X_{S}}$ (Figure \ref{fig:Contour_PfS_Xs}), or a lower $\mathrm{\eta_{E}}$ might be tolerable if $\mathrm{POE_{net}}$ is very high, etc. The utility of our economic viability criterion $\Qecon \geq 1$ is that it highlights these trade-offs clearly. A fusion plant design must simultaneously satisfy a whole set of criteria and if any one factor is too weak, the others have to compensate accordingly.  

These tradeoffs were examined in the previous section where the economic gain space is evaluated in detail with two control parameters being varied, so we now consider more generally how to access $\Qecon \ge 1$.
In more general terms we want to separate the 10-dimensional space of controlling parameter values into two distinct regions, those that yield economically viable FPPs, and those that do not. To that end, we first define the \emph{feasible ranges} of each of the 10 controlling parameters as the set of values feasible solely due to physical and mathematical admissibility: for example, efficiencies must lie in $(0,1)$, costs and lifetimes must be nonnegative, and the price of energy must be strictly positive. These ranges are deliberately broad and include values that may be unrealistic or unattainable in practice. They are essential for studying the mathematical geometry of the hypersurface $\mathrm{Q_{econ}}=1$ across the entire admissible domain of parameters. In contrast, we can also define \emph{plausible ranges} of parameter values, typically narrower intervals around base case design assumptions for FPPs, based on engineering projections and economic precedent. They provide a more realistic basis for sensitivity analysis and design trade-offs. Considering both sets of ranges highlights the difference between exploring theoretical properties of the model and assessing practical design viability. Plausible ranges for the 10 controlling parameters are given in Table \ref{tab:Default_parameters}.

\subsection{Viable and non-viable regions}
\label{ssec:viable_nonviable}
Denote by $\theta \equiv [\, P_f/S,\ \tau_{\rm rep},\ \cdots,\ i \,]$ the vector of all 10 controlling parameters (Table~\ref{tab:Control_parameters}), which is an element of the feasible parameter space $\Theta$, defined as the Cartesian product of all parameter domains admissible solely from physical and mathematical principles:
\begin{equation}
\begin{split}
\Theta \equiv (0,\infty)\times[0,\infty)\times(0,\infty)\times(0,\infty)\times(0,\infty)\times
(0,1)\times  \\
[0,\infty)\times[0,\infty)\times[0,\infty)\times(-100,\infty),
\end{split}
\end{equation}
where $\mathrm{i}$ is expressed in percent, so the lower bound $\mathrm{i>-100}$ ensures that $\mathrm{1+0.01\,i>0}$ in the annuity factor.
Within the feasible set $\Theta$, the plausible set $P$ is defined as the subset corresponding to engineering and economic expectations for NOAK fusion plants as given in Tables~\ref{tab:Control_parameters} and~\ref{tab:Default_parameters}:
\begin{equation}
P \equiv \{\theta \in \Theta : P_f/S \in [0.5,10],\ \ldots,\ i \in [-2,5]\} .
\end{equation}

For the log-coordinate results below, and for any numerical closest-viable-design optimization, we further restrict attention to a strictly positive box-constrained subset
\begin{equation}
\Theta_{+} \equiv \prod_{k=1}^{10} [\underline{\theta}_k,\overline{\theta}_k] \subset \Theta,
\qquad
0 < \underline{\theta}_k \le \overline{\theta}_k < \infty .
\end{equation}
This distinction is important because the transformation $\ell=\log\theta$ is only defined on the strictly positive interior. In practice, lower bounds such as $\mathrm{(c/Y)_{\rm target}\ge 10^{-6}}$ or $\tau_{\rm rep}\ge 10^{-4}$ are numerically indistinguishable from zero while keeping the optimization on the interior where the log-log-concavity results of the Supplement apply. Sensitivity scans may still be discussed on the broader feasible set $\Theta$, but the theorem-level log-coordinate statements below concern $\Theta_{+}$.

Let $Q_{\mathrm{econ}}:\boldsymbol{\Theta}\to \mathbb{R}$ denote the economic viability criterion and define the viable region
\begin{equation}
\mathcal{V} =
\{\, \btheta \in \boldsymbol{\Theta} \;|\; Q_{\mathrm{econ}}(\btheta) \geq 1 \,\}
\end{equation}
and the infeasible region
\begin{equation}
\mathcal{I} = 
\{\, \btheta \in \boldsymbol{\Theta} \;|\; Q_{\mathrm{econ}}(\btheta) < 1 \,\}\ \hspace{3mm}.
\end{equation}
The economically viable region $\mathcal{V}$ is essentially the set of all combinations of these parameters that yield at least break-even economics. This region has a complex shape due to the nonlinear interplay of parameters in the $\Qecon$ formula. We can think of the boundary $\Qecon = 1$ as a kind of hypersurface in this 10-D space: 
\begin{equation}
\partial \mathcal{V} =
\{\, \btheta \in \boldsymbol{\Theta} \;|\; Q_{\mathrm{econ}}(\btheta) = 1 \,\}\ \hspace{3mm}.
\end{equation}
One side of this hypersurface corresponds to profitable designs ($\Qecon\ge 1$) and the other side to unprofitable ones ($\Qecon<1$).

\subsection{Geometric properties of $\mathcal{V}$}
\label{ssec:geometric}
\textbf{\textit{Connectedness}}: The feasible region is expected to be continuous and (largely) connected. Small changes in a parameter produce small changes in $\Qecon$ because the underlying equations are continuous. If a design is just barely viable ($\Qecon\approx 1$), a slight improvement in any favorable direction (e.g. a bit higher efficiency or lower cost) will yield $\Qecon>1$ (still viable), and a slight degradation will yield $\Qecon<1$ (just infeasible). There are no isolated ``pockets'' of viability separated by gaps; rather, there is one continuous region, provided all parameters remain in physically meaningful ranges. The complement (non-viable region $\mathcal{I}$) is likewise continuous---essentially it is the set of points below the $\Qecon=1$ hypersurface. For example, starting from a viable point and gradually worsening one parameter (e.g., slowly raise the cost or lower the power density), will move the FPP design continuously into the infeasible side once the threshold $\delta\mathcal{V}$ is crossed. There is no discrete jump; it is a smooth boundary crossing where net profit transitions from positive to negative.

\textbf{\textit{Non-Convexity}}: The feasible region is not strictly convex in all ten dimensions, due to the nonlinear nature of the $\Qecon$ constraint. In a convex region, any linear interpolation between two viable design points would also be viable. Here, that is not guaranteed, because improving one parameter can compensate for worsening another in a highly nonlinear way. However, in many 2-dimensional slices of the design space, the boundary does exhibit a convex-like shape. For instance, if we plot a two-parameter trade-off like 
$\mathrm{P_{f}/S}$ versus $\mathrm{X_{S}}$ (power density vs. fluence limit) holding other parameters fixed, the contours of constant $\Qecon$---often referred to as ``isoquants'' in the economics literature---are convex curves (see Figure \ref{fig:Contour_PfS_Xs} in the preceding section). This convex isoquant implies, for example, that to maintain $\Qecon=1$ some increase in one parameter can be compensated by a decrease in another in a smooth fashion rather than non-monotonically. In economic terms, these isoquants may be viewed as ``indifference curves''---another borrowed concept from economics involving the loci of combinations of parameters that yield the same value for $\Qecon$ (see Discussion). The region above the contour $\Qecon=1$ (better in all parameters) would satisfy $\Qecon>1$, and tends to look convex in that local projection. Nevertheless, considering all 10 parameters at once, the viable set $\mathcal{V}$ is defined by a nonlinear inequality and can have some curved boundaries and trade-off surfaces. It is not a simple polyhedron or hypercube, but more of a warped multi-dimensional volume.

\textbf{\textit{Log-Log-Concavity:}} Although $\mathcal{V}$ is not convex with respect to arithmetic (linear) combinations of parameters, a stronger structural property holds on the strictly positive box-constrained domain $\Theta_{+}$. Define $\ell \equiv \log \theta$ componentwise on $\Theta_{+}$. It is shown in the Supplement \cite{supplement} that $\log \Qecon$ is a concave function of $\ell$ on $\Theta_{+}$, i.e.\ $\Qecon$ is log-log-concave on the strictly positive interior. The proof exploits the fact that, after clearing the utilization factor, $\Qecon$ has the form of a monomial divided by a sum of log-log-convex terms, together with the log-log-convexity of the annuity factor $\phi(\mathrm{i,\tau_{\rm life})}$. The immediate geometric consequence is that every super-level set
\[
\{\theta \in \Theta_{+} : \Qecon(\theta)\ge q^\ast\}
\]
is convex in log-coordinates: on the domain actually used for the weighted closest-viable-design optimization, there are no disconnected pockets of viability, no re-entrant corners, and no spurious feasible regions. This convexity is not a statement about the full boundary-inclusive feasible set $\Theta$; it applies to the strictly positive subset on which the log transformation is well-defined.

\textbf{\textit{Boundary Properties}}: One striking feature of $\mathcal{V}$ is the presence of thresholds or ``cliffs'' in certain directions. Because some parameters must exceed minimum values for viability, the $\Qecon=1$ surface often lies near those thresholds. For example, there may be a minimum required $\mathrm{P_{f}/S}$ and $\mathrm{X_{S}}$ such that below those values, no solution exists. In the $\mathrm{P_{f}/S}$ vs $\mathrm{X_{S}}$ plane, this manifests as a steep corner---if both power density and fluence limit are too low, $\Qecon$ falls off sharply. In these regions, the $\Qecon$ surface is almost like a cliff: small deviations can lead to non-viability. Above the cliff, however, improvements yield diminishing returns---once in the viable region, changing a parameter further helps less and less. For instance, increasing $\mathrm{X_{S}}$ beyond a certain point (so components last extremely long) might not increase $\Qecon$ much if the plant is already running at $\mathrm{>90\%}$ utilization. The feasible region thus has a kind of ``corner'' or knee in many 2-D projections: one must first climb up to a threshold, after which further improvements flatten out in value. This implies that the boundary $\Qecon=1$ is often curved and has high curvature in some areas---it is not a flat plane through the space, but rather an irregular surface, bent and steep in some directions (near thresholds) and flatter in others (where additional margin exists).

\textbf{\textit{Boundedness}}: In principle, $\mathcal{V}$ is unbounded in several directions. Improving some parameters without limit will ensure economic viability, for example decreasing FPP costs to zero.  Practically, many parameters have natural limits: efficiencies cannot exceed 100\%, power density cannot be arbitrarily high, etc. If we confine attention to realistic ranges as given in Table \ref{tab:Default_parameters}, the viable region is effectively bounded by these physical/market limits. But within these bounds, $\mathcal{V}$ typically occupies a substantial volume if all parameters are near their favorable limits. For example, a combination of high
$\mathrm{P_{f}/S}$, high $\mathrm{X_{S}}$, low cost, etc., will be deep inside the viable region. Conversely, the complement $\mathcal{I}$ extends toward the opposite extreme: e.g. as we approach very low power density, very short component life, very high cost, etc., $\Qecon$ plummets well below 1. The volume of $\mathcal{I}$ is generally much larger than that of $\mathcal{V}$, which is another way of stating the obvious: achieving economic fusion is hard.

In summary, the economically viable design space for FPPs is a single contiguous region bounded by a nonlinear, curved hypersurface defined by $\Qecon=1$. This surface is not a simple shape but can be visualized through 2-D slices as a set of curves that often look convex and have clear threshold boundaries (see previous section). Within the region, $\Qecon > 1$ (and especially $\gg1$) indicates increasingly profitable designs. The complement is all other points---generally less ``extreme'' in performance---which yield $\Qecon<1$ and thus are economically non-viable. Importantly, because of the monotonic influence of each parameter (improving any one, holding others fixed, will never reduce $\Qecon$), i.e. starting with a viable design and improving any parameter further will yield viable designs, and vice versa. This monotonicity is why plotting iso-$\Qecon$ ``indifference'' curves is useful---it illustrates trade-offs between pairs of parameters.

\subsection{Closest viable design}
\label{ssec:closestviable}

Because the 10 controlling parameters carry heterogeneous units (MW/m$^2$, years, \%, \$/MW-h, M\$/m$^2$, etc.), a standard unweighted Euclidean norm is dimensionally ill-posed: a unit change in one parameter is incommensurable with a unit change in another. We therefore formulate the projection using a diagonal weighted norm
\[
\|\theta-\theta_o\|_W^2 \equiv (\theta-\theta_o)^\top W(\theta-\theta_o),
\qquad
W=\mathrm{diag}(w_1,\ldots,w_{10}),\quad w_i>0.
\]
Each weight admits a conceptual decomposition
\begin{equation}
w_i = \frac{\hat w_i}{s_i^2},
\label{eq:closest-viable}
\end{equation}
where $s_i$ is a scale factor that renders $(\Delta\theta_i/s_i)^2$ dimensionless (e.g.\ the width of the plausible range from Table~2), and $\hat w_i>0$ is a dimensionless difficulty factor reflecting the relative cost or difficulty of changing parameter $i$. In implementation the two roles are combined into a single $w_i$, but the decomposition clarifies that the norm is dimensionally consistent by construction and that the difficulty factors are modeling inputs.

For illustrative weighted calculations, we adopt the convention
\[
s_i = \theta_{i,\max}^{\rm plausible}-\theta_{i,\min}^{\rm plausible},
\qquad
\hat w_i = 1 \quad \text{for all } i,
\]
unless otherwise stated, so that
\[
w_i = \frac{1}{\bigl(\theta_{i,\max}^{\rm plausible}-\theta_{i,\min}^{\rm plausible}\bigr)^2}.
\]
Alternative choices of $\hat w_i$ can then be used to encode stakeholder-specific judgments about relative difficulty.

The closest viable design is the solution to
\begin{equation}
\min_{\theta\in\Theta_{+}}
\sum_{i=1}^{10} w_i(\theta_i-\theta_{o,i})^2
\qquad
\text{subject to}
\qquad
\Qecon(\theta)\ge 1 .
 \label{eq:closest-viable}
\end{equation}

The solution to this optimization problem is essentially a weighted least-squares projection of $\theta_o$ onto the viable region within the box-constrained positive domain. It can be shown that:
\begin{enumerate}
\item \textbf{(Existence)} If $\Theta_{+}\cap V$ is nonempty, then the weighted distance
\[
d_W(\theta_o,\Theta_{+}\cap V) = \inf_{\theta\in\Theta_{+}\cap V}\|\theta-\theta_o\|_W
\]
is attained at some $\theta^\star \in \Theta_{+}\cap V$.

\item \textbf{(First-order optimality)} If $\Qecon$ is continuously differentiable in a neighborhood of $\theta^\star$ and the active constraint is regular, then there exists $\lambda^\star \ge 0$ such that
\begin{equation}
2W(\theta^\star-\theta_o)=\lambda^\star \nabla \Qecon(\theta^\star),
\qquad
\Qecon(\theta^\star)=1 .
 \label{eq:kkt}
\end{equation}
Equivalently, the weighted displacement $W(\theta^\star-\theta_o)$ is normal to the hypersurface $\Qecon=1$ at $\theta^\star$.

\item \textbf{(Uniqueness)} Because $\Qecon$ is log-log-concave on $\Theta_{+}$ (Section~\ref{ssec:geometric}), the viability frontier is convex in log-coordinates on that domain. If $\Theta_{+}\cap V$ is nonempty and the lower bounds satisfy $\underline{\theta}_i > \theta_{o,i}/2$ for all $i$, then the weighted optimization problem has a unique global minimizer for any choice of positive diagonal weights (see Supplement, Proposition~6.1). Numerical verification (SLSQP from 50 random initializations) confirms convergence to the same solution to 8 significant figures.
\end{enumerate}

The optimization framework also offers practical information regarding the most effective way to achieve viability: parameters with large gradient-to-weight ratios $\frac{|\partial \Qecon/\partial \theta_i|}{w_i}$ should be adjusted the most, because they offer the best cost-to-impact ratio at the viability boundary.

For example, if $\Qecon$ is most sensitive to $\mathrm{P_f/S}$, the weighted projection may primarily involve increasing the areal power density. If the design falls short mainly because of a high cost of capital, a reduction in $i$ through loan guarantees or related policy support may be the more efficient adjustment. The key concept is that the model provides a weighted notion of nearest approach to economic viability in the 10-dimensional control space.

It is important to distinguish the roles played by the formula and
the weights in this optimization. The formula for
$\mathrm{Q_{{econ}}}$ determines the \emph{viable
region}---the set of parameter combinations achieving
$\mathrm{Q_{econ} \geq 1}$. The weights do not change this
region; they determine which point on its boundary is selected as
``closest.''  This dependence is not a defect but the point of
the decision problem, because the 10 controlling parameters span
fundamentally different categories of effort and agency.
\emph{Engineering parameters} ($\mathrm{P_f/S}$, $\mathrm{X_S}$, $\mathrm{\eta_E}$)
can only be improved through sustained physics and materials R\&D;
increasing the areal fusion power density requires advances in
plasma confinement or magnet technology, extending the energy
fluence limit requires developing radiation-tolerant materials, and
improving net conversion efficiency requires blanket and thermal
cycle innovation.
\emph{Market and financing parameters} ($\mathrm{POE_{net}}$,
$i$) are set by energy markets, central bank policy, the project's
credit profile, and macroeconomic conditions; the FPP designer does
not control them directly, though they can be influenced indirectly
through government loan guarantees, contracts-for-difference, or
capacity payments.
\emph{Construction and replacement parameters}
($\mathrm{({\rm M\$}/S)_S}$, $(\mathrm{{\rm M\$}/S)_{\mathrm{FPP}}}$,
$\tau_{\mathrm{rep}}$, $\tau_{\mathrm{life}}$) depend on
industrial-economic factors such as supply chain maturity,
manufacturing learning curves, construction logistics, and
maintenance technology.

These categories are not interchangeable. Asking a fusion power
plant designer to ``reduce the interest rate by 2 percentage
points'' is a categorically different directive than ``increase
the power density by 1~MW/m$^2$.''  A uniform weighting would
conflate these entirely different categories of effort; the
difficulty weights $\hat{w}_i$ encode the asymmetries in effort
and agency.

Indeed, the isoquant analysis of Section~\ref{sec:extract}
already does this implicitly: the ``lines of constant differential
technology risk'' with slope $\mathrm{dX_S/d(P_f/S) = -2}$ in Figure~10
are precisely a ratio of difficulty weights, expressing the
judgment that, at the margin, a 1~MW-y/m$^2$ improvement in
fluence limit is twice as costly as a 1~MW/m$^2$ improvement in
power density. The weighted closest-viable-design optimization generalizes this tangency-slope concept from pairwise parameter tradeoffs to the full 10-dimensional space, and from graphical inspection to a constrained optimization problem that, provided the feasible set is nonempty and the lower bounds satisfy $\underline{\theta}_i > \theta_{o,i}/2$ for all $i$, has a unique global minimizer for any choice of positive diagonal weights (see Supplement, Proposition~6.1).

Different stakeholders would reasonably assign different weights.
A plasma physicist might weight power density improvements as
relatively cheap (low $\hat{w}_i$) and financing improvements as
expensive (high $\hat{w}_i$), reflecting the view that the physics
is solvable but the capital markets are hard to influence.
A financial engineer might take the opposite view: public loan
guarantees and contracts-for-difference are proven policy tools,
while achieving 6~MW/m$^2$ power density is an unsolved physics
problem. Each would obtain a different closest viable
design---a different pathway to viability---and the comparison of
these pathways is itself informative, revealing which parameters
offer the most economic leverage.

Also, the weighted-Euclidean projection identifies the smallest weighted
portfolio of parameter changes needed to reach viability which is an
\emph{endpoint} metric, not a literal R\&D trajectory through
design space. (The isoquant and tangency-point analysis in
Figures~10--11 is closer to a trajectory concept, but even there
the analysis is comparative-static rather than truly dynamic.)

The log-log-concavity of $\mathrm{Q_{econ}}$ ensures that the
``cliffs'' visible in plots of $\mathrm{Q_{econ}}$ versus
individual parameters (e.g., the steep rise at low $\mathrm{P_f/S}$) do not
create non-convexities in the function's \emph{level sets}.
The viability contour is a smooth convex curve in log-coordinates
with no inlets or peninsulas, guaranteeing that the endpoint
calculation is well-posed regardless of the landscape's steepness.
When the optimum is interior to the box constraints, the KKT
stationarity condition (Equation~\ref{eq:kkt}) shows that
parameters with large $\
\frac{|\partial Q_{\mathrm{econ}}/\partial
\theta_i|}{w_i}$ are adjusted the
most---the optimization preferentially changes the parameters
offering the best cost-to-impact ratio. This is exactly the
economic information that the sensitivity analysis
(Figures~2--8) provides qualitatively; the weighted optimization
formalizes it.

The differences in relative difficulty of making improvements in
the 10 controlling parameters suggests that a more realistic set of
solutions to the optimization problem may be achieved by weighting
the parameters according to their importance or difficulty.
Practical implementations also impose lower and upper bounds on
parameters within $\Theta$ (e.g., $\mathrm{\eta_E \leq 1 - \epsilon}$)
to avoid degenerate optima on open boundaries.

\subsection{Technology tradeoffs via isoquants}
\label{sec:extract}

As quantified in the preceding sections, a compelling feature of the economic model is that it provides insights to local sensitivities over a wide range of controlling parameters. Many of the controlling parameters incorporate aspects of technical performance, and therefore aspects of technical risk. Thus the model can link the cost of improving technical performance to the economic benefit of that performance, as well as other interactions between different design parameters as a FPP design moves into regions of higher $\mathrm{Q_{econ}}$.

An example of this model feature is shown in Figure \ref{fig:tangent_example}, which recasts the results from Figure \ref{fig:Contour_PfS_Xs}, isolating discrete contours of $\mathrm{Q_{econ}}$ in the $\mathrm{X_S}$ versus $\mathrm{P_f/S}$ space.  These ``isoquants" of $\mathrm{Q_{econ}}$ are convex in shape, which suggests treating them as ``lines of indifference'', i.e., a graphical representation in economics that shows all combinations of two goods, and therefore costs, that provide a consumer with the same level of utility or satisfaction \cite{pareto2014manual}.  However in our case the axes are not ``goods'' per se but rather represent the cost of technical improvement or risk.  Visual examination of the curves indicate that there are regions of the isoquants where increased technical risk become unjustified because the $\mathrm{Q_{econ}}$ saturate (noted by stars on the figure).  It is intuitive that to move upwards to higher $\mathrm{Q_{econ}}$ the most efficient means is by accessing the lower left corners of the isoquants.
The weighted closest-viable-design optimization of Section~\ref{ssec:closestviable} generalizes this tangency-slope analysis from pairwise parameter tradeoffs to the full 10-dimensional control space, and from graphical inspection to a constrained optimization problem that, under the box-constrained strictly positive assumptions stated in Section~\ref{ssec:closestviable}, has a unique minimizer in log-coordinates (see Supplement, Proposition~6.1).

\begin{figure}[h]
    \centering
    \includegraphics[width=0.75\textwidth]{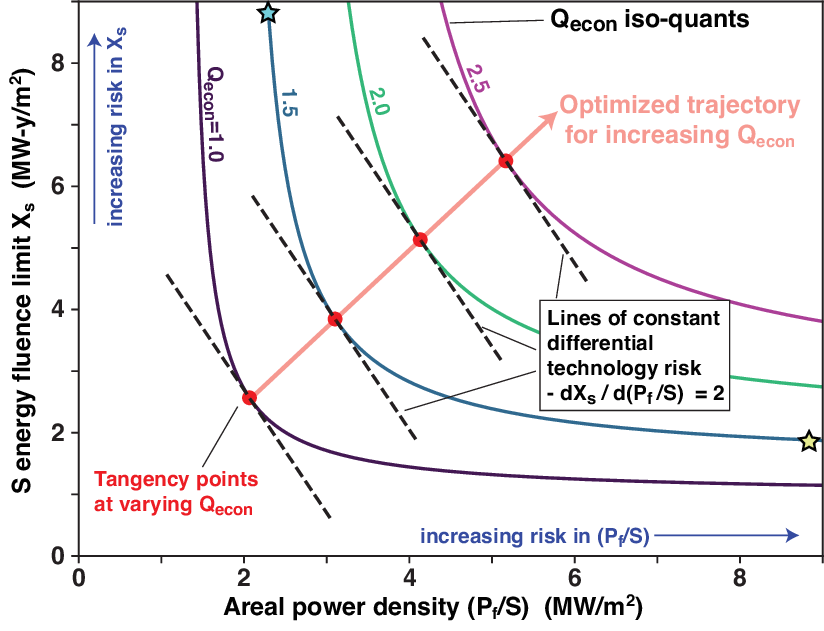}
    \caption{Economic model output as in Figure \ref{fig:Contour_PfS_Xs} isolating discrete contours of constant $\mathrm{Q_{econ}}$ or ``isoquants''. The stars on the $\mathrm{Q_{econ} =1.5}$ line indicate where improvements in 
    $\mathrm{X_S}$ (blue star) or improvements in $\mathrm{P_f/S}$ (yellow star) are becoming ineffective at increasing $\mathrm{Q_{econ}}$ and thus the increased technical risk associated with those improvement is unjustified.  Lines of constant differential technology risk (with assumed slope of -2) provide unique tangency points at each iso-quant, which are the most efficient points to access a design target $\mathrm{Q_{econ}}$.
    }
    \label{fig:tangent_example}
\end{figure}

These curves are reminiscent of the intuitions provided by the Lawson criterion where curves of constant $\mathrm{Q_p}$ are convex-like in the T versus $\mathrm{n\cdot \tau_E}$ space \cite{wurzel2022progress}.  Achievement of a target $\mathrm{Q_p}$ is typically optimized by targeting the lower-left corner of the contour, since both T and $\mathrm{n\cdot \tau_E}$ are difficult to achieve and come with a design ``cost'' (e.g., larger device, larger energy driver, etc.).  An extension of this concept, which quantifies this intuition, is widely used in magnetic fusion; the Plasma OPeration CONtour (POPCON) \cite{houlberg1982contour}. A POPCON, which also assumes temporal equilibrium, allows one to scope the most efficient means of accessing a targeted design performance of fusion power and $\mathrm{Q_p}$.  With knowledge of the design costs involved POPCONs can be used then to maximize efficiency of achieving a fusion design target, for example the least expensive combination of external heating power and device size, which each have their own cost penalty in an MFE FPP design. Thus the 2-D contour representations shown here can be thought of as an ``economic POPCON''.  It is interesting to note that while Lawson and POPCONs have significant physics simplifications, they remain widely used as design tools \cite{rutherford2024manta, creely2020overview} in fusion \emph{because} of their flexibility, simplicity and transparency.

\begin{figure}[h]
    \centering
    \includegraphics[width=0.75\textwidth]{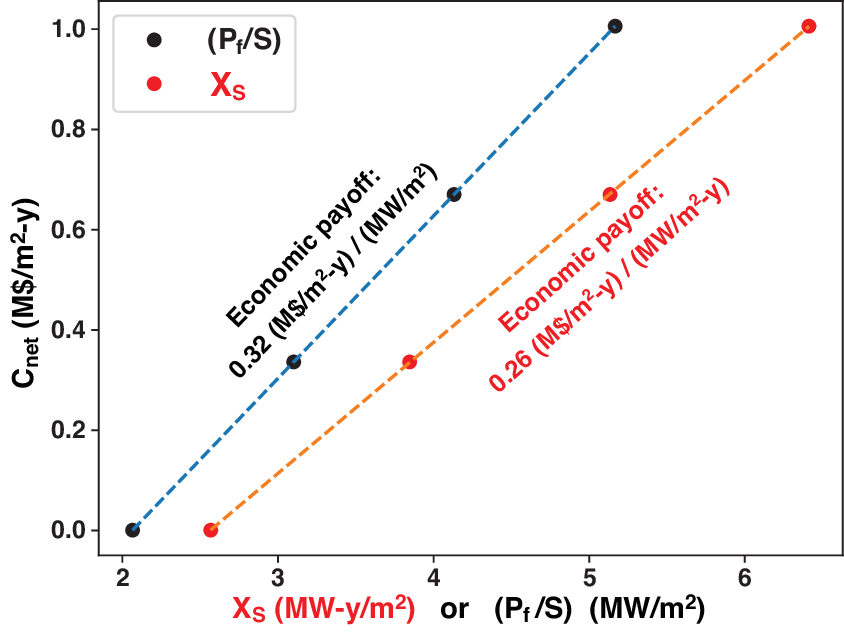}
    \caption{Tangency points from increasing $\mathrm{Q_{econ}}$ isoquants of Figure \ref{fig:tangent_example} showing resulting net economic gain rate $\mathrm{C_{net}}$ versus increments in fusion power density (black points) and S energy fluence limit $\mathrm{X_S}$ (red points). The fitted slopes indicate the optimized economic payoff rate with fixed differential technology risk. 
    }
    \label{fig:tangent_payoff}
\end{figure}

The isoquants of $\mathrm{Q_{econ}}$ can also quantify economic FPP design efficiency by providing an assessment of the different technology risk to access a target performance. This requires an engineering assessment that determines the relative difficulty (or development cost) to comparatively increasing one design parameters at the cost of another. In the example case of Figure \ref{fig:tangent_example} this could be linked to the consequences of increasing energy fluence limit of S (to increment $\mathrm{X_S}$) at the price of decreased heat removal capacity in S that decrements $\mathrm{P_f/S}$. This could occur for example in the choice of a material, or thickness of material, at S that increases its energy fluence limit but with lower heat conductivity.    An actual assessment of this risk is beyond the scope of this work, but here for an example we have assumed a differential technology risk $\mathrm{dX_S/d(P_f/S)=-2}$.  Due to their convex shape there are unique tangency points on the isoquants which match the differential technology risk, and these tangency points can be identified at varying $\mathrm{Q_{econ}}$ contours.  Figure \ref{fig:tangent_payoff} shows the results of plotting these tangency points versus the net economic return $\mathrm{C_{net}}$ (Equation \ref{eq:C_net_definition}) for the two controlling parameters. The slopes of these tangency points then provide an economic payoff rate gained by the technology improvements for the FPP. Not only does this payoff inform the optimized path to deal with technical risk, it also informs the FPP developer on the worth of the development costs. 

\subsection{Monte Carlo distributions of $\mathrm{Q_{econ}}$ and U}\label{sec:ROI distributions}

We are motivated to understand the impact of variance in the controlling parameters on the economic viability.
 By assigning probability  or confidence intervals to the input parameters listed in Table \ref{tab:Control_parameters}, the input parameters can be statistically distributed within those bounds using Monte Carlo techniques, and the resulting probability distribution functions of model outcomes, $\mathrm{Q_{econ}}$ and U,  recorded.  This will be able to capture some of the effects of variance, however each simulation run still uses the simplifications of operational and market equilibrium.

An illustrative set of distributions are shown in Figure \ref{fig:Cnet_distribution_example}, with their model inputs noted in the caption.  Figure \ref{fig:Cnet_distribution_example} (a-c) shows the results of a normal distribution applied to the areal power density $\mathrm{P_f/S}$ with an assigned mean value $\mathrm{\mu =4.0 \enspace MW/m^2}$, a standard deviation $\mathrm{\sigma = 0.4 \enspace MW/m^2}$ and a coefficient of variation $\mathrm{CV \equiv \sigma/\mu=0.1} $.  Due to nonlinearities and offsets in the model, the resulting economic outcomes can have a different CV than the initial distribution;  the utilization has a lower CV =0.01/0.89 $\simeq$ 0.01, as does the $\mathrm{Q_{econ}}$ CV =0.08/1.57 $\simeq$ 0.05, and it is also negatively skewed. The CV of $\Qecon$ provides a simple measure of relative dispersion in economic outcomes. A larger $CV$ indicates more variability relative to the mean, and is similar in spirit to, but distinct from, the Sharpe ratio \cite{sharpe1994sharpe}, which indicates the distribution of risk to average reward in investments, a critical parameter in financing and investment decisions. 

\begin{figure}[h]
    \centering
    \includegraphics[width=1.0\textwidth]{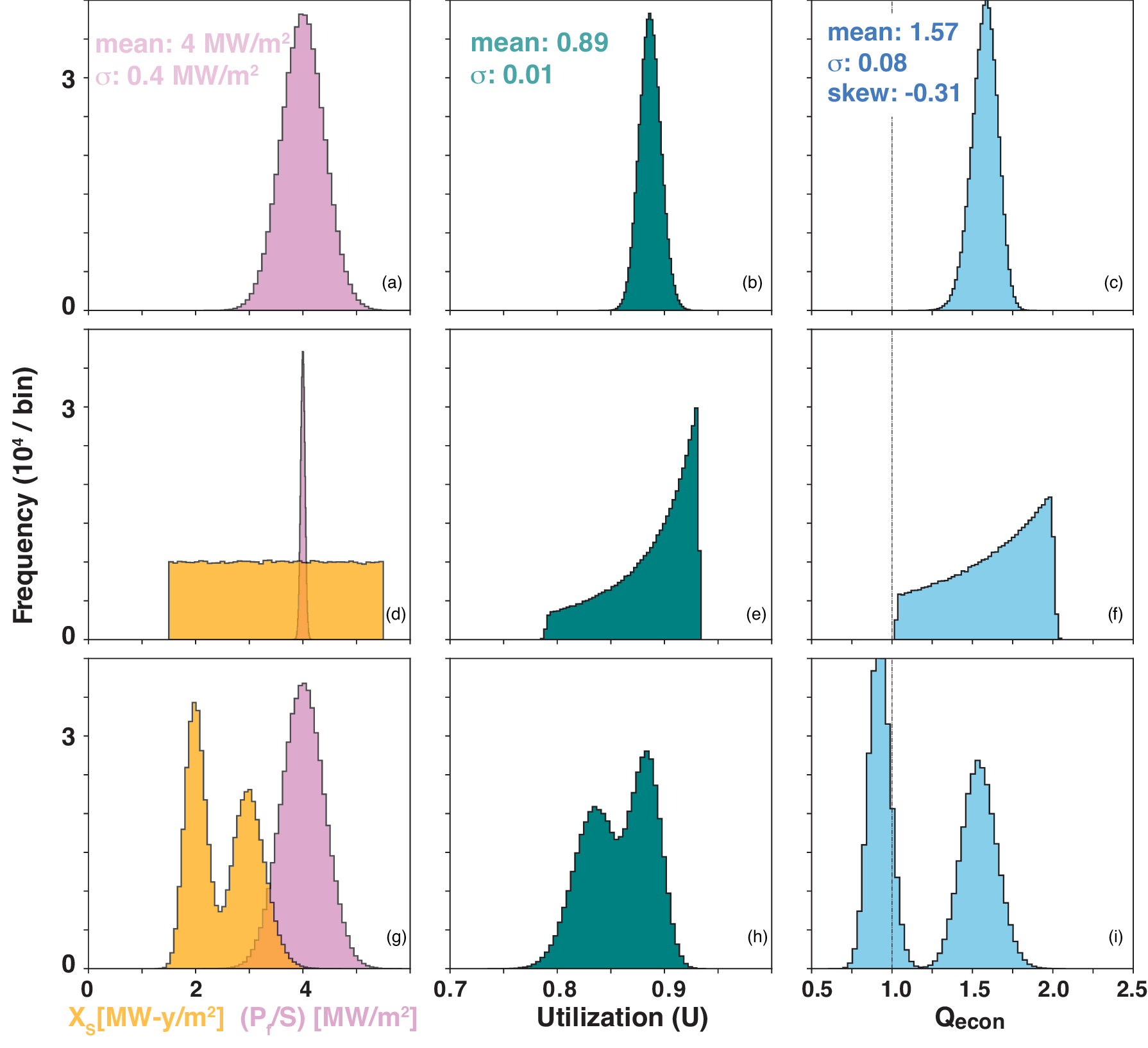}
    \caption{Examples of using the economic gain model to determine distributions of outcomes.  Except as noted the input parameters as stated in Table \ref{tab:Default_parameters}. Left column: distribution of input parameters of energy fluence limit $\mathrm{X_{S}}$ [orange] and/or fusion power density $\mathrm{P_f/S}$ [violet]. Middle and right columns: calculated distributions of utilization and $\mathrm{Q_{econ}}$ using $\mathrm{10^5}$ samples. (a-c) Input normal distribution of $\mathrm{P_f/S}$  (other inputs fixed). Fitted mean and standard deviation values noted in panels (d-f) Narrow $\mathrm{P_f/S}$ input distribution (mean: 4, $\mathrm{\sigma}$=0.04) and $\mathrm{X_{S}}$ uniformly varied from 1.5 - 5.5 $\mathrm{MW \text -  y/m^2}$. (g-i) Normal distribution of $\mathrm{P_f/S}$ as in (a-c), with bimodal distribution of $\mathrm{X_{S}}$ due to two distinct operating temperatures of equal probability: left peak (lower T, $\mathrm{\eta_E}$=0.30, $\mathrm{X_{S}}$ mean:2.0, $\mathrm{\sigma}$=0.2), right peak (higher T, $\mathrm{\eta_E}$=0.4, $\mathrm{X_{S}}$ mean:3.0, $\mathrm{\sigma}$=0.3).
    }
    \label{fig:Cnet_distribution_example}
\end{figure}

In the second example, given in Figure \ref{fig:Cnet_distribution_example} (d-f), the $\mathrm{P_f/S}$ distribution is narrowed ($\mathrm{\mu}$ =4.0, CV=0.01) while the $\mathrm{X_S}$ is uniformly distributed between 2.5 and 5.5, indicating a large range of uncertainty in the quality of knowledge for the energy fluence limit of S. The resulting U and  $Q_{econ}$  are both highly positively skewed. Since the lower U and $\mathrm{Q_{econ}}$ values arise from the lower part of the $\mathrm{X_S}$ distribution, this suggests the $\mathrm{Q_{econ}}$ ``value'' from obtaining high $\mathrm{X_S}$ is not uniform. 

In the third example, given in Figure \ref{fig:Cnet_distribution_example} (g-i), the normal $\mathrm{P_f/S}$ distribution is applied as in (a-c), but the operating temperature is taken as a binary decision to operate at high or low temperature for S, with a 50\% probability assigned to each. The operating temperature is assumed to vary the two input parameters simultaneously. First, the energy conversion efficiency $\mathrm{\eta_E}$ is partially given by its Carnot efficiency, which switches between precisely 0.3 and 0.4 at low and high temperatures, respectively. (The $\mathrm{\eta_E}$ can be accurately calculated based on the energy conversion cycle.) Second, the higher operating temperature is assumed to increase the average $\mathrm{X_S}$ due to thermal annealing effects, switching from 2.0 to 3.0, where in either case a CV=0.1 normal distribution is assigned, assuming that there will still be uncertainty in the fluence limit. The resulting U and $\mathrm{Q_{econ}}$ feature bimodal distributions. It is noteworthy that for $\mathrm{Q_{econ}}$, significant fractions of the distribution are below unity, which is unacceptable from an economic viewpoint . 

These $\mathrm{Q_{econ}}$ and U distributions indicate that, even in this simplified example, there are complex tradeoffs and interactions in the physics and engineering design in an FPP that must be considered for understanding the probabilities of economic return in an FPP.  In addition the mathematical simplicity of the model allows one to perform a wide range of statistical analysis of how the input parameters move the $\mathrm{Q_{econ}}$ contour space.

\section{Discussion}\label{sec:Discussion}

\subsection{Considerations regarding FPP areal cost}\label{sec:Areal_cost_appendix}

The FPP normalized areal cost, $\mathrm{(M\$/m^2)_{FPP}}$ in 2025 \$, have large uncertainty due to the simple fact that an FPP has not yet been constructed and operated.  The devices that have been built close to FPP scale and/or are under construction, are first-of-a-kind devices with an emphasis on confinement research, not fusion energy production. Furthermore the organization and financing structures of fusion programs makes it difficult to extract accurate costing with wide variations in in-kind contributions, development costs, site credits, etc. A starting point is a survey of confinement devices that more closely approach FPP conditions. In D-T magnetic fusion energy (MFE), JET and TFTR were $\mathrm{\sim \$ 1000 M } $ projects \cite{JET_WIKI,meade1997affordable} which approached $\mathrm{Q_p \sim 1}$ with $\mathrm{S \sim 100 \hspace{2mm} m^2}$ and so resulting in $\mathrm{(M\$/m^2)_{FPP} \sim 10}$ \cite{romanelli2015overview,hawryluk1998fusion}; however these lacked any blanket or energy extraction, and did not use the superconducting magnets required for magnetic fusion FPPs.  The non-DT superconducting MFE stellarators W7-X and LHD are of similar scale with $\mathrm{S \sim 100  \enspace m^2}$ \cite{wolf2016wendelstein,motojima2000progress}and an estimated development and construction cost $\mathrm{\$1200 M}$ and $\mathrm{\$950 M}$ respectively \cite{W7X_WIKI,LHD_Japan_Times} $ \mathrm{\rightarrow (M\$/m^2)_{FPP} \sim 10}$. In inertial fusion the National Ignition Facility has achieved $\mathrm{Q_p > 1}$ in D-T, with target chamber $\mathrm{S \sim 100 \hspace{2mm} m^2}$ and cost of $\mathrm{\sim \$ 3000 M } $ and $\mathrm{(M\$/m^2)_{FPP} \sim 30}$, again without a blanket or energy extraction. For NIF it should be noted this cost included substantial technology development costs of the laser drivers, optics and targets.  Thus the conclusion is that FPP-adjacent fusion experimental devices have normalized costs that have order of magnitude $\mathrm{(M\$/m^2)_{FPP} \sim 10}$. 

Moving to proto-commercial projects, ITER and SPARC are $\mathrm{Q_p \sim 10}$ (Table \ref{tab:FPP_parameters}) D-T MFE tokamak devices presently under construction but differ greatly in physical scale due to their  magnetic field strengths from use of varying superconductors: ITER $\mathrm{S \sim 800 \enspace m^2}$ at 5.4 tesla using $\mathrm{NbSn_3}$ superconductors with $\mathrm{P_f/S\sim0.7 \enspace MW/m^2} $, and SPARC $\mathrm{\sim 50 \enspace m^2}$ at $\mathrm{P_f/S\sim3 \enspace MW/m^2}$ at 12 tesla using high-field REBCO superconductors.  The cost of ITER is disputed \cite{kramer2018ITER} with estimates ranging from \$25,000 - 60,000 M. Using an intermediate costs provided by the ITER organization following a re-baseline of the design $\sim$ \$33,000M  leads to $\mathrm{(M\$/m^2)_{FPP} \sim 45}$.  SPARC's cost is reported at $\mathrm{\sim \$ 660 M}$ or $\mathrm{(M\$/m^2)_{FPP} \sim 15}$.  Both ITER and SPARC  have a blanket/shield to deal with the large flux of D-T neutrons, but neither will extract energy for commercialization.   The expectation is that as first-of-a-kind devices the normalized costs is elevated by the R+D development, costs associated with initiating fabrication capacity, and the other ``first-encounter'' issues with fusion power at a commercial scale; thus normalized costs should decrease for a mature FPP compared to these. 

\begin{table}[ht]
\caption{Sample FPP-class D-T fusion device costs and power performance.  Costs inflation adjusted to 2025 from cost statement year. If S area is not provided the plasma surface area is used.
\newline $^*$ ITER organization stated project cost in 2018 \cite{kramer2022ITER} and additional 5B\$ from 2024 re-baseline \cite{Matthews2024ITER}
\newline $^+$ SPARC stated costs by CFS \cite{kramer2021investors}
\newline $^\odot$ direct costs 
\newline $^\#$ target chamber dimensions containing molten salt flow
\newline $^\&$ from MANTA \cite{rutherford2024manta} study with similar tokamak dimensions
\newline $^\%$ from lower cost range in \cite{anklam2011life}
\label{tab:FPP_parameters}
}
\centering

\begin{tabular}{|c||c||c||c||c||c||c||c||c|}
\hline
Name/type                      &   $S$  &   Cost      &    $(M\$/S)_{FPP}$ &  $P_f$  &  $P_e$  & $\eta_E$  & $P_f/S$  & $c_{O/N}$ \\
                            &   $m^2$ &  $2025 M\$$ &    $M \$/m^2$      &   $MW$   &   $MW$    &           &  $MW/m^2$   &  $\$/W$  \\
\hline \hline
ITER \cite{ikeda2007progress} &  750  &  $\sim$ 33,000 $^*$ &  $\sim$ 45    & 500    &  -    &  -      &  0.7           & - \\
\hline
SPARC \cite{creely2020overview} & 45   &  $\sim$ 660$^+$     &  $\sim$ 15    & 140    &  -    &  -      &  3.1           & - \\
\hline
ARIES-RS \cite{Overview_ARIES-RS} & 420 &  5050 $^\odot$      &  12.0    & 2170  &  1000    &  0.46     &  5.2           & 5.1 \\
\hline
ARIES-ST \cite{Overview_ARIES-ST} & 580 &  5330  $^\odot$     &  9.2    & 2980  &  1000    &  0.34     &  5.1           & 5.3 \\
\hline
ARIES-AT \cite{Overview_ARIES-AT} & 460 &  3500  $^\odot$    &  7.7    & 1719  &  1000    &  0.58     &  3.8           & 3.5 \\
\hline
ARIES-CS \cite{Overview_ARIES-CS} & 750 &  4500  $^\odot$   &  6.0    & 2440  &  1000    &  0.41     &  3.3          & 4.5 \\
\hline
HYLIFE2 \cite{moir1994hylife} &  190 $^\#$ &  2750 $^\odot$     &  14.7    & 2100  &  940    &  0.45     &  11.1          & 2.9 \\
\hline
LIFE.2 \cite{anklam2011life} &  390  &  5740 $^\%$     &  14.5    & 2200  &  940    &  0.45     &  5.6          & 5.7 \\
\hline
ARC-V2E \cite{CreelySOFE_2025} &  285 &  3400 $^\&$     &  11.9  & 1130  &  400    &  0.35     &  4.0          & 8.5 \\
\hline
Stellaris \cite{lion2025stellaris} &  940 &  -            &  -    & 2700  &  1000    &  0.37     &  2.9          & - \\
\hline
\end{tabular}
\end{table}

An aside is offered regarding the history of fusion costing. All the examples above point to the challenge of accurately and/or openly reporting fusion costs. Despite the construction of many fusion relevant experiments around the world, the fusion research community has largely ignored developing agreed-upon accounting and costing standards. While this may seem unnecessary for scientific exploration, the fact is that fusion research has always been motivated as a practical energy source, and therefore as fusion projects more closely approach those of an FPP, it is unfortunate that there is ``self-obfuscation'' as to cost effectiveness. It would be unacceptable if there are factors of 2-3 reporting uncertainty in fusion science results, so it is unclear why it has been accepted such inaccuracy  is allowed for costing.  Finding documented costs of historic fusion devices is difficult on the public record, with the majority of the above costings above pulled from new articles or online sources such as Wikipedia, rather than technical publications. For ITER, the largest fusion experiment built to date, the very nature of its international agreement guarantees  that costs will never be known accurately as admitted by their spokesperson ``that ITER doesn’t provide an official estimate of construction costs because the participating countries have different methods of pricing out their in-kind contributions—mostly in the form of fabricated reactor components—and those estimates are not reported to the ITER Organization'' \cite{kramer2018ITER}. This seems an unacceptable answer for a device being built with public funds, but worse yet, disallows for effective evaluation of cost which must be a critical criterion for fusion economics. Conversely, SPARC is being constructed by a private company Commonwealth Fusion Systems funded by private capital.  In this case the company will be obligated to report accurately the full development and construction costs to their shareholders; indeed this will be the case with any private sector fusion developer. This is a welcome and necessary development for developing fusion energy. However private companies have no obligation to report publicly their detailed costs, and in fact for competitive reasons have strong disincentives to provide details of their supply chain and manufacturing. Of course as the FPP technology approaches pilot stages and market readiness, the customers will require accurate costing, as they would with any large energy source purchase. 

Commercial D-T FPP design studies are another source for costing estimates.  The US-based ARIES studies (see Table \ref{tab:FPP_parameters} for parameters and references) are insightful because they compare across multiple MFE concepts (tokamaks at varying aspect ratios, and a stellarator in ARIES-CS) but keep the net electric power fixed at 1000 MW.  The ARIES studies provide ``bottom-up'' physics and engineering designs, with specific choices of MFE configuration and fusion technology and cost estimates based on same, but uniformly sought to minimize LCOE. This makes them a useful comparison to the ``top-down'' model developed here which is impartial to specific configurations and technology but imposes high-level economic constraints with highly variable performance metrics. Direct costs are used in Table \ref{tab:FPP_parameters} from these studies as the best proxy to FPP costs. The ARIES areal FPP costs range from 6 - 12 $\mathrm{M\$/m^2}$ and the overnight normalized costs from 3.5 - 5.1 \$/W, which seems a fairly small variation given the very significant differences in their configurations and technologies.  A common feature is normalized power density $\mathrm{P_f/S > 3 \enspace MW/m^2}$ which is commented to in Section {\ref{sec:results_two_params}} as a consistent feature from the sensitivity studies on economic viability.

FPP studies for IFE concepts are included with HYLIFE-2  and LIFE.2 which are ion and laser driver concepts respectively with FPP costs $\sim 15 \enspace \mathrm{M\$/m^2}$ and the overnight normalized costs from 5.7 - 8.5 \$/W, and power density $\mathrm{P_f/S > 5 \enspace MW/m^2}$. The HYLIFE design stands out for high power density because it used a flowing molten salt wall, which permitted higher average power removal than solid surfaces. 

A recent version of an FPP MFE tokamak is ARC from Commonwealth Fusion Systems, using high-field REBCO magnets as does SPARC,  about half size in terms of S and electric power compared to ARIES designs,  but has similar costs and power when normalized to S or net electrical power.  In a similar vein a recent stellarator design Stellaris from Proxima Fusion uses REBCO high-B magnets and while costing information is not available, one notes the similar fusion power density $\mathrm{P_f/S \geq 3 \enspace MW/m^2}$ as to the other FPP designs.  

From these FPP studies (excluding ITER and SPARC, which are experimental devices) we conclude that a reasonable base case value of FPP normalized cost is $\mathrm{(M\$/m^2)_{FPP} \sim 10}$.  Furthermore it appears that normalizing to S is generally reasonable since we recover similar costing and performance despite the large range of confinement and technology choices applied in these FPP studies.

\subsection{Considerations regarding S replacement}\label{sec:S_cost}

A key result of the sensitivity analysis is that the economic performance and sufficient utilization are sensitive to the replacement cost and replacement time of S (c.f. Figures \ref{fig:Contour_simple_five} and \ref{fig:Contour_fixedPfS_ArealcostS}).  However, since no working FPP extraction surface S or blanket has ever been deployed or replaced, there is a challenge to giving these numerical results an appropriate context.

Regarding the replacement time, the approximate result of the scans through parameter space is that $\mathrm{\tau_{rep} \leq 0.25 \enspace y}$ will be an FPP requirement for robust economic performance.  This can be contrasted to the ITER project, where a changeout of the 740 blanket shield modules will require 2 years \cite{haange1999remote} using remote handling. Even though ITER is not an FPP design (it does not produce an energy product), this highlights that significant conceptual and technology changes must be adopted for FPPs. For example, these modifications in magnetic fusion energy would be horizontal vacuum vessel splitting \cite{CreelySOFE_2025} and demountable coils \cite{rutherford2024manta}, and in inertial fusion energy, a continually replenished S via flowing liquid \cite{moir1994hylife}.  Whatever the S replacement concept, however, it is clear that a time-efficient S replacement must become a high priority for the FPP designer from the outset. FPP designs should aspire to the fuel replacement times achieved in the fission energy industry, the basis of the 0.1 year choice for the FPP base case value for the replacement time (Table \ref{tab:Default_parameters}). 

With respect to S replacement costs, the first consideration is the inequality requirement stated in Equation \ref{eq:A_requirement}, which defines the boundary of basic economic viability. Inserting the parameters into the formulas and setting target costs to zero, we find:
\begin{equation}\label{eq:areal_cost_inequality}
    (M\$/S)_S \leq 8.76\times10^{-3}POE_{net}  \cdot \eta_E \cdot X_S \hspace{3mm}.
\end{equation}\vspace{1mm}
This evaluates to $\mathrm{(M\$/S)_S \leq 1.1 \enspace M\$/m^2}$, using the base case parameters of Table \ref{tab:Default_parameters}. Note this is a necessary but insufficient condition to reach economic gain since power density must be sufficiently high to hit economic breakeven.  Conversely, one can set the S replacement costs to zero $\mathrm{A_3}=0$ and evaluate the threshold for target costs:
\begin{equation}\label{eq:target_cost_inequality}
    (c/Y)_{target} \leq 2.8\times10^{-4}POE_{net}  \cdot \eta_E \hspace{3mm}.
\end{equation}\vspace{1mm}
This evaluates to $\mathrm{(c/Y)_{target} \leq 0.011 \enspace \$/MJ}$. As previously discussed, this highlights the idea that both S replacement and target costs should be viewed as consumables, and that basic economic viability can only be achieved if the consumable costs do not need exceed the value of the energy product.

There is little to no real-world experience with the S areal normalized cost $\mathrm{(M\$/S)_S}$. The ARIES MFE designs provide some cost ranges from its engineering designs, since the first wall and blanket cost are itemized (these correspond to S in those designs since they are regularly replaced in D-T fusion)  \cite{Overview_ARIES-AT,Overview_ARIES-RS,Overview_ARIES-CS}.  In the following the $\mathrm{(M\$/S)_S}$ are in units of $\mathrm{M\$/m^2}$, using inflation-adjusted 2025 dollars:  ARIES-RS=0.41,  ARIES-AT=0.32 and ARIES-CS=0.13.  The ARIES blanket costs can also be provided as a percentage of the overall FPP direct cost: ARIES-RS=3.4 \%,  ARIES-AT=4.2 \% and ARIES-CS=2.3 \%.  These are similar in magnitude to the 0.3 $\mathrm{M\$/m^2}$ and $\sim 3\%$ estimates for the base case values of Table \ref{tab:Default_parameters}.  However, these blanket designs have never actually been built, and so one must use this cost comparison with caution.

Another means to provide context is to compare these ranges of S costs to present commercial-scale products that have similar traits, such as structural integrity under loading and heat removal.  Engineered structural components are typically costed per mass ($\mathrm{M_t}$ in metric tons), i.e., the normalized cost would be $\mathrm{(M\$/M_t)}$ in millions of dollars per ton. Therefore, one must take into consideration the geometry of the S replacement, since the model S cost is $\mathrm{(M\$/S)_S}$ normalized to the surface area through which the primary fusion energy passes.  This is accomplished by assigning a depth $\mathrm{L_S}$ [m] behind S that must be replaced, thus defining the replacement volume $\mathrm{V_S=S\cdot L_S}$. Next, one must assign the volumetric fraction, $\mathrm{f_S}$, of $\mathrm{V_S}$ that are components that require replacement. For example, if neutrons are the primary cause of damage, this would be approximately the volumetric fraction of solid  components in $\mathrm{V_S}$.  This leads to a formulation for the S areal cost:
\begin{equation}\label{eq:S_areal_cost_geometry}
   (M\$/S)_S = (M\$/M_t) \cdot f_S \cdot L_S \cdot \rho_S \hspace{3mm},
\end{equation}\vspace{1mm}
where $\mathrm{\rho_S \enspace[t/m^3]}$ is the volume-averaged mass density of the replaced components. 

Equation \ref{eq:S_areal_cost_geometry} implies that simply from a geometric viewpoint, the FPP designer will be motivated to decrease areal cost by decreasing both $\mathrm{L_S}$ and $\mathrm{f_S}$.  However, $\mathrm{L_S}$ is highly constrained by multiple physical requirements: to capture the fusion product energy, to provide adequate radiation shielding, and in the case of D-T fusion, to produce tritium from neutron-lithium reactions. Decreasing $\mathrm{f_S}$ is feasible as long as the functionality of S is maintained (e.g., structural integrity or vacuum quality). A typical choice in decreasing $\mathrm{f_S}$ is to maximize the fraction of $\mathrm{V_S}$ using liquids, since they can be replaced by flow, and cannot suffer disordering degradation. Examples of this are the SiC blanket of ARIES-AT, which has a large fraction of PbLi liquid eutectic \cite{Overview_ARIES-AT}, the liquid immersion blanket using a FLiBe molten salt in ARC \cite{CreelySOFE_2025, rutherford2024manta} which has $\mathrm{f_S}\lesssim0.05$ and the flowing FLiBe blanket in the HYLIFE-2 IFE design \cite{moir1994hylife} where $\mathrm{f_S}\rightarrow0$.

As an example, we can estimate the fixed $\mathrm{L_S=1\enspace m}$, typical for D-T blankets, and a mass density of steel $\mathrm{\rho_S = 8 \enspace [t/m^3]}$ as a typical structural solid material to provide a maximum allowed S cost target,
\begin{equation}\label{eq:S_areal_cost_approx}
   (M\$/S)_S \approx 8 \enspace(M\$/M_t) \cdot f_S
   \leq 8.76\times10^{-3}POE_{net}  \cdot \eta_E \cdot X_S \hspace{3mm},
\end{equation}\vspace{1mm}
or
\begin{equation}\label{eq:S_areal_cost_approx_calc}
   (M\$/M_t) 
    \leq \frac{0.14}{f_S} \cdot 
    \left( \frac{POE_{net}}{100} \right) \cdot        
    \left( \frac{\eta_E}{0.4} \right) \cdot
    \left( \frac{X_S}{3.125} \right)    
    \hspace{3mm}, \vspace{4mm}
\end{equation}
using the base case values of Table \ref{tab:Default_parameters} and placed in a form where the sensitivity to the control parameters is evident. This highlights a design tension: the normalized cost threshold for S can be increased by improving its efficiency $\mathrm{\eta_E}$ and/or energy fluence $\mathrm{X_S}$, but those design choices must not increase the complexity / cost too much, or else the underlying economic viability will not be met.
Yet satisfying Equation \ref{eq:S_areal_cost_approx_calc}. which is equivalent to $\mathrm{(M\$/S)_S \leq 1.1 \enspace M\$/m^2}$ at base case values is only the basic economic requirement that $\mathrm{A \ge0}$. As revealed in the scoping studies, see Figure \ref{fig:Contour_fixedPfS_ArealcostS}, a more realistic limit for $\mathrm{Q_{econ} \ge 1}$ is S costing at or below the base case value of $\mathrm{(M\$/S)_S \leq 0.3 \enspace M\$/m^2}$, in which case the cost limit, with all parameters at base case value,  becomes
\begin{equation}\label{eq:S_areal_cost_approx_threshold}
   (M\$/M_t) 
    \leq \frac{0.038}{f_S}  
    \hspace{3mm}. \vspace{4mm}
\end{equation}
To gauge the challenge of meeting this S cost, and to find desirable/required $\mathrm{f_S}$, example costs from commercial or mass-produced products are examined. 

\begin{itemize}

    \item \textbf{Structural materials} such as high-strength reduced activation steels are candidates for D-T blankets, with material costs that vary between normalized costs of 0.015 (P92) and 0.075 (Eurofer) $\mathrm{(M\$/M_t)}$ \cite{zammuto2015long}, which ranges from $\sim 0.4 - 2 $ the cost limit of Equation \ref{eq:S_areal_cost_approx_threshold}. This  shows that specialized materials, even absent of assembly costs, could drive $\mathrm{(M\$/S)_S}$ to undesirable levels with high structural fractions $\mathrm{f_S}$.

    \item \textbf{Jet turbo engines} are complex engineered objects that operate in demanding conditions and require high reliability, yet are produced at significant scale for commercial airliners. The Rolls Royce Trent 900 series has a cost of $\mathrm{\sim 25 M\$}$  and a dry mass of 6.2 t \cite{RR_turbo_WIKI}, giving a normalized cost of $\mathrm{ \sim 3.8 \enspace M\$/M_t}$, about a hundred times the cost target. This high cost is due to the use of expensive, high-performance materials, complex engineering involving over $10^4$ components, extensive research and development, and quality control. This price point would likely be non-viable for a FPP, and it points to economic issues that would arise with overly complex designs.

    \item \textbf{Automobiles} are multi-component assembled commercial products. In the U.S., the average cost is 0.049 M\$ \cite{COX_KBB} and the average mass is 2 t \cite{Auto_insurance} giving a normalized cost of $\mathrm{ \sim 0.025 \enspace M\$/M_t}$ or $(\mathrm{M\$/S)_S\sim0.2\enspace f_S}$, which appears to be of the correct order of magnitude for FPP economic viability even with an $\mathrm{f_S \sim 1}$. This highlights the need for effective supply chain and assembly, given that automobiles are a highly mature mass-manufactured product with over a century of commercial experience and competition.

    \item \textbf{Nuclear submarines} are assembled products that are exposed to intense mechanical and nuclear environments. A Virginia-class submarine has a cost of $\mathrm{\sim 2800 M\$}$ and a mass of $\mathrm{\sim 7900 t}$, leading to a normalized cost of $\mathrm{ \sim 0.35 \enspace M\$/M_t}$. Using this costing implies that $\mathrm{f_S} \leq 0.1$ to meet the S costing goal.
      
\end{itemize}

This cursory examination of assembled products with similar traits to S and blanket designs illuminates the severe challenge to meeting the $\mathrm{(M\$/S)_S}$ cost targets if $\mathrm{f_S \sim 1}$. Or alternatively that FPP designers will need to seek robust designs with low $\mathrm{f_S}$ to control costs.

\subsection{Additional economic considerations}\label{sec:additional considerations}

The stylized model presented above offers a clear and quantitative framework for understanding the economics of an FPP.   However, from a financial perspective, a number of additional considerations can materially affect both the levelized cost of energy (LCOE) and the economic Lawson criterion $\mathrm{Q_{econ}}$.
The model's fidelity is clearly limited by its simplifying assumptions:  technical design parameters are known with perfect certainty such that fusion is a mature technology, and that both FPP operations and energy markets are in steady state. 
These abstractions yield considerable transparency and generality in the characterization of economic viability for FPPs of various designs, but have substantial implications for the cost of capital, stochastic variability in plant performance and revenues, working-capital requirements, and the influence of market design and policy mechanisms. 

Thus, while the model is most accurately applied to
an $n$th-of-a-kind (NOAK) fusion power plant under steady-state conditions, we are motivated to examine other considerations that can expand its applicability.
The following points outline key qualifications and potential extensions that would allow our framework to capture ever more realistic FPP scenarios for financing, operations, and investment under uncertainty.

\begin{enumerate}

\item \textbf{Risk-adjusted cost of capital.}
The framework treats construction and financing through a single real interest rate $\mathrm{i}$ in the amortization term $\mathrm{C_{fixed}}$ (Equations~\ref{eq:C_CF_cost_summary}, \ref{eq:C_CF_definition}). In practice, a fusion project's capital stack typically includes senior and junior debt, tax equity, and sponsor equity, each with distinct required returns and covenants. The appropriate discount rate is thus a risk-adjusted weighted average cost of capital (WACC). Operating risks (e.g., unplanned outages, market price risk for POE, and force majeure), technology risks (e.g., materials performance underlying $\mathrm{X_S}$ and heat removal limits setting $\mathrm{P_f/S}$), and policy risks (e.g., carbon prices or permits) all raise the required return on certain tranches of capital. Because $\mathrm{Q_{econ}}$ is monotonic in i (cf.~Figure~\ref{fig:Contour_simple_five}d, showing higher required $\mathrm{P_f/S}$ as i increases), a properly risk-adjusted WACC will generally lower $\mathrm{Q_{econ}}$ and shift the breakeven power density $\mathrm{(P_f/S)_{BE}}$ to the right. Conversely, credit enhancements such as government loan guarantees, investment tax credits, or contracts-for-difference effectively lower the WACC and improve economic viability without altering engineering parameters. A more sophisticated analysis would use multiple costs of capital, each corresponding to a different risk class, rather than a single WACC. The monotonicity of $\mathrm{Q_{econ}}$ in $\mathrm{i}$ and the
log-log-concavity established in the Supplement \cite{supplement} ensure that the
closest-viable-design projection responds smoothly to changes in
the cost of capital, with no risk of multiple local optima or
discontinuous jumps in the recommended parameter adjustments.

\item \textbf{More realistic, non-constant POE.}
The net price of energy ($\mathrm{POE_{net}}$) is assumed to be an inflation-adjusted constant, entering linearly into Equations~\ref{eq:Cgain_summary}, \ref{eq:A_definition}, and \ref{eq:LCOE_definition_Qecon}. In reality, realized revenue depends heavily on market conditions which vary outside of inflation---hourly prices, congestion, ancillary service demand, and contract structure. A fusion plant could earn multiple stacked revenue streams: grid energy, capacity payments, ancillary services, industrial heat, hydrogen, or data-center power purchase agreements (PPAs). Each product exhibits different statistical properties (e.g., mean, volatility, and correlation with outages). Because $\mathrm{LCOE_{eff} = POE_{net} /Q_{econ}}$ (Equation~\ref{eq:LCOE_definition_Qecon}), moving from full merchant exposure to long-term offtake contracts (fixed-price or floor) compresses the variance of POE and increases risk-adjusted $\mathrm{Q_{econ}}$ even if the expected price is unchanged. A small addition of a capacity payment or CfD floor can move $\mathrm{A_1}$ (Equation~\ref{eq:A_definition}) across the $\mathrm{Q_{econ} \ge1}$ threshold without any hardware change. This suggests important dynamics to FPP economic viability as a function of the stochastic properties of the $\mathrm{POE_{net}}$, especially due to drivers like geopolitical priorities, market sentiment, the impact of growing global energy demand by AI and data centers, and the shifting social norms regarding the urgency of climate change mitigation and clean energy sources.

\item \textbf{Stochastic uptime and replacement time.}
Utilization (Equation~\ref{eq:Utilization_definition_summary}) is modeled deterministically based on the ultimate durability of S, but in practice there will be a non-zero probability of random, unplanned outages due to component failure.  This probability will be larger in the FOAK FPP due to lower technical and operational maturity; one would also expect $\mathrm{\tau_{rep}}$ to have more uncertainty in a FOAK. Unplanned outages during high-price periods impose much larger financial penalties than those in off-peak seasons. Insurance, modular spares, and improved maintenance logistics can reduce not only the mean but also the variance of $\tau_{\text{rep}}$, which is crucial because $\mathrm{Q_{econ}}$ is nonlinear in U. Lenders typically underwrite conservative availability metrics (P95 or P99), so $\mathrm{Q_{econ}}$ should be evaluated under downside utilization scenarios. These facts motivated using the model in Monte Carlo simulations as illustrated Figure~\ref{fig:Cnet_distribution_example}. Structuring outage windows and maintenance carve-outs in PPAs can mitigate penalty exposure and improve debt serviceability.

\item \textbf{Consumables.}
The model treats target costs and S-replacement costs as ``consumables'' (Equations~\ref{eq:C_target_summary}--\ref{eq:C_S_Rep_summary}), and Figure~\ref{fig:Contour_simple_five}c,e highlights their nearly identical economic roles. Financially, these recurring cost streams should be hedged through long-term supply contracts or vertical integration. Because the iso-contours show diminishing gains beyond threshold $\mathrm{X_S}$ (Figures~\ref{fig:Contour_PfS_Xs}, \ref{fig:Contour_fixedPfS_Xs_yaxis}--\ref{fig:Contour_fixedPfS_Xs_xaxis}), overspending on exotic materials or designs to extend $\mathrm{X_S}$ may deteriorate NPV if it does not also reduce $\mathrm{\tau_{rep}}$ or improve $\mathrm{P_f/S}$. Make-or-buy decisions for S-fabrication must incorporate learning-curve effects and working-capital needs: maintaining higher inventories of prefabricated modules reduces downtime risk but ties up otherwise-liquid capital. Coordinated multi-plant procurement can exploit scale economies and reduce cost variance feeding into $\mathrm{C_{S,rep}}$.

\item \textbf{FOAK vs.\ NOAK.}
The model is most accurately applied to NOAK plants operating in steady state. However, mature FPPs do not exist yet and FOAK projects will be necessary, and will likely bear substantially higher costs due to technology development and first encounters with integrated operations. A complete financial roadmap must bridge FOAK risk through staged equity, milestone-based debt, and public credit support, with debt pricing that declines as technical risk is retired. FOAK valuation could be modeled as a sequence of real options---to continue, expand, or redesign---rather than a static NPV. Each option exercise updates the plant's position relative to the $\mathrm{Q_{econ}=1}$ hypersurface, gradually moving into the economically viable volume $\mathcal{V}$. An example of using the model for the assessment of payoff to technical advances is shown in Figure \ref{fig:tangent_payoff}.  This perspective avoids underinvestment in early de-risking stages that yield the greatest marginal gains in later NOAK economics.  

\item \textbf{Market design and policy levers.}
Figures~\ref{fig:Contour_PfS_taurep}- \ref{fig:Contour_PfS_Xs} show sharp engineering thresholds in $\mathrm{P_f/S}$, $\mathrm{X_S}$, and $\mathrm{\tau_{rep}}$, but policy instruments can shift economic outcomes without hardware changes. Contracts-for-Difference convert merchant volatility into fixed cashflows, effectively raising $\mathrm{A_1}$ (Equation~\ref{eq:A_definition}). Capacity markets remunerate availability and stabilize revenues linked to U. Carbon pricing and clean-energy credits raise effective $\mathrm{POE_{net}}$ for zero-carbon generation, while government guarantees directly reduce i. Each of these mechanisms acts as a ``financial control knob''—a policy-based variable that moves the design toward or across the $\mathrm{Q_{econ}\ge1}$  boundary without altering the physics or technology  of FPPs.

\item \textbf{Revenue/utilization covariance and scarcity premia.}
The deterministic model multiplies mean $\mathrm{POE_{net}}$ by mean U, but financial performance depends on the covariance between these two random variables. A dispatchable FPP can schedule outages during low-price seasons, while an unplanned S failure during scarcity events forfeits high rents. Portfolio operators can optimize outage timing against market forecasts and embed maintenance windows in contracts. Incorporating on-site thermal storage or auxiliary systems can maintain contractual firmness and capture scarcity premiums. These operational and contractual strategies increase risk-adjusted $\mathrm{Q_{econ}}$ beyond the deterministic baseline and reduce the left-tail risk highlighted by the Monte Carlo distributions in Figure~\ref{fig:Cnet_distribution_example}.

\item \textbf{Degradation uncertainty and the distributions of $\mathrm{P_f/S}$ and $\mathrm{X_S}$.}
Monte Carlo simulation results in Figure~\ref{fig:Cnet_distribution_example} show that variability in $\mathrm{P_f/S}$ and $\mathrm{X_S}$ leads to skewed and bimodal $\mathrm{Q_{econ}}$ distributions. Because lenders focus on conservative quantiles (e.g., P90 or P95 $\mathrm{Q_{econ}}$), engineering tolerances should consider minimizing variance in component performance in addition to, or instead of, maximizing mean output. Performance-linked EPC guarantees and availability wraps can transfer part of this tail risk to contractors, effectively reducing WACC and improving bankability even if the expected $\mathrm{Q_{econ}}$ remains constant.

\item \textbf{Working capital and runway risk.}
Although temporal equilibrium greatly simplifies our analysis, actual cashflows are time-dependent and likely nonstationary. Target procurement, S-fabrication, and replacement cycles generate working-capital swings and interest-during-construction (IDC) costs not captured in static $\mathrm{C_{fixed}}$. Liquidity shortfalls during ramp-up can trigger covenants even when lifetime $\mathrm{Q_{econ}}$ exceeds unity. Detailed project models should incorporate draw schedules, IDC, inventory cycles, and payment calendars. Aligning PPA payment profiles to maintenance cadences can reduce liquidity and runway risk more effectively than increasing revolving credit capacity.

\item \textbf{Plant end-of-life considerations.}
While $\mathrm{\tau_{life}}$ enters the financing term (Equations~\ref{eq:C_CF_cost_summary}, \ref{eq:C_CF_definition}), the model omits any consideration of residual value and decommissioning costs. Financial modeling should include reserves for end-of-life dismantling and waste handling as well as salvage value for reusable components (e.g., HTS magnets, turbines, vacuum systems, power electronics, etc.). The potential to repower with upgraded S designs or higher $\mathrm{\eta_E}$ cycles introduces a positive option value that offsets some terminal costs. Given the convex tradeoffs observed in Figures~\ref{fig:Contour_PfS_taurep}--\ref{fig:Contour_PfS_Xs}, designing for mid-life upgrades may be economically preferable to over-engineering initial specifications.

\item \textbf{Portfolio optimization.} The geometric representation of the viable region $\mathcal{V}$ in multi-dimensional parameter space naturally extends to a portfolio of plants or technologies. Imperfectly correlated risks---different $S$ technologies, regional $\mathrm{POE_{\rm net}}$ bases, or contract mixes, for example---can substantially improve portfolio-level $\Qecon$ (i.e., higher mean, lower volatility) for a given capital budget. The closest-viable-design optimization of Section~\ref{ssec:closestviable} offers a simple but effective capital-allocation heuristic: when the optimum is interior to the box constraints, parameters with large $\frac{|\partial \Qecon/\partial \theta_i|}{w_i}$ offer the greatest marginal improvement per unit weighted effort, and therefore are the most attractive targets for incremental R\&D or policy investment. When bounds are active, the associated lower- and upper-bound multipliers indicate which practical constraints are binding. This framework aids decision-making on whether to prioritize materials research, outage reduction, or financing innovation.

\item {\bf Steady-state vs. disequilibrium.} As with much of neoclassical economics, the assumption of temporal equilibrium in the FPP context is a powerful simplification that yields closed-form solutions and a complete characterization of economic viability. However, in practice, markets are often in disequilibrium, and frictions such as asymmetric information, incomplete or missing markets, misaligned incentives, behavioral biases, government policies, and taxes and other transactions costs can lead to market inefficiencies for extended periods of time. Moreover, by definition, the transition from FOAK to NOAK technologies implies non-stationarities in corresponding  costs, revenues, earnings, and other economic parameters. Extending the model to an adaptive framework \cite{lo:2017,lo:2024} may yield more realistic implications for short- and medium-run dynamics while preserving the long-run implications of our steady-state analysis. 

\end{enumerate}

\noindent
Taken together, these considerations underscore that the economics of fusion power cannot be fully characterized by engineering efficiency or capital cost alone. Financial structure, risk allocation, policy design, macroeconomic conditions and market integration jointly determine whether a technically viable fusion plant achieves commercial viability. Incorporating these extensions into the analytical framework would bridge the gap between the theoretical economic Lawson criterion and real-world investment decisions, providing a unified language for engineers, financiers, and policymakers to assess the path toward deployable, scalable fusion energy.

\section{Conclusion}
\label{sec:Conclusion}

The framework of physical assumptions used to derive the Lawson criterion for fusion energy gain can be fruitfully analogized to provide a similar set of criteria to estimate fusion economic gain for FPP designs. These criteria are agnostic with regards to technology and power output, and can be applied to any fusion confinement concept.

Sensitivity analysis conducted on the ten controlling parameters of the model generates a number of surprises. One consistent result is that there exists a threshold power density $\mathrm{P_f/S\sim 2 \enspace MW/m^2}$ for basic economic viability $\mathrm{Q_{econ}=1}$. This goes against the conventional wisdom that FPPs might operate profitably at low power densities that would require a less costly control surface S. Another surprise is the importance of low replacement costs and fast replacement times for the control surface S relative to its energy fluence limits. It is far better economically to have an FPP whose S can be quickly and cheaply replaced than it is to have an FPP with a maximally resilient S. The interplay between controlling parameters in the model usefully describes tradeoffs in the FPP design solution space, which can be analyzed symbolically or visualized graphically in scans between pairs of parameters.

The strengths of the Lawson approach are, paradoxically, also its weaknesses. The controlling parameters of the model, whether scientific, engineering, or economic, operate at a very high level of abstraction, and therefore they cannot give detailed prescriptions about FPP engineering design or financing the construction of an FPP. However, these parameters are more than robust enough to define a solution space for a given target, or to show that a given target is not economically feasible. It is also troublesome that some of the most important parameters are not accurately known and must be adduced from distant examples. Nevertheless, plausible numbers for the controlling parameters lead to values comparable to those found in existing energy systems for given model outputs.

With these caveats in mind, this framework and its associated model are able to provide new insights into the design space of future FPPs independent of specific knowledge of their technologies. It confirms that low-cost financing will be necessary to the economic success of any new FPP, it highlights the importance of the replacement cost and frequency of the control surface of the fusion reaction, and it overturns the idea that a very low power density will allow a FPP to become economically viable. We hope that the simplicity, flexibility, and transparency of this model will make it a staple in the fusion development space, like the Lawson criterion before it.

%***************************************************
% END OF REGULAR ARTICLE
%****************************************************

%\backmatter
\newpage

\section*{Supplementary information}

Illustrative examples and a fusion economics ``calculator'' that allows users to check the economic viability of their own parametric specifications can be found at \url{https://andrewwlo.github.io/fusioneconomics/}.
Also, a technical supplement establishing the mathematical properties
of $Q_{\mathrm{econ}}$---including its M\"obius form, concavity
in the effective power loading, log-log-concavity in all 10
parameters, and the uniqueness of the closest-viable-design
projection---is provided as an accompanying document~\cite{supplement}.

\section*{Acknowledgments}

Financial support from Rutherford Energy Ventures, LP and Stone Mountain Capital is gratefully acknowledged. No funding bodies had any role in study design, data collection and analysis, decision to publish, or preparation of this manuscript.  No direct funding was received for this study. The authors were personally salaried by their institutions during the period of writing (though no specific salary was set aside or given for the writing of this manuscript). 

We thank Peter Hancock for many helpful discussions. We also thank Layla Araiinejad for earlier discussions. The views and opinions expressed in this article are those of the authors only and do not necessarily reflect the views and opinions of any institution or agency, any of their affiliates or employees, or any of the individuals or organizations acknowledged above.

\section*{Declarations}

All authors are affiliated with Rutherford Energy Ventures, LP, a consulting and investment advisory firm specializing in fusion energy.

\noindent

\newpage

\begin{appendices}

\section{Sample derivation of energy conversion efficiency $\eta_E$}\label{sec:etae_appendix}

Consider a D-T FPP producing electricity as its energy product.  The basic definition of efficiency is
\begin{equation}\label{eq:eta_definition_electric}
    \eta_{E} \equiv \frac{P_{net}}{P_f} = \frac{P_{e,net}}{P_f} \hspace{3mm}, \vspace{2mm}
\end{equation}
where $\mathrm{P_{e,net}}$ denotes net electricity power (subscript ``e''), which can be defined by
\begin{equation}\label{eq:eta_definition_net}
    P_{e,net} =  P_{e,gross} -  P_{e,internal} \hspace{3mm}. \vspace{2mm} 
\end{equation}

The gross electric power $\mathrm{P_{e,gross}}$ is determined from
\begin{equation}\label{eq:Pe_gross_definition}
    P_{e,gross} =  P_f \enspace \eta_{th} \enspace M_n + (P_f + P_{ext}) \enspace f_{direct} \hspace{3mm}, \vspace{1mm}
\end{equation}
with $\mathrm{P_f}$ fusion power, $\mathrm{\eta_{th}}$ the blanket thermal-electric conversion efficiency, $\mathrm{M_n \equiv P_{th}/P_f}$ the effective multiplier setting total thermal power $\mathrm{P_{th}}$, including for example blanket nuclear reactions and $\mathrm{f_{direct}}$ is the fraction of fusion and externally coupled power directly converted to electricity .  The internal electric power consumption $\mathrm{P_{e,internal}}$ is determined from two primary requirements in an FPP to provide pumping power to circulate coolants $\mathrm{P_{pump}}$ to exhaust thermal power, and time-averaged external power delivered to the plasma required to maintain/obtain fusion conditions $\mathrm{P_{ext}}$.
\begin{equation}\label{eq:Pe_internal_definition}
    P_{e,internal} = P_{e,pump} + P_{e,ext} \hspace{3mm}, \vspace{2mm}
\end{equation}
with
\begin{equation}\label{eq:Pe_pump_definition}
    P_{e,pump} = \frac{f_p P_f M_n}{\eta_{e,pump}} \hspace{3mm}, \vspace{2mm}
\end{equation}
where $\mathrm{f_p \equiv P_{pump}/{P_f}} $ is the fractional pumping power requirement for the FPP, and $\mathrm{\eta_{e,pump}}$ is the wall-plug efficiency of the pumping, and the electrical power for external heating
\begin{equation}\label{eq:Pe_ext_definition}
    P_{e,ext} = \frac{P_{ext}}{\eta_{e,ext}} = 
    \frac{P_{f}}{Q_p \enspace\eta_{e,ext}} \hspace{3mm}, \vspace{2mm}
\end{equation}
where $\mathrm{\eta_{e,ext}}$ is the wall-plug efficiency of the external power and $\mathrm{Q_p \equiv P_f/P_{ext}}$ is the plasma energy gain from the Lawson criterion.

Pulling together these terms, and noting that all depend linearly on $P_f$ the FPP electrical energy efficiency can be obtained from
\begin{equation}\label{eq:etae_electric_derived}
    \eta_E = (\eta_{th} - \frac{f_p}{\eta_{e,pump}})M_n + (1+1/Q_p)f_{direct}- \frac{1}{\eta_{e,ext} \enspace Q_p} \hspace{3mm}. \vspace{2mm}
\end{equation}
This example derivation indicates that $\mathrm{\eta_E}$ can depend on a variety of plasma physics ($\mathrm{Q_p}$), nuclear physics ($\mathrm{M_n}$), thermodynamic ($\mathrm{\eta_{th}}$), electromagnetic ($\mathrm{f_{direct}}$) and other engineering parameters.

\section{Case study: Estimating $\mathrm{X_S}$ for D-T and D-D neutrons }\label{sec:Xs_appendix}

For D-T fusion, where 80 \% of the energy fluence arises from 14.1 MeV neutrons, we are interested in estimating energy fluence limit $\mathrm{X_S}$ in terms of $\mathrm{L_{dpa}}$ which expresses the ``displacement per atom'' (dpa) limit of S.  We use $\mathrm{L_{dpa}}$ since dpa is the most common metric for gauging nuclear component lifetime \cite{was2007fundamentals}  even though the lifetime can depend on other physical parameters such as transmutation effects. Eventually, all these effects (thermal cycles, displacements, transmutation) are linearly linked to energy throughput, making D-T a good case study on $\mathrm{X_S}$.

We define a parameter $\mathrm{F_{dpa,D-T}}$ $\mathrm{[dpa/(MW/m^2)\text -  y]}$ as a constant of proportionality relating the neutron energy fluence $\mathrm{(MW/m^2)\text -  y}$ to dpa.  Knowledge of  $\mathrm{F_{dpa}}$ then allows for the definition,
\begin{equation}\label{eq:X_S_DT}
X_{S,DT} \enspace  [MW \text -  y/m^2] =  
\frac{L_{dpa }}{0.8 \cdot F_{dpa,DT}} \hspace{3mm}, \vspace{2mm}
\end{equation}
where the 0.8 accounts for the fixed ratio of neutron to fusion power density in D-T. Thus knowledge of $\mathrm{F_{dpa}}$ provides the means to calculate/estimate $\mathrm{X_S}$ for a D-T FPP.

We start with a heuristic explanation of why values of $\mathrm{F_{dpa}}$ can be reasonably bounded for structural materials in S due to the nature of energy transfer. A fundamental requirement of S and the backing blanket is that it must slow down the neutrons, primarily by elastic collisions with the atoms of atomic mass A in S. For typical atomic number of the atoms in S ($\mathrm{A > 12}$) in structural materials the neutrons must undergo many collisions ($\mathrm{ \sim > 50}$) before thermalizing. Atomic displacements arise due to the kinematics of the energy transfer from neutrons to atoms exceeding a threshold displacement energy.  However, almost by default, candidate solid structural materials have significant atomic displacement energy thresholds ($\mathrm{E_d>  40 \enspace eV}$), and with absolute values over a limited range of values $\mathrm{\sim \times2}$. So in the end each neutron can be thought of as being required to induce a statistically similar number of displacements in S, since conceptually what is happening is that the neutron's kinetic energy is being transferred into heat (atoms movement) and displacement through atom self-collisions following the primary collision. Therefore the total number of displacements by each source neutron is $\mathrm{\sim E_{0,neutron}/E_d}$ with their distribution being set by neutron transport. However the physical density of the displacements is immaterial since dpa is displacement density \textit{normalized} to the material atomic density. 

A quantitative approximation of $\mathrm{F_{dpa}}$ can be obtained starting with the observations above. Each elastic collision on average fractionally decreases neutron energy by a kinematic factor KF, averaged over scattering angles, which to a close approximation is given by $\mathrm{KF \sim 2/A}$ for $\mathrm{A > 10}$. This will be appropriate since the atomic mass expected in most solid components ($\mathrm{(Z \geq 4-5)}$. The number of collision to slow the neutrons from their starting energy ($\mathrm{E_0=14.1 \enspace MeV}$) to some threshold lower energy ($\mathrm{E_{min}}$) is taken from,
\begin{equation}\label{eq:N_slowing}
\begin{split}
\frac{E_{min}}{E_0}=(1-KF)^{N_{slow}}  \rightarrow \\
N_{slow}=\frac{\ln({E_{min}/E_0})}{\ln(1-KF)} \simeq\frac{\ln({E_{min}/E_0})}{\ln(1-2/A)} \hspace{3mm}, \vspace{2mm}
\end{split}
\end{equation}
where $\mathrm{N_{slow}}$ is the statistical average of the number of slowing collisions between neutrons and target atoms in S. For simplification this assumes a homogeneous S composed of atoms of atomic mass A. The fluence of 1 MW-y of 14.1 MeV neutrons over 1 $\mathrm{m^2}$ is equivalent to a number fluence of
\begin{equation}\label{eq:neutron number fluence}
\begin{split}
& \Phi_{MW-y} \enspace   = \\
& \frac{10^6 [J/s/MW]\cdot 3.15\times 10^7 [s/y]}{1.6\times10^{-19}[J/eV]\cdot14.1\times10^6[eV/neutron]\cdot1m^2} \\
& =1.4 \times 10^{25} \enspace [neutrons/m^2/MW \text -  y]
\hspace{3mm}, \vspace{2mm}
\end{split}
\end{equation}
while the average number of displacements (``disp") caused by the neutron collision is 
\begin{equation}\label{eq:N_disp}
N_{disp} = \frac{0.8*KF\cdot \bar{E}_{n} \cdot N_{slow}}{E_{disp}} 
\hspace{3mm}, \vspace{2mm}
\end{equation}
where $\mathrm{\bar{E}_{n}}$ [eV] is the collision-averaged energy of neutrons, the 0.8 factor comes from the Kinchin-Pease model \cite{kinchin1955displacement}, and $\mathrm{E_{disp}}$ [eV] is the threshold displacement  energy of atoms in the material composing S. The density of displacements is then approximated by,
\begin{equation}\label{eq:n_disp}
n_{disp,MW-y} \enspace [disp/m^3]= \frac { \Phi_{MW-y} \cdot N_{disp} } {L_n} 
\hspace{3mm}, \vspace{2mm}
\end{equation}
where $\mathrm{L_n}$ [m] is the characteristic distance into S in which the collisions are occurring.  Treating the neutron transport as a random-walk process we can estimate
\begin{equation}\label{eq:Ln_estimate}
L_n\simeq(N_{slow})^{1/2} \enspace \lambda_n=(N_{slow})^{1/2} \enspace (\sigma \cdot n_S)^{-1} \hspace{3mm}, \vspace{2mm}
\end{equation}
where $\mathrm{\lambda_n}$ [m] is the mean-free path between collisions, $\mathrm{\sigma}$ [$\mathrm{m^2}$] is the microscopic cross-section for the neutron-atom collision and $\mathrm{n_S}$ [$\mathrm{m^{-3}}$] is the atom volumetric density of atoms in S. The displacements per atom is the normalized density of displacement, i.e., $\mathrm{n_{disp}/n_S}$, and gathering terms from Equations \ref{eq:N_slowing} - \ref{eq:Ln_estimate}, $\mathrm{F_{dpa}}$ for D-T can be estimated from,
\begin{equation}\label{eq:Fdpa_formula}
\begin{split}
& F_{dpa,DT} \simeq \frac{\Phi_{MW-y} \cdot N_{disp}}{L_n\cdot n_S}  \\
& \simeq 1.4 \times 10^{25}\cdot \sigma\cdot 
[\frac{\ln({E_{min}/E_0})}{\ln(1-KF)}]^{1/2}\cdot\frac{0.8 \cdot KF\cdot\bar{E_n}}{E_{disp}} \hspace{3mm}, \vspace{2mm}
\end{split}
\end{equation}
which is independent of $\mathrm{n_S}$ as previously stated. 

For an estimate of $\mathrm{F_{dpa,DT}}$, the microscopic cross-section is assigned to 1 barn or $\mathrm{10^{-28} \enspace m^2}$ typical for fast neutrons elastic collisions, the kinematic factor $\mathrm{KF \simeq 2/A}$.  The displacement energy will be calculated over a range 40 - 80 eV typical for candidate fusion structural materials and realizing that the expected accuracy of this simple model is $\mathrm{\sim 2}$.  The average neutron energy of the $\mathrm{N_{slow}}$ collisions is estimated from a logarithmic average between $\mathrm{E_0}$ to $\mathrm{E_{min}}$ as
\begin{equation}\label{eq:Ebar_formular}
\bar{E}_{n}\sim exp(0.55\cdot[\ln{E_0}+\ln{E_{min}}])\hspace{3mm}, \vspace{2mm}
\end{equation}
where the 0.55 is found to provide better agreement to numerical simulations than 0.5. 

A choice of $\mathrm{E_{min}}$ is required and two options are considered. The first option sets the minimum energy is fixed at $\mathrm{E_{min}=10 \enspace keV}$ with the reasoning that this assures that only fast neutrons are being considered, justified by the fact that we are most interested in the peak dpa rate which will be driven by energetic neutrons near the front of S. Choosing a fixed minium neutron energy is standard practice in considering effects of neutrons on materials \cite{zinkle2013challenges} (e.g., neutron degradation of high-T superconductors \cite{sorbom2016determination}). Using Equation \ref{eq:Ebar_formular} this fixes $\mathrm{\bar{E}_{n} = 1.6 \times10^6 \enspace eV}$ for evaluating Equation \ref{eq:Fdpa_formula}.
% \begin{equation}\label{eq:Fdpa_formula_10keV}
% \begin{split}
% & F_{dpa,DT}  \simeq 21\cdot 
% [\frac{-7.25}{\ln(1-2/A)}]^{1/2}\cdot(2/A) \\
% & \mathrm{for  \enspace E_{min}=10 \enspace keV}
% \end{split}
% \end{equation}\vspace{1mm}
A second option is to scale $\mathrm{E_{min} = M\cdot E_{disp}/KF }$  so that the minimum energy is a fixed multiplier M of the threshold energy for a neutron to cause a primary collision displacement. For this example we choose M=4 to assure this condition (noting $\mathrm{F_{dpa}}$ is approximately logarithmic sensitive to this choice of M).

% \begin{equation}\label{eq:Fdpa_formula_Emin}
% \begin{split}
% & F_{dpa,DT}  \simeq 2.8\times10^{-3}\cdot 
% [\frac{\ln{(2.5 \cdot A \cdot E_{disp}/E_0)}}{ \ln{(1-2/A )}}]^{1/2} \\
% &\cdot \frac{(2/A)\cdot exp (0.5\cdot[\ln{E_0}+\ln{(2.5\cdot cdot A \cdot E_{disp}}])}{E_{disp}} \\
% & \mathrm{for \enspace E_{min}=M \cdot E_{disp}/KF }
% \end{split}
% \end{equation}\vspace{1mm}

These example formulas are shown graphically in Figure \ref{fig:Fdpa}. Quantitatively these simple formulas provide a good match to example neutronics calculation, despite the significant simplifications invoked in the formulations. Furthermore, the general trend is as noted from the qualitative description, that $\mathrm{F_{dpa}}$ is weakly sensitive to the composition / atomic number of S since the kinematic factor of individual collision tends to be washed out by the collective requirements of moderating the neutrons. Indeed it is more likely that the chosen displacement energy for the materials in S has a larger impact than A

\begin{figure}[h]
    \centering
    \includegraphics[width=0.75\textwidth]{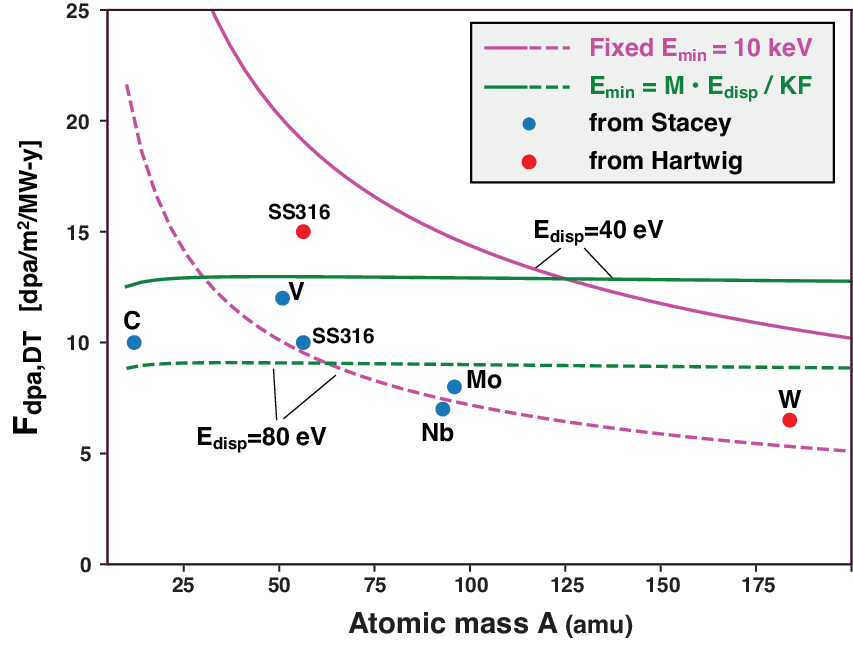}
    \caption{$\mathrm{F_{dpa}}$ as a function of atomic mass for a MW-y energy fluence of 14.1 MeV D-T neutrons .  Blue lines  with fixed $\mathrm{E_{min}=10 \enspace keV}$ and green lines with $\mathrm{E_{min} = 4\cdot E_{disp}/KF }$. Curves for displacement energies $\mathrm{E_{disp}}$ of 40 and 80 eV shown. Example results from neutronics are shown as circles with elements labeled from (green) Stacey \cite{stacey1985fusion} for single-composition blankets and (red) Hartwig \cite{HartwigPC} using OpenMC with thin first wall ($\mathrm{\sim cm}$) and elemental lithium as backing blanket material.}
    \label{fig:Fdpa}
\end{figure}

The conclusion is that for D-T fusion one may broadly expect a conversion factor $\mathrm{F_{dpa}\sim10}$ which then informs $\mathrm{X_S}$.  From Equation \ref{eq:X_S_DT} a $\mathrm{L_{dpa}=25}$  and $\mathrm{F_{dpa}\sim10}$ results in an estimated $\mathrm{X_S=3.125 \enspace MW \text -  y/m^2}$. While this is the base case value listed in Table  \ref{tab:Default_parameters} it is emphasized that the $\mathrm{X_S}$ model is in fact generic and not linked to only D-T fusion. Nonetheless this exercise in estimating $\mathrm{X_S}$ through dpa limits for D-T fusion neutrons provides further context for the  economic model, and firmly ties together the concept of limited S lifetime due to energy fluence.

The $\mathrm{X_S}$ framework  and formulas provided (Equations \ref{eq:N_slowing} through \ref{eq:Fdpa_formula} ) are easily  adapted for establishing operational period in other fusion fuels where neutrons remain the dominant degradation mechanism for S. If so, then one must determine then the relative fraction of exiting fusion power as neutrons and their energies specific to the fuel cycle / plasma design, and the production of displacement damage. This is a more complex optimization of the fusion reactivity (see for example the derivations of non D-T fusion fuel cycles for plasma gain in \cite{wurzel2022progress}), and must contain an accurate calculation of neutron-producing ``side'' reactions (e.g., D-D fusion in D-3He ,  proton-alpha fusion in p-11B \cite{moreau1977potentiality}), and their energies, even in systems where the principal fusion reaction is aneutronic. 

For illustration we can choose a simplified example of a D-D fuel cycle, where the tritium and 3He produced by the D-D reactions are assumed to be removed before burning. In this case a quarter of the fusion products produced are neutrons with energy $\mathrm{E_0=2.45}$ MeV, comprising 34 \% of the fusion energy and $\mathrm{\Phi_{MW-y}=8.1 \times10^{25}}$, providing
\begin{equation}\label{eq:Fdpa_formula_DD}
\begin{split}
& F_{dpa,DD}   \\
& \simeq 8.1 \times 10^{25}\cdot \sigma\cdot 
[\frac{\ln({E_{min}/E_0})}{\ln(1-KF)}]^{1/2}\cdot\frac{0.8 \cdot KF\cdot\bar{E_n}}{E_{disp}} \hspace{3mm}, \vspace{2mm}
\end{split}
\end{equation}
which evaluates to $\mathrm{F_{dpa,DD}\simeq16}$ (which roughly agrees with neutronic simulations) for A $\sim$50 with $\mathrm{E_{disp}=80}$ eV using the $\mathrm{E_{min} = 4\cdot E_{disp}/KF }$ assumption.  The resulting energy fluence limit is
\begin{equation}\label{eq:X_S_DD}
X_{S,DD} \enspace [MW \text -  y/m^2] = 
\frac{L_{dpa }}{0.34 \cdot F_{dpa,DT}}  \hspace{3mm}, \vspace{2mm}
\end{equation}
where the 0.34 reflects the fractional energy in neutrons.  Using again $\mathrm{L_{dpa}=25}$, evaluated at A $\sim$ 50 results in $\mathrm{X_S \simeq 4.6 \enspace MW \text -  y/m^2 }$, which is $\sim$ 50 \% higher than the D-T $\mathrm{X_S}$ case. This modest increase in $\mathrm{X_S}$, even though the neutron fraction is much lower in D-D,  is due to the fact $\mathrm{X_S}$ properly links the fluence limit to \emph{total} fusion power density passing through S, which is the primary concern in a commercial FPP.

Alternatively, in ``neutron-poor'' fusion cycles one may determine that other energy transmission mechanisms through S lead to lifetime limits faster than neutrons. This could be charged particle flux or high-energy photon flux in surfaces used for direct energy capture, or in electromagnetic components used for EM energy recovery, as discussed in the main text. Or the component damage may have compounded effects from the simultaneous direct energy capture and neutron flux. Regardless of these details, all fusion energy must pass through S, and one can construct constituent relations between that energy fluence and S lifetime. The extension of this economic framework to consider the effect of fusion fuel cycle on economics will be the subject of future work. However by examining Equation \ref{eq:Cgain_definition} and Equation \ref{eq:Cgain_full} it is apparent there are interactions between fusion power density and S lifetime, both which are strongly impacted by the choice of fusion fuel.

Of course for specific FPP S designs and fuel cycles would use accurate neutronics, sputtering, thermal, etc calculations to determine $\mathrm{X_S}$. However these simple examples illuminate the issue at hand, which is that fusion energy fluence through S will lead to finite operational lifetime, and determining this limit is a key aspect of understanding FPP economics.

\end{appendices}


\newpage

\bibliography{sn-bibliography}% common bib file

\newcommand{\beginsupplement}{%
    \setcounter{table}{0}
    \renewcommand{\thetable}{S\arabic{table}}%
    \setcounter{figure}{0}
    \renewcommand{\thefigure}{S\arabic{figure}}%
    \setcounter{section}{0}
    \renewcommand{\thesection}{S\arabic{section}}
    \setcounter{equation}{0}
    \renewcommand{\theequation}{S\arabic{equation}}
}

\appendix
\clearpage
\beginsupplement

\theoremstyle{plain}

\newtheorem{lemma}[theorem]{Lemma}
\newtheorem{corollary}[theorem]{Corollary}

\theoremstyle{definition}

\theoremstyle{remark}

\newcommand{\R}{\mathbb{R}}
\newcommand{\bell}{\boldsymbol{\ell}}

\newcommand{\Csrep}{C_{S,\mathrm{rep}}}
\newcommand{\Ccf}{C_{fixed}}
\newcommand{\pfs}{P_f\!/S}
\newcommand{\taur}{\tau_{\mathrm{rep}}}
\newcommand{\tauL}{\tau_{\mathrm{life}}}
\newcommand{\mSs}{(M\$/S)_S}
\newcommand{\mSf}{(M\$/S)_{\mathrm{FPP}}}
\newcommand{\LSE}{\mathrm{LSE}}

\section{Properties of $\mathrm{\Qecon}$: Concavity, Log-Log-Concavity,
  and Implications for Closest-Viable-Design Optimization\\
  {\textnormal {\normalsize A. Lo \\
  Laboratory for Financial Engineering, Massachusetts Institute of Technology, 
  Cambridge MA 02139 USA\\
  \\
  Supplement to\\ 
  \textit{Criteria for the economic viability of fusion power plants}\\
   D.G. Whyte, A. Lo, R. Bielajew, M. Hancock, R. Moeykens, G. Shaw} }
   }

\bigskip

% \title{%
%   Properties of $\mathrm{\Qecon}$: Concavity, Log-Log-Concavity,\\[4pt]
%   and Implications for Closest-Viable-Design Optimization}
% \author{%
%   A. Lo, 
%   \textit{Massachusetts Institute of Technology, Cambridge MA 02139}\\
%   \\
%   Supplement to\\ 
%   \textit{Criteria for the economic viability of fusion power plants}
%   \\D.G. Whyte, A. Lo, R. Bielajew, \\
%   M. Hancock, R. Moeykens, G. Shaw\\
%   preprint available on Arxiv}
% \date{April 2026}
%new

% \begin{document}
% \maketitle
% \tableofcontents
% \bigskip

\subsection{Setup and Definitions}\label{sec:setup}

We work in 10-dimensional control space with parameter vector
\[
\btheta = \bigl(\pfs,\;\taur,\;\tauL,\;\POE,\;X_S,\;\eta_E,\;(c/Y)_{\mathrm{target}},\;\mSs,\;\mSf,\;i\bigr),
\]
where all components are strictly positive. Throughout we adopt the shorthand
\[
p \equiv \pfs, \quad \tau \equiv \taur, \quad \tau_L \equiv \tauL, \quad
E \equiv \POE, \quad X \equiv X_S, \quad \eta \equiv \eta_E,
\]
\[
c_Y \equiv (c/Y)_{\mathrm{target}}, \quad m_S \equiv \mSs, \quad m_F \equiv \mSf.
\]

\begin{remark}[Notation conventions]\label{rem:notation}
We follow the manuscript's notation and units throughout. The interest rate $i$ is in \emph{percent} (nominal value $i=5$, not $0.05$); the factor $0.01$ appears explicitly in the annuity formula~\eqref{eq:Ccf}. The quantity $\mSf$ is the areal plant capital cost in M\$/m$^2$, distinct from the overnight cost $c_{O\!/N}$ in \$/W derived via $c_{O\!/N} = \mSf / [(\pfs)\,\eta_E]$.
\end{remark}

\begin{definition}[Constituent rates]\label{def:rates}
\begin{align}
\Cgain &= \alpha\, E\, \eta\, p\, U, &&\alpha = 8.76 \times 10^{-3}, \label{eq:Cgain}\\
\Ctarget &= \beta\, c_Y\, p\, U, &&\beta = 31.5, \label{eq:Ctarget}\\
\Csrep &= m_S \cdot \frac{p}{X} \cdot U, \label{eq:Csrep}\\
\Ccf &= m_F \cdot \varphi(i,\tau_L), \qquad
\varphi(i,\tau_L) \equiv \frac{0.01\,i\,(1+0.01\,i)^{\tau_L}}{(1+0.01\,i)^{\tau_L} - 1}, \label{eq:Ccf}
\end{align}
where the utilization factor is
\begin{equation}\label{eq:U}
U = \frac{X}{X + p\,\tau}.
\end{equation}
\end{definition}

\begin{definition}[Economic gain factor]\label{def:Qecon}
\begin{equation}\label{eq:Qecon}
\Qecon(\btheta) = \frac{\Cgain}{\Ctarget + \Csrep + \Ccf}.
\end{equation}
\end{definition}

\subsection{Derivation of the M\"obius Form}\label{sec:mobius}

Define the \emph{effective power loading} $x \equiv p \cdot U = pX/(X + p\tau)$. In terms of~$x$:
\[
\Cgain = \alpha E\eta \cdot x, \qquad
\Ctarget = \beta c_Y \cdot x, \qquad
\Csrep = \frac{m_S}{X} \cdot x.
\]

\begin{proposition}[M\"obius representation]\label{prop:mobius}
With $a \equiv \alpha E\eta$, $b \equiv \beta c_Y + m_S/X$, $d \equiv m_F\varphi(i,\tau_L)$:
\begin{equation}\label{eq:Qmobius}
\boxed{\;\Qecon = \frac{a\,x}{b\,x + d},\;}
\end{equation}
a M\"obius transformation in $x$, with $a,b,d > 0$.
\end{proposition}

\subsection{Concavity in $x$}\label{sec:concavity_x}

\begin{proposition}\label{prop:concavity_x}
For $a,b,d > 0$, $Q(x) = ax/(bx+d)$ is strictly increasing and strictly concave on $(0,\infty)$, with $dQ/dx = ad/(bx+d)^2 > 0$ and $d^2Q/dx^2 = -2abd/(bx+d)^3 < 0$.
\end{proposition}

\subsection{Non-Concavity in the Full Parameter Space}\label{sec:nonconcave}

\begin{proposition}\label{prop:nonconcavity}
$\Qecon(\btheta)$ is not concave on~$\R^{10}_{++}$.
\end{proposition}

\begin{proof}
Fix all parameters except $m_F$ at any strictly positive values (e.g.\ the nominal values of Table~2 with $c_Y = 10^{-4}$). Then $\Qecon = ax/(bx + \varphi m_F)$ with $ax, \varphi > 0$ constant, and $d^2\Qecon/dm_F^2 = 2ax\varphi^2/(bx+\varphi m_F)^3 > 0$. Since $\Qecon$ is strictly convex along this line, it cannot be concave on~$\R^{10}_{++}$.
\end{proof}

\subsection{Log-Log-Concavity}\label{sec:loglogconcavity}

\subsubsection{Definitions}

\begin{definition}\label{def:llc}
$f:\R^n_{++}\to\R_{++}$ is \emph{log-log-concave} if $\log f(e^{\bell})$ is concave in $\bell = (\log\theta_1,\ldots,\log\theta_n)$. Equivalently, for all $\btheta^{(1)}, \btheta^{(2)}$ and $\lambda \in [0,1]$:
\[
f\!\bigl(\btheta^{(1)\lambda} \odot \btheta^{(2)(1-\lambda)}\bigr)
  \;\geq\;
  f(\btheta^{(1)})^\lambda \cdot f(\btheta^{(2)})^{1-\lambda},
\]
where $\odot$ denotes componentwise geometric combination.
\end{definition}

\begin{lemma}[Monomial over posynomial]\label{lem:monposyn}
If $N = c_0\prod\theta_i^{\alpha_i}$ is a monomial ($c_0 > 0$) and $D = \sum_{k} c_k\prod\theta_i^{\beta_{ki}}$ is a posynomial ($c_k > 0$), then $N/D$ is log-log-concave.
\end{lemma}

\begin{proof}
$\log(N/D) = (\text{affine in }\bell) - \LSE(g_1(\bell),\ldots,g_K(\bell))$, where each $g_k$ is affine. Since $\LSE$ is convex, the result is concave.
\end{proof}

\subsubsection{Log-log-convexity of the annuity factor}

\begin{lemma}\label{lem:phi}
For $r,T > 0$, define $\varphi(r,T) = r(1+r)^T/[(1+r)^T-1]$. Then $\varphi$ is log-log-convex: $\log\varphi$ is convex in $(\log r, \log T)$.
\end{lemma}

\begin{proof}
Write $u = \log r$, $v = \log T$. Define
\[
s(u) = \log(1+e^u), \quad w(u,v) = e^v s(u), \quad g(w) = w - \log(e^w - 1).
\]
Then $\log\varphi = u + g(w(u,v))$. Since $u$ is linear in $(u,v)$, it suffices to show $g \circ w$ has a positive semidefinite Hessian. The derivatives of $g$ and $s$ are:
\[
g' = \frac{-1}{e^w-1},\quad g'' = \frac{e^w}{(e^w-1)^2},\quad
s' = \frac{r}{1+r},\quad s'' = \frac{r}{(1+r)^2}.
\]
The partial derivatives of $w$ are $w_u = Ts'$, $w_v = w$, $w_{uu} = Ts''$, $w_{vv} = w$, $w_{uv} = Ts'$. The Hessian of $g(w(u,v))$ is $H_{ij} = g''w_i w_j + g'w_{ij}$, giving:
\begin{align*}
H_{11} &= g''T^2(s')^2 + g'Ts'', \\
H_{22} &= g''w^2 + g'w, \\
H_{12} &= g''Ts'w + g'Ts'.
\end{align*}

\noindent\textbf{$H_{22} \geq 0$:} Equivalent to $g''w + g' \geq 0$, i.e.\ $we^w/(e^w-1)^2 \geq 1/(e^w-1)$, i.e.\ $w \geq 1-e^{-w}$, which holds for $w > 0$.

\smallskip
\noindent\textbf{$H_{11} \geq 0$:} Factor as $Ts''(g''Tr + g')$ using $(s')^2/s'' = r$. The condition $g''Tr + g' \geq 0$ reduces to $Tr \geq 1-(1+r)^{-T}$. By convexity of $t \mapsto (1+r)^{-t}$, we have $(1+r)^{-T} \geq 1 - T\log(1+r)$, and since $\log(1+r) < r$ for $r > 0$, it follows that $(1+r)^{-T} > 1 - Tr$, hence $Tr > 1 - (1+r)^{-T}$.

\smallskip
\noindent\textbf{$\det H \geq 0$:} Expanding $H_{11}H_{22} - H_{12}^2$ and using $w = Ts$:
\begin{align*}
\det H &= (g''T^2s'^2 + g'Ts'')(g''w^2 + g'w) - (g''Ts'w + g'Ts')^2 \\
&= T^2(ss'' - s'^2)\,g'\,(g''w + g').
\end{align*}
The four factors have definite signs:
\begin{enumerate}[label=(\roman*)]
\item $T^2 > 0$.
\item $ss'' - s'^2 = r[\log(1+r) - r]/(1+r)^2 < 0$, since $\log(1+r) < r$ for $r > 0$.
\item $g' = -1/(e^w-1) < 0$.
\item $g''w + g' = [e^w(w-1)+1]/(e^w-1)^2 > 0$, since $h(w) = e^w(w-1)+1$ satisfies $h(0) = 0$ and $h'(w) = we^w > 0$.
\end{enumerate}
Therefore $\det H = (+)(-)(-)(+) > 0$.

Together, $H_{11} \geq 0$, $H_{22} \geq 0$, $\det H \geq 0$ establish positive semidefiniteness.
\end{proof}

\subsubsection{The full 10-parameter result}

\begin{theorem}[Log-log-concavity of $\Qecon$]\label{thm:llc}
$\Qecon$ is log-log-concave as a function of all 10 parameters on $\R^{10}_{++}$.
\end{theorem}

\begin{proof}
Multiplying numerator and denominator by $(X + p\tau)$:
\[
\Qecon = \frac{\alpha E\eta pX}{\beta c_Y pX + m_Sp + m_F\varphi X + m_F\varphi p\tau} = \frac{N}{D}.
\]
The numerator $N = \alpha E\eta pX$ is a monomial, so $\log N$ is affine in $\bell = \log\btheta$.

For the denominator, it suffices to show each summand is log-log-convex (since a positive sum of log-log-convex functions is log-log-convex):
\begin{enumerate}[label=(\roman*)]
\item $\beta c_Y pX$ and $m_Sp$: monomials, hence log-log-affine.
\item $m_F\varphi X$ and $m_F\varphi p\tau$: products of monomials with $\varphi$, which is log-log-convex by Lemma~\ref{lem:phi}. In log-coordinates, $\log(m_F\varphi X) = \ell_{m_F} + \ell_X + \log\varphi$; since $\log\varphi$ is convex in $(\ell_i, \ell_{\tau_L})$ and the remaining terms are linear, the sum is convex.
\end{enumerate}
Hence $\log D$ is convex in $\bell$, and $\log\Qecon = \log N - \log D$ is concave.

\smallskip
\noindent\textbf{Note:} With Lemma~\ref{lem:phi} now fully proved analytically (all three PSD conditions---$H_{11} \geq 0$, $H_{22} \geq 0$, $\det H \geq 0$---established via elementary inequalities), the 10-parameter log-log-concavity is a theorem-level result, not dependent on numerical verification.
\end{proof}

\begin{remark}[Non-strict concavity]\label{rem:nonstrict}
The concavity is not strict in general. Varying $E$ alone gives $\Qecon \propto E$, so $\log\Qecon$ is affine in $\log E$. The same holds for $\eta$. Strict concavity holds only along directions engaging the curvature of $\log D$.
\end{remark}

\begin{remark}[Boundary]\label{rem:boundary}
The result holds on $\R^{10}_{++}$. The manuscript's ranges include boundary values ($c_Y = 0$, $\tau_{\mathrm{rep}} = 0$, $m_S = 0$); any log-parameterized implementation needs strictly positive lower bounds (e.g.\ $c_Y \geq 10^{-6}$), which are numerically indistinguishable from zero.
\end{remark}

\subsubsection{Failure of ordinary quasiconcavity}\label{sec:not_qc}

\begin{proposition}\label{prop:notqc}
$\Qecon$ is not quasiconcave with respect to arithmetic combinations.
\end{proposition}

\begin{proof}
Take
\begin{align*}
\btheta^{(1)} &= (p, \tau, \tau_L, E, X, \eta, c_Y, m_S, m_F, i) = (2,\; 0.5,\; 30,\; 100,\; 2,\; 0.6,\; 0.002,\; 0.1,\; 3,\; 5), \\
\btheta^{(2)} &= (6,\; 0.05,\; 30,\; 200,\; 2,\; 0.6,\; 0.005,\; 0.1,\; 20,\; 5),
\end{align*}
with $\lambda = 0.75$. Then $\Qecon(\btheta^{(1)}) \approx 2.03$, $\Qecon(\btheta^{(2)}) \approx 2.30$, but $\Qecon(0.75\btheta^{(1)} + 0.25\btheta^{(2)}) \approx 1.71 < \min(2.03, 2.30)$.
\end{proof}

\subsection{Implications for Optimization}\label{sec:optimization}

\subsubsection{The manuscript's formulation}

The manuscript (Equation~52) poses the closest-viable-design problem using a diagonal weighted norm:
\begin{equation}\label{eq:opt_unweighted}
\min_{\theta \in \Theta_+} \sum_{i=1}^{10} w_i (\theta_i - \theta_{0,i})^2 \quad \text{subject to} \quad Q_{\mathrm{econ}}(\theta) \geq 1,
\end{equation}
with $W = \mathrm{diag}(w_1,\dots,w_{10})$, $w_i > 0$.  The diagonal structure matches the per-parameter normalization and difficulty decomposition described in Section~6.2.  When the active constraint is regular, the manuscript notes (Equation~53) that the KKT stationarity condition gives
\[
2W(\theta^* - \theta_0) = \lambda^* \nabla Q_{\mathrm{econ}}(\theta^*), \qquad Q_{\mathrm{econ}}(\theta^*) = 1,
\]
with scalar $\lambda^* \geq 0$.  Equivalently, the weighted displacement $W(\theta^* - \theta_0)$ is normal to the hypersurface $Q_{\mathrm{econ}} = 1$ at $\theta^*$.

\subsubsection{The normalization issue}

The weighted norm~\eqref{eq:opt_unweighted} addresses the fact that the 10 parameters carry heterogeneous units (MW/m$^2$, years, \%, \$/MW-h, M\$/m$^2$, etc.), so a unit change in one parameter is incommensurable with a unit change in another unless each weight is normalized in some manner. The weighted norm resolves this issue. Each weight $w_i$ can be decomposed conceptually as
\begin{equation}\label{eq:weight_decomp}
w_i = \frac{{\hat w}_i}{s_i^2},
\end{equation}
where $s_i$ is a scale factor that makes $(\Delta\theta_i / s_i)^2$ dimensionless (e.g.\ $s_i$ could be the width of the plausible range from Table~2) and $\hat{w}_i > 0$ is a dimensionless difficulty factor reflecting the relative cost or difficulty of changing parameter~$i$. In implementation the two roles are combined into a single $w_i$, but the decomposition clarifies that the norm is dimensionally consistent by construction and that the difficulty factors are modeling inputs.

\subsubsection{Role of $\Qecon$ versus the weights}

The formula for $\Qecon$ determines the \emph{viable region}---the set of parameter combinations achieving $\Qecon \geq q^*$. The weights do not change that region; they determine which point on its boundary is selected as ``closest.'' Different stakeholders (a plasma physicist, a financial engineer, a policymaker) would reasonably assign different $\hat{w}_i$ and obtain different projections, each representing a different pathway to viability.

Economically, the 10 parameters span fundamentally different categories: engineering parameters ($p$, $X$, $\eta$) require physics R\&D; market parameters ($E$, $i$) are set exogenously; construction parameters ($m_S$, $m_F$) depend on industrial-economic factors. A uniform weighting would conflate these categories---treating a 2-percentage-point reduction in the interest rate (which might require government credit guarantees) as equally ``costly'' as a 1\,MW/m$^2$ increase in power density (which requires advances in plasma confinement).

\subsubsection{Convexity structure}

By Theorem~\ref{thm:llc}, $\{\Qecon \geq q^*\}$ is convex in $\bell = \log\btheta$. Reparameterizing~\eqref{eq:opt_unweighted} (or its weighted generalization) with $\ell_i = \log\theta_i$ yields a convex feasible region: no disconnected pockets of viability, no re-entrant corners, and no spurious feasible regions.

\begin{proposition}[Uniqueness]\label{prop:wellposed}
Consider the weighted problem $\min_{\btheta \in \Theta} \|\btheta - \btheta_0\|_W^2$ with $W = \mathrm{diag}(w_1,\ldots,w_{10})$, $w_i > 0$, subject to $\Qecon(\btheta) \geq q^*$, reparameterized in log-coordinates. If the feasible set is nonempty and the box constraints satisfy $\theta_i^{\min} > \theta_{0,i}/2$ for all~$i$, then the problem has a unique global minimizer.
\end{proposition}

\begin{proof}
Each $h_i(\ell_i) = w_i(e^{\ell_i} - \theta_{0,i})^2$ has $h_i''(\ell_i) = 2w_i e^{\ell_i}(2e^{\ell_i} - \theta_{0,i}) > 0$ on $\ell_i > \log(\theta_{0,i}/2)$. The objective is strictly convex; the feasible set is convex (by log-log-concavity) and compact; uniqueness follows.
\end{proof}

\begin{remark}[Numerical verification]\label{rem:numerical}
SLSQP from 50 random initializations converges to the same solution (objective values identical to 8 significant figures), consistent with Proposition~\ref{prop:wellposed}.
\end{remark}

\begin{remark}[On the steep gradients]\label{rem:cliffs}
In log-coordinates, the level sets $\{\Qecon \geq q\}$ are convex. When the optimum is interior to the box constraints, the KKT stationarity condition reads $2W(\btheta - \btheta_0) = \lambda\nabla\Qecon(\btheta)$ with scalar $\lambda \geq 0$; when box bounds are active, additional bound multipliers enter. Parameters with large $|\partial\Qecon/\partial\theta_i|$ relative to $w_i$ are adjusted the most---the optimization preferentially changes parameters with the best cost-to-impact ratio.
\end{remark}

\begin{remark}[Endpoint, not trajectory]\label{rem:endpoint}
The projection identifies the closest viable \emph{endpoint}---the smallest weighted portfolio of parameter changes needed to reach viability. It is not intended to represent a literal R\&D trajectory through design space.
\end{remark}

\subsection{Summary}\label{sec:summary}

\begin{center}
\begin{tabular}{@{} l l l @{}}
\toprule
\textbf{Property} & \textbf{Domain} & \textbf{Result} \\
\midrule
M\"obius form & $x \in (0,\infty)$ & $Q = ax/(bx+d)$---Prop.~\ref{prop:mobius} \\[3pt]
Strict concavity in $x$ & $x \in (0,\infty)$ & $d^2Q/dx^2 < 0$---Prop.~\ref{prop:concavity_x} \\[3pt]
Concavity in $\btheta$ & $\R^{10}_{++}$ & \textbf{Fails}---Prop.~\ref{prop:nonconcavity} \\[3pt]
Log-log-concavity & $\R^{10}_{++}$ & \textbf{Holds} (not strict)---Thm.~\ref{thm:llc} \\[3pt]
Quasiconcavity (arithmetic) & $\R^{10}_{++}$ & \textbf{Fails}---Prop.~\ref{prop:notqc} \\[3pt]
Quasiconcavity (geometric) & $\R^{10}_{++}$ & \textbf{Holds} \\[3pt]
Unique optimizer & Nonempty feasible set, & \textbf{Holds}---Prop.~\ref{prop:wellposed} \\
 & $\theta_i^{\min} > \theta_{0,i}/2$ & \\
\bottomrule
\end{tabular}
\end{center}

\medskip

$\Qecon$ is log-log-concave in all 10 parameters: $\log\Qecon = (\text{affine}) - (\text{convex})$ in log-coordinates. The convexity of $\log D$ follows from the monomial-over-posynomial structure for 8 parameters (Lemma~\ref{lem:monposyn}) together with the analytically proved log-log-convexity of the annuity factor $\varphi$ (Lemma~\ref{lem:phi}). The concavity is not strict along pure-numerator directions ($E$, $\eta$), but uniqueness of the projection follows from the strict convexity of the distance objective (Proposition~\ref{prop:wellposed}).

\end{document}